%% file: jdyn.tex
\documentclass[11pt,a4paper]{article}

\RequirePackage[top=12mm,bottom=18mm,left=30mm,right=30mm,head=12mm,includeheadfoot]{geometry}

\usepackage[utf8]{inputenc}
\usepackage[T1]{fontenc}
\usepackage[english]{babel}
\usepackage{amsmath,amssymb,amsthm}
\usepackage{mathrsfs,dsfont} 
\usepackage[bibstyle=alphabetic, backend=biber, sorting=nyt, citestyle=alphabetic, giveninits, maxbibnames=10, url=false, isbn=false]{biblatex} 
\usepackage{mathtools, ifthen, xparse}
\usepackage[shortlabels]{enumitem}
\usepackage{csquotes}\MakeOuterQuote"
\usepackage[dvipsnames,table]{xcolor}
\usepackage{graphicx}
\usepackage{tikz,tikz-cd}
\usepackage{hyperref}
\hypersetup{
    pdftitle = {Convergence of Dynamics on Inductive Systems of Banach Spaces},
    pdfauthor = {van Luijk, Stottmeister, Werner}
    verbose=false,
    colorlinks,
    citecolor=MidnightBlue,
    linkcolor=MidnightBlue,
    urlcolor=OliveGreen
}
\usepackage[capitalize]{cleveref}

\input{base.tex}
\def\1{{\mathbf1}}
\let\idty\1
\def\Nl{{\mathbb N}}\def\Cx{{\mathbb C}}
\DeclareMathOperator\sym{sym}
\def\flow{{\mathcal F}}
\renewcommand\AA{{\mathcal A}}

\renewcommand\limsup\varlimsup
\renewcommand\liminf\varliminf
\newcommand\dom[1]{D(#1)}

\newcommand\A{\mathcal A}
\newcommand\fA{\mathfrak A}
\renewcommand\L{\mathcal L}     
\renewcommand\D{{\mathcal D}}
\newcommand\cE{{\mathcal E}}
\newcommand\fF{{\mathfrak F}}
\newcommand\fS{{\mathfrak S}}
\newcommand\C{\mathscr{C}}      
\newcommand\nets{{\mathbf{N}}}  
\def\oo{\infty}

\let\d\relax
\def\blob{{\hspace{0.0pt}{\mathbin{\vcenter{\hbox{\scalebox{0.5}{$\bullet$}}}}}\hspace{+1pt}}}
\def\d{_{\blob{\hspace{-1pt}}}} 

\def\j#1#2{j_{#1#2}}
\def\jt{\ensuremath{\j{}{}}}   \def\jts{\ensuremath{\j{}{}^*}}   
\def\jlim{\jt\mkern-1mu\hbox{-}\mkern-3mu\lim}
\def\jslim{\jt\mkern-.5mu\hbox{$^*$}\mkern-3mu\hbox{-}\mkern-3mu\lim}
\def\ojt{{\ensuremath{\tilde\jmath}}}
\def\oj#1#2{\ojt_{#1#2}}
\def\ojlim{\ojt\mkern-0mu\hbox{-}\mkern-3mu\lim}
\DeclareMathOperator\LIM{\mathop{LIM}}

\def\wslim{w^*\mkern-3mu\hbox{-}\mkern-3mu\lim}

\newcommand\vep{\varepsilon}
\newcommand\h{_\hbar}
\newcommand\hp{_{\hbar'}}
\newcommand\hb{\ensuremath{\hbar}}
\newcommand\hbp{\ensuremath{\hbar'}}
\renewcommand\TC{{\mathfrak T}}
\newcommand\TT{{\mathcal T}}
\newcommand\LL{{\mathcal L}}
\newcommand\R{{\mathcal R}}
\newcommand\MM{{\mathbb M}}
\newcommand\T{{\mathbb T}}
\newcommand\CAR{\textup{CAR}}
\newcommand\sign{\textup{sign}}
\newcommand\Ad{\textup{Ad}}
\newcommand\fh{{\mathfrak h}}
\newcommand\ltwo{\ell^{2}}

\DeclarePairedDelimiterX\dual[2]{\langle}{\rangle}{#1 , #2}     

\let\seminorm\normiii

\title{Convergence of Dynamics \\ on Inductive Systems of Banach Spaces}
\author{Lauritz van Luijk, Alexander Stottmeister, Reinhard F. Werner}
\date{
\small
    Institut f\"ur Theoretische Physik, Leibniz Universit\"at Hannover, \\
    Appelstraße 2, 30167 Hannover, Germany
\\[2ex]\today
}

\begin{document}

\maketitle

\begin{abstract}
    \noindent
Many features of physical systems, both qualitative and quantitative, become sharply defined or tractable only in some limiting situation. Examples are phase transitions in the thermodynamic limit, the emergence of classical mechanics from quantum theory at large action, and continuum quantum field theory arising from renormalization group fixed points. It would seem that few methods can be useful in such diverse applications. However, we here present a flexible modeling tool for the limit of theories, soft inductive limits, constituting a generalization of inductive limits of Banach spaces. In this context, general criteria for the convergence of dynamics will be formulated, and these criteria will be shown to apply in the situations mentioned and more.
\end{abstract}

\tableofcontents

\section{Introduction}\label{sec:intro}

Many features of physical theories become really clear only in some limiting situation. Quite often, the limit is not just the limit of some parameters in a fixed framework, but the structure of the theory changes in the limit. For example, sharp phase transitions appear only in an infinite volume limit, quantum theory becomes classical in the limit $\hb\to0$, or a quantum field theory is defined after removing a length-scale cutoff $\vep\to0$. In such cases, one has to define carefully what the limit means, and it would seem that the required techniques differ significantly from case to case.
However, it turns out that for the limit of dynamical evolutions, there is a common technical core, an abstract limit theorem, hereafter called the \emph{evolution theorem}, that considerably simplifies and unifies the proofs of limit theorems in quite diverse settings, from the classical and mean field limits to the thermodynamic limit and renormalization-group limits. This abstract evolution theorem is the topic of the current paper and will be illustrated by salient examples.

The structural limits are essentially obtained via inductive-limit constructions for a sequence of Banach spaces that describe either states or observables of the approximating systems. Convergent sequences have an element in each of these spaces.
It will be convenient to generalize the notion of inductive limits further and tolerate norm-small deviations in a convergent sequence.  Hence the limits will be ``in norm'' in the sense that sequences whose norm difference goes to zero have the same limit. The resulting \emph{soft inductive limits} then generalize the completion construction.

The essential technical notion underlying soft inductive limits is a generalization of the concept of Cauchy sequences to inductive limits.
For an inductive system $(E,j)$ of Banach spaces $E_n$ with connecting map $\j nm$, we coin this notion \emph{\jt-convergence}.
It relies on the fact that the connecting maps $\j nm$ allow for a notion of distance for elements in different approximating spaces.
We construct the limit space $E_{\oo}$ as the quotient of the Banach space of convergent sequences with respect to null sequences, which is automatically complete. Thereby, we obtain an explicit description for all its elements, in contrast with the standard inductive limits where limit space arises as the completion of $\bigcup_n E_n$.

This is particularly beneficial in the context of the evolution theorem for dynamics on the limit space when the latter does not happen to be uniformly continuous, as this requires a discussion of unbounded operators and their domains.
The main realization behind the evolution theorem is that central notions of the semigroup theory on Banach spaces have a version for sequences in an inductive limit.
This splits the problem into showing that, on the one hand, the generators have a dense set of convergent sequences in their domain, allowing the definition of the generator in the limit space, and, on the other hand, the resolvent of this limit operator has dense range. For this second step, one only needs to work in the limit space, which is often easier than tracking properties of the approximating systems. Informally the evolution theorem can be stated as follows:
\begin{letterthm}[The evolution theorem, cf.~\cref{thm:evolution}] \label{thm:A}
Given a (soft) inductive system $(E,\jt)$ along with approximating dynamics $T_{n}(t)$ admitting generators $A_{n}$. Then we have the following equivalent characterizations of the dynamics on the limit space $E_{\oo}$:
    \begin{enumerate}[(1)]
     	\item
            The approximating dynamics $T_{n}(t)$ preserves \jt-convergence, and the resulting limit dynamics $T_\oo(t)$ is strongly continuous in $t$.
        \item
            The resolvents $R_n(\lambda)=(\lambda-A_n)^{-1}$ preserve \jt-convergence, and the resulting limit operators $R_\oo(\lambda)$ have dense range.
        \item
            There is a dense subspace of \jt-convergent sequences $\D$ such that $(\lambda-A_{n})\D$ is also a dense subspace of \jt-convergent sequences.
        \item
            The limit generator $A_\oo$ is well-defined and generates a strongly continuous dynamics $T_{\oo}(t)$.
    \end{enumerate}
We conclude that the limit dynamics $T_\oo(t)$ is a strongly continuous one-parameter semigroup with generator $A_\oo$. The latter's domain is obtained by acting with the limit resolvent $R_{\oo}(\lambda)$ on \jt-convergent sequences.
\end{letterthm}

The density mentioned in the third item is with respect to a natural seminorm topology on \jt-convergent nets.
Another important aspect is the observation that soft inductive limits are categorically well-behaved:
If the inductive system of Banach spaces carries additional structure which is respected by the dynamics, the same will hold in the limit.
As an example, consider an inductive system where the Banach spaces are C*-algebras, and the connecting maps are asymptotically multiplicative and completely positive so that the limit space is again a C*-algebra.
If the dynamics is completely positive at every scale and convergent then the limiting dynamics will also be completely positive.

Besides their relevance in physics, inductive limits are a major constructive tool in mathematics. In particular, in the theory of operator algebras, various interesting objects can be constructed from simple building blocks using inductive limits \cite{blackadar2006operator}. For example, restricting to matrix algebras as approximating objects leads to the class of AF algebras and hyperfinite factors.
To allow for greater flexibility, generalized inductive limits have been proposed \cite{blackadar1997generalized} emphasizing the importance of asymptotic concepts to define the limit object, thereby allowing for the construction of NF algebras. Only recently, it has been shown that the asymptotic properties of generalized inductive systems can be further relaxed, allowing for all separable nuclear C*-algebras to be characterized by an inductive limit construction \cite{courtney2023completely, courtney2023nuclearity}.
\\

Our paper is organized as follows:
In \cref{sec:strict}, we introduce and discuss the notion of \jt-convergence, which is a Cauchy-type criterion for sequences in an inductive system, and use it to construct the limit space.
This section deals with inductive systems of Banach spaces, i.e., systems where the transitivity relation $\j nm\j ml=\j nl$ holds exactly.
We show in \cref{sec:soft} that one can relax this criterion to an asymptotic version, which still allows for constructing the limit space using \jt-convergent sequences.
This relaxed transitivity is the defining property of soft inductive systems.
Instead of introducing the concepts of soft inductive systems and that of \jt-convergence at once, we split them into separate chapters for pedagogical reasons.
In \cref{sec:operations}, we analyze the convergence of sequences of operations between inductive systems leading to operations between the limit spaces.
This is necessary as we are ultimately interested in the convergence of dynamics on inductive systems, which is considered in \cref{sec:dynamics}.
Because of its central importance to our applications, we explicitly consider dynamics and \jt-convergence for (soft) inductive systems of C*-algebras with completely positive connecting maps in \cref{sec:cp}.
To illustrate our results, we discuss in detail the four examples which motivated us to analyze the convergence of dynamics abstractly in \cref{sec:ex} and illustrate how our abstract evolution theorem helps in concrete situations.
In \cref{sec:comparison}, we conclude by comparing our evolution theorem with a result on the convergence of dynamics in limits of Banach spaces by Kurtz, our notion of \jt-convergence with continuous fields of Banach spaces, and soft inductive limits of C*-algebras with generalized inductive limits of C*-algebras.
For the convenience of the reader, we discuss the convergence of implemented dynamics in GNS representations in \cref{sec:appendix_implementors} and interchangeability of  inductive limits and Lie-Trotter limits of convergent dynamics in \cref{sec:trotter}.
\\

\section{Convergent nets in inductive systems}\label{sec:strict}

An inductive system $(E,j)$ of Banach spaces over a directed set $(N,\le)$ is a collection $\set{E_n}_{n\in N}$ of Banach spaces together with connecting maps $\set{\j nm}_{n<m}$ which are linear contractions $\j nm : E_m \to E_n$, $\norm{\j nm}\leq 1$, whenever $n>m$, such that
\begin{equation}\label{eq:strict}
    \j nl = \j nm \circ \j ml, \quad n>m>l.
\end{equation}
For every inductive system, there is a limit space $E_\oo$ and a net of contractions $\j\oo n: E_n\to E_\oo$ such that $\j\oo m =  \j\oo n\circ \j nm$ whenever $n>m$.
If the $\j nm$ are isometric, it may be constructed by completing the union $\bigcup_n E_n$ with respect to its natural norm.
We will, however, discuss a different construction of the limit space in terms of convergent nets.
This construction has several advantages. For example, it offers a direct description of every element of $E_\oo$ (no completion is necessary).

A net in $(E,j)$ is a net $(x_n)_{n\in N}$ of elements $x_n\in E_n$. We often denote nets in $(E,j)$ by $x\d$.
We denote the space of uniformly bounded nets by $\nets(E,j)$ and equip it with the norm
\begin{equation}\label{eq:sup_norm}
    \norm{x\d}_\nets \coloneqq \sup_n \ \norm{x_n}.
\end{equation}
It is easy to see that $(\nets(E,j),\norm{}_\nets)$ is again a Banach space. In fact, it is nothing but the Banach space product $\Pi_{n\in N} E_n$.
We further equip $\nets(E,j)$ with the following seminorm
\begin{equation}\label{eq:seminorm}
    \seminorm{x\d} \coloneqq \limsup_n \norm{x_n},
\end{equation}
where $\limsup_n$ denotes the limit superior along the directed set $N$.

We define $\j nn$ as the identity map on $E_n$ for all $n\in N$, and we set $\j nm = 0$ whenever $n \ngeq m$.
Given some $m\in N$ and some $x_{m}\in E_{m}$, we define a uniformly bounded net $\j{\blob} m x_m$ in an obvious way, i.e., at index $n$ it is equal to $\j nm x_m$.
We refer to such nets as {\bf basic nets}.
We say that a net is {\bf \jt-convergent} if it can be approximated in seminorm by basic nets, and we denote the space of \jt-convergent nets by $\C(E,j)$, i.e., $x\d$ is \jt-convergent if
\begin{equation}\label{eq:convergent_nets}
    x\d \in \C(E,j) \coloneqq \Bar{ \set{y\d\in\nets(E,j) \given y\d \text{ is basic }}}^{\seminorm{}}.
\end{equation}
It is straightforward to check that $\C(E,j)$ is a vector space, and we equip it with the topology induced by the seminorm.
Note that this turns $\j\blob m$ into linear contraction from $E_m$ to $\C(E,j)$.

\begin{lem}\label{thm:j_convergence}
    A net $x\d\in \nets(E,j)$ is \jt-convergent if and only if
    \begin{equation}\label{eq:j_convergence}
        \lim_{n\gg m} \norm{x_n - \j nm x_m} = 0,
    \end{equation}
    where $\lim_{n\gg m} =  \lim_m \limsup_n$. If $x\d$ is \jt-convergent then the limit $\lim_n \norm{x_n}$ exists (hence is equal to $\seminorm{x\d}$).
\end{lem}

It is clear that basic sequences $x\d =  \j\blob m x_m$ always satisfy \eqref{eq:j_convergence} since the norm difference becomes zero for all $n\geq m$.
We can rewrite \eqref{eq:j_convergence} in terms of the seminorm as
\begin{equation}\label{eq:j_convergence_seminorm}
    \lim_m \seminorm{x\d -\j\blob m x_m} = 0.
\end{equation}
We will later see that there is another equivalent definition provided that the connecting maps are (asymptotically) isometric.

\begin{proof}
    \cref{eq:j_convergence_seminorm} already proves that nets satisfying \eqref{eq:j_convergence} can be approximated by basic nets.
    The converse follows from the triangle inequality:
    For $\eps>0$ pick $y_l\in E_l$ such that $\seminorm{x\d - y\d}  < \eps$ where $y\d= \j\blob l y_l$.
    \begin{align*}
        \lim_m \seminorm{x\d - \j\blob m x_m}
        &\le \lim_m \paren[\Big]{\seminorm{x\d - y\d}+\seminorm{y\d - \j\blob m y_m} + \seminorm{\j\blob m (x_m - y_m) }}  \\
        &\le \eps + 0 + \lim_m \seminorm{x_m - y_m } =  \eps.
    \end{align*}
    To see that the limit of the norms exists, note that $\liminf_n \norm{x_n} \ge \limsup_n \norm{x_n}$ follows from
    \[
        0 =  \lim_{n\gg m} \norm{x_n - \j nm x_m} \ge \lim_{n\gg m} (\norm{x_n} - \norm{x_m}) =  \limsup_n \norm{x_n} - \liminf_m \norm{x_m}. \qedhere
    \]
\end{proof}

We say that a net $x\d$ is a (\jt-){\bf null net}, if $\seminorm{x\d}= \lim_n \norm{x\d} = 0$ and we denote the subspace of null nets by $\C_0(E,j)$.
Since the zero net $0$ is basic and since $\seminorm{x\d-0}= \seminorm{x\d}=0$, all null nets are \jt-convergent.
Two nets $x\d$ and $y\d$ have vanishing seminorm distance if and only if $(x\d - y\d)$ is a null net.

Since the subspace $\C_0(E,j)$ is the preimage of $\set0$ under the seminorm, it is a closed subspace of $\C(E,j)$ and we can consider the quotient space
\begin{equation}\label{eq:limit_space}
    E_\infty \coloneqq \quotient{\C(E,j)}{\C_0(E,j)}.
\end{equation}
We will show that $E_\oo$ is the limit space of the inductive system $(E,j)$, and we will refer to the equivalence class of a \jt-convergent net $x\d$ as its (\jt-)limit.
The natural projection onto the quotient will be deoted by $\jlim: \C(E,j)\to E_\oo$ and we write
\begin{equation}\label{eq:jlimit_notation}
    \jlim_n x_n \coloneqq \jlim x\d \in E_\oo, \quad x\d \in \C(E,j).
\end{equation}
To keep notation concise, we often denote \jt-limit of a net $x\d$ simply by $x_\oo$.
The quotient structure induces a map $\j\oo m$ from $E_m$ to $E_\oo$ by assigning to $x_m$ the equivalence class of the basic sequence $\j\blob m x_m$, i.e.,
 \begin{equation}\label{eq:embedding_into_limit}
    \j\oo m x_m = \jlim_n \j nm x_m.
\end{equation}
Equivalently, $\j\oo m \coloneqq \jlim\circ \j\blob m$.
The seminorm induces a norm on $E_\oo$, namely
\begin{equation}\label{eq:limit_space_norm}
    \norm{x_\oo} = \seminorm{x\d}, \quad x\d\in\C(E,j).
\end{equation}
The next result states that $E_\oo$ is a Banach space and satisfies a property which
\begin{prop}
    $E_\oo$ is a Banach space with the norm defined in \eqref{eq:limit_space_norm}.
    The pair $(E_\oo,\j\oo\blob)$ is uniquely determined up to isometric isomorphism by the following universal property:

    Let $F$ be a Banach space and let $T_n:E_n \to F$ be a net of contractions such that $T_m =  T_n\j nm $ whenever $n>m$, then there is a contraction $T_\oo: E_\oo \to F$ such that $T_\oo \j\oo n =  T_n$ for all $n$.
\end{prop}

The universal property can even be strengthened, under the assumption on $(F,T\d)$ stated in the proposition, $T\d$ is guaranteed to map \jt-convergent nets to Cauchy nets in $F$ and the limit operator satisfies $T_\oo \jlim_n x_n =  \lim_n T_n x_n$.
We stress that completeness means that every element of the limit space arises as the limit of some \jt-convergent sequence.
This is not the case in the standard construction where controlling elements outside the union $\bigcup_n E_n$ is cumbersome. This control is, however, much needed for the discussion of unbounded operators on $E_\oo$.

\begin{proof}
    The proof of completeness of $E_\oo$ will be given under weaker assumptions in \cref{sec:soft}.
    That $T\d$ maps \jt-convergent nets to Cauchy nets in $F$ follows from the estimate $\norm{T_n x_n - \j nm T_mx_m}= \norm{T_n(x_n - \j nmx_m)}\le \norm{x_n -\j nm x_m}$.
    We obtain a well-defined linear contraction $T_\oo:E_\oo\to F$ through $T_\oo \jlim_nx_n=\lim_n T_nx_n$.
    That $\j\oo n T_n = T_\oo \j\oo n$ follows directly from $T_n\j nm = \j nm T_m$.
\end{proof}

In the case of isometric connecting maps, we have the following equivalent definition of \jt-convergence in terms of the maps $\j\oo n$:
\begin{lem}
    Assume that the $\j nm$ are isometric, then a net $x\d$ is \jt-convergent if and only if the net $\j\oo n x_n$ converges in $E_\oo$.
\end{lem}

In fact, we will prove in the next section that it suffices if the $\j nm$ are asymptotically isometric in the sense that
\begin{align*}
    \lim_{n\gg m} \paren[\Big]{\inf_{\norm{x_m}=1} \norm{\j nm x_m}} =1\,.
\end{align*}
\null

Before moving on, we discuss the notions we introduced above for the elementary example of a constant inductive system.
We will see that the notion of \jt-convergence becomes equivalent to being a Cauchy net.

\begin{exa}[Constant inductive systems]\label{exa:constant}
    Let $E$ be a Banach space, then the trivial inductive system $(E,\id)$ is obtained by setting $E_n = E$ and $\j nm = \id_E$, and we also set $N = \NN$ for simplicity.
    The space $\nets(E,\id)$ consists of uniformly bounded sequences in $E$.
    Basic sequences are just constant sequences, so a sequence $(x_n)$ is \jt-convergent if and only if it is approximated in seminorm by constant sequences, which is equivalent to being Cauchy in $E$, i.e., being \jt-convergent is equivalent to being a convergent sequence because $E$ is complete.
    In particular, we see that for a sequence $(x_n)$ condition \eqref{eq:j_convergence} with $\j nm = \id$ is equivalent to being a Cauchy sequence, which can also be seen directly from more elementary arguments.
    In standard sequence space notation, the three space $\nets(E,\id)$, $\C(E,\id)$ and $\C_0(E,\id)$ are are equal to $\ell^\oo(\NN;E)$, $c(\NN;E)$ and $c_0(\NN;E)$, respectively.

    Therefore, the construction of the limit space $E_\oo$ corresponds to the standard construction of considering first the space of Cauchy sequences and then taking the quotient with respect to null sequences, which unsurprisingly shows that $E_\oo \cong E$.
    In fact, this motivated the above construction for the general case.
\end{exa}

We collect some useful properties that we will repeatedly use later on. The proof will be given in the next section in the more general setting of soft inductive limits.

\begin{prop}\label{thm:density_thm}
    \begin{enumerate}[(1)]
        \item\label{it:convergence_seminorm}
            Let $x\d$ be \jt-convergent and let $(x\d\up\alpha)_\alpha$ be a net of \jt-convergent nets. The following are equivalent,
            \begin{enumerate}[(i)]
                \item $\seminorm{x\d - x\d\up\alpha}\to 0$ as $\alpha\to\oo$,
                \item $x_\oo\up\alpha \to x_\oo$ in $E_\oo$ as $\alpha\to\oo$,
                \item there are $\tilde x\d\up\alpha\in \C(E,j)$ such that $\tilde x_\oo\up\alpha = x_\oo\up\alpha$ for all $\alpha$ and such that
                    \[
                        \norm{x\d - x\d\up\alpha}_\nets \to 0\quad\text{as $\alpha\to\oo$}.
                    \]
            \end{enumerate}
        \item\label{it:convergence_of_embeddings}
            For every \jt-convergent net $x\d$, the net $(\j\oo n x_n)_{n\in N}\subset E_\oo$ converges to $x_\oo$.
        \item\label{it:subspace_density}
            A subspace $\D\subset \C(E,j)$ is seminorm dense if and only if $\D_\oo = \set{x_\oo \given x\d\in \D}$ is dense in $E_\oo$.
        \item\label{it:basics_startpoint}
            For any finite collection $x\up1\d,\ldots,x\up k\d$ of \jt-convergent sequences and any $\eps>0$, there are $m\in N$ and $x_m\up1,\ldots,x_m\up k\in E_k$ such that $\seminorm{x\d\up i-\j\blob m x_m\up i} <\eps$.
    \end{enumerate}
\end{prop}

We now turn to a brief discussion of the dual notion of weak* convergence for nets of continuous linear functionals associated with $\C(E,j)$ (cf.~\cref{sec:quantscale}). To this end, we denote the continuous duals of the Banach spaces $E_n$ by $E_n'$, and we denote the dual pairing between $E_n$ and $E_n'$ by $\dual\placeholder\placeholder$.
A net of continuous linear functionals $\varphi\d$ associated with $\C(E,j)$ is a collection $\set{\varphi_n}_{n\in N}$ with $\varphi_n\in E_n'$.

\begin{defin}\label{def:js_conv}
    A uniformly bounded net $\varphi\d$ of functionals $\varphi_n\in E_n'$ is {\bf \jts-convergent}, if for all $x\d\in \C(E,j)$ the limit $\lim_n \dual{x_n}{\varphi_n}$ exists.
    In this case $x_\oo \mapsto \lim_n \dual{x_n}{\varphi_n}$ for some $x\d$ with \jt-limit $x_\oo$ defines a bounded linear functional $\varphi_\oo$ on $E_\oo$ with $\norm{\varphi_\oo}\leq\sup_{n}\norm{\varphi_n}$, which we also denote by $\jslim_n \varphi_n$.
\end{defin}

We have the following easy consequences of the definition.

\begin{lem}\label{thm:js_conv}
    \begin{enumerate}[(1)]
        \item \label{it:projectively_consistent}
            Let $\varphi_\oo\in E_\oo'$ and set $\varphi_n =  \varphi_\oo\circ \j\oo m\in E_n'$. Then $\varphi\d$ is \jts-convergent and  $\jslim_n \varphi_n =  \varphi_\oo$.
        \item
            Let $\varphi\d$ be a uniformly bounded net of continuous linear functionals.
            Then $\varphi\d$ is \jts-convergent if and only if for all $m\in N$, the nets $(\varphi_n\circ \j nm)$ are $w^*$-convergent in the limit $n\to\oo$.
            In this case
            \begin{equation}\label{eq:js_limit}
                \wslim_n \big( \varphi_n\circ \j nm \big) = \big(\jslim_n \varphi_n\big) \circ \j \oo m.
            \end{equation}
    \end{enumerate}
\end{lem}

Nets of functionals that arise from a functional $\varphi_\oo$ on $E_\oo$ as in \ref{it:projectively_consistent} are called projectively consistent \cite{takeda1955inductive} and can equivalently be characterized as those nets that satisfy $\varphi_n\circ\j nm = \varphi_m$ for all $n>m$.

\begin{rem}
\label{rem:strictc}
Given a strict inductive system $(E, j)$, we may assume that each $\j nm$ is injective because if we had $x_{m}\in\ker\j nm$ for some $n\geq m$, the basic sequence $\j{\blob}m(x_{m})$ would be null and, therefore, we would have $\jlim_{n}\j nm(x_{m}) = 0$. This, in turn, would entail that we could restrict the inductive system to $\tilde{E}_{m} = E_{m}/\cup_{n\geq m}\ker\j nm$.
\end{rem}

\section{Soft inductive limits}\label{sec:soft}

The construction of the limit space and the theory of \jt-convergence works in a much more general setting if the assumptions of inductive systems are weakened to what we call soft inductive systems.
This weaker notion is characterized by relaxing the equality $\j nl =  \j nm \j ml$ to an asymptotic version.
Readers interested only in standard inductive systems can skip this section and may ignore the word ``soft'' in subsequent sections.
All results stated for inductive systems in this paper hold in this generalized setting if not explicitly said otherwise, and all proofs are already written to be valid in this setting.

\begin{defin}\label{def:soft}
    A {\bf soft inductive system} of Banach spaces over a directed set $(N,\leq)$ is a tuple $(E,j)$, where $E= \set{E_n}_{n\in N}$ is a family of Banach spaces and where $j= \set{\j nm}_{n\geq m}$ is a collection of linear contractions such that
    \begin{equation}\label{eq:soft_inductive_sys}
        \lim_{n\gg m} \norm{(\j nl  - \j nm \j ml) x_l} = 0\quad\forall l\in N,\  x_l\in E_l.
    \end{equation}
\end{defin}

Standard inductive systems of Banach spaces are soft inductive systems with the additional property that $\j nl =  \j nm \j ml$. These will be called {\bf strict} to emphasize that they are a special case.
Under these assumptions, the usual construction of the limit space as a completion of the union fails, and the limit space's universal property becomes meaningless.
In fact, the latter can be replaced by an asymptotic version much in the same way as \eqref{eq:soft_inductive_sys} is an asymptotic version of \eqref{eq:strict}.
The way of constructing the limit space via equivalence classes of \jt-convergent nets will, however, still be possible by the same arguments.
In fact, we will not have to change a single word in the proofs presented in \cref{sec:strict} as they directly apply to this generalized setting.

As in the case of strict inductive systems, we denote by $\nets(E,j)$ the Banach space of uniformly bounded nets with the sup-norm $\norm{}_\nets$.
Basic nets are nets of the form $x\d =  \j\blob m x_m$ for some $m$, $x_m\in E_m$, where we set $\j mm= \id_{E_m}$ and $\j nm= 0$ if $n \ngeq m$.
We now come to the importance of the assumption \eqref{eq:soft_inductive_sys}, which can be rewritten as $\lim_{n\gg m} \seminorm{ x\d - \j\blob m x_m}=0$ for all basic nets with the seminorm defined as in \eqref{eq:seminorm}.
This guarantees that the proof of \cref{thm:j_convergence} still works and we get:

\begin{lem}
    Let $(E,j)$ be a soft inductive system. Then a net $x\d\in\nets(E,j)$ can be approximated in seminorm by basic nets if and only if
    \begin{equation}\label{eq:j_convergence_2}
        \lim_{n\gg m} \norm{x_n - \j nm x_m} =  \lim_m \seminorm{x\d- \j\blob mx_m} = 0.
    \end{equation}
\end{lem}

As in the strict case, such nets will be called \jt-convergent, and the space of \jt-convergent nets is denoted $\C(E,j)$. An important class of \jt-convergent nets are the null nets which are the nets such that $\lim_n \norm{x_n}= 0$, and we denote the subspace of null nets again by $\C_0(E,j)$.
The limit space $E_\oo$ is defined as the quotient of $\C(E,j)$ by $\C_0(E,j)$ as before and equip it with the induced norm (see \cref{eq:limit_space_norm}).
The projection onto the quotient is denoted by $\jlim$ and we write $\jlim_n x_n =  \jlim x\d$ and the maps $\j\oo m : E_m \to E_\oo$ are defined by setting $\j\oo m x_m =  \jlim_n \j nm x_m$.

\begin{lem}
    Both $\C(E,j)$ and $\C_0(E,j)$ are closed with respect to the $\norm{}_\nets$-norm topology and the Banach space quotient is isomometrically isomorphic to $E_\oo$, i.e.,
    \begin{equation}\label{eq:quotients}
        \quotient{(\C(E,j),\norm{}_\nets)}{(\C_0(E,j),\norm{}_\nets)} \cong E_\oo.
    \end{equation}
    In particular, $E_\oo$ is a Banach space, and the norm is given by
    \begin{equation}\label{eq:quotient_norm}
        \norm{x_\oo} =  \seminorm{x\d} = \inf_{y\d \in \C_0(E,j)} \norm{x\d+y\d}_\nets,\quad x\d\in\C(E,j).
    \end{equation}
\end{lem}

\begin{proof}
    Suppose $x\d\up\alpha$ is a $\norm{}_\nets$-Cauchy sequence in $\C(E,j)$ and let $x\d$ be its limit in $\nets(E,j)$.
    Then $x\d\in\C(E,j)$ follows from
    \begin{align*}
        \lim_{n\gg m}\norm{x_n-\j nm x_m}
        &\le \lim_{n\gg m} \Big( \norm{x_n-x_n\up\alpha}_\nets + \norm{x_n\up\alpha - \j nm x_m\up\alpha} + \norm{\j nm(x_m - x_m\up\alpha)}\Big) \\
        &\le 2\norm{x\d - x\d\up\alpha} \xrightarrow{\alpha\to\oo} 0.
    \end{align*}
    Now suppose that $x\d\up\alpha$ are all null nets. Then $\lim_n \norm{x_n} \le \lim_n \big( \norm{x_n - x_n\up\alpha} + \norm{x_n\up\alpha}\big) \le \norm{x\d-x\d\up\alpha}\to 0$ and thus $x\d\in\C_0(E,j)$.

    To show the norm equality, let $x\d\in \C(E,j)$ and define a net $y\d\up m\in\C_0(E,j)$ by setting $y_n\up m = - x_n$ if $n< m$ and $y_n\up m = 0$ else.
    We have
    \[
        \norm{x_\oo}= \limsup_n \ \norm{x_n} =  \lim_m \ \sup_{n>m} \ \norm{x_n} = \lim_m \ \norm{x\d - y\d\up m}_\nets \ge \inf_{y\d\in\C_0(E,j)} \norm{x\d + y\d}
    \]
    and thus
    \[
        \norm{x_\oo} \ge \inf_{y\d\in\C_0(E,j)} \norm{x\d+y\d}_\nets \ge \inf_{y\d\in\C_0(E,j)} \seminorm{x\d+y\d} = \seminorm{x\d} = \norm{x_\oo}.
    \]
\end{proof}

\begin{lem}
    \cref{thm:density_thm} holds for soft inductive systems.
\end{lem}

\begin{proof}
    \ref{it:convergence_seminorm}: The equivalence of (i) and (ii) is obvious since $\norm{x_\oo-x_\oo\up\alpha}= \seminorm{x\d-x\d\up\alpha}$ and it is clear that (iii) implies (i).
    To show (i) implies (iii), we define the net $\tilde x\d\up\alpha$ as follows:
    \[
        \tilde x_n\up\alpha \coloneqq
        \case{ x_n - \seminorm{x\d-x\d\up\alpha}\, \frac{x_n-x_n\up\alpha}{\norm{x_n-x_n\up\alpha}} , & \text{ if } x_n\up\alpha \neq x_n\\ x_n, & \text{ if } x_n\up\alpha = x_n}\ .
    \]
    This definition clearly guarantees that $\tilde x_\oo\up\alpha = x_\oo\up\alpha$ and $\norm{x_n - x_n\up\alpha} \leq \seminorm{x\d-x\d\up\alpha}$ which goes to zero as $\alpha\to\oo$ but is independent of $n$.

    \ref{it:convergence_of_embeddings}: This follows immediately from \eqref{eq:j_convergence} because $\lim_m \norm{x_\oo - \j\oo m x_m} = \lim_{n\gg m} \norm{x_n - \j nm x_m} = 0$.

    \ref{it:subspace_density}: Observe that $\D_\oo =  \jlim \D$. Since the projection of a dense subspace onto a quotient is dense, this shows that density of $\D$ implies density of $\D_\oo$.
    For the converse, let $x\d\in\C(E,j)$ be given. Then, by the density of $\D_\oo$, there are $x\d\up\alpha$ so that $x_\oo\up\alpha \to x_\oo$.
    But this implies that $\seminorm{x\d - x\d\up\alpha} = \norm{x_\oo - x_\oo\up\alpha} \to 0$ and thus that density of $\D_\oo$ implies density of $\D$.

    \ref{it:basics_startpoint}: Pick basic sequences $\j\blob {m_i} x_{m_i}$ such that $\seminorm{x\up i -\j\blob {m_i}x_{m_i}} < \eps/2$ for all $i=1,\ldots,k$.
    Now pick $m$ large enough so that $\seminorm{\j \blob {m_i}x_{m_i}-\j\blob m \j m{m_i} x_{m_i}} < \eps/2$ for all $i=1,\ldots,k$.
    We set $x_m\up i\j m {m_i}x_{m_i}$, and the claim follows from the triangle inequality.
\end{proof}

The limit space of a soft inductive system satisfies the following universal property.

\begin{prop}\label{thm:soft_universal_property}
    Let $(E,j)$ be a soft inductive system of Banach spaces, and let $F$ be another Banach space.
    Let $T\d$ be a net of linear contractions $T_n:E_n\to F$ which maps \jt-convergent nets to Cauchy nets in $F$, then there is an operator $T_\oo:E_\oo \to F$ such that $T_\oo(\jlim_n x_n)=\lim_n T_n(x_n)$.
\end{prop}

This is a special case of a general result on the convergence of operations between inductive systems (see \cref{thm:jjconv}).
The importance of this property is that it uniquely determines the limit space $E_\oo$ and the maps $\j\oo \blob$:

\begin{prop}\label{thm:soft_uniqueness}
    Let $(E,j)$ be a soft inductive system of Banach spaces.
    Let $\tilde E_\oo$ be a Banach space and let $\oj\oo n :E_n \to \tilde E_\oo$ be linear contractions, such that $\lim_n\norm{(\oj\oo n\j nm -\oj\oo m)x_m}=0$ for all $x_m\in E_m$.
    Then $\ojlim:\C(E,j)\to \tilde E_\oo,\,x\d\mapsto \ojlim_n x_n = \lim_n \oj\oo nx_n$ is a linear contraction.
    If \cref{thm:soft_uniqueness} holds for $\tilde E_\oo$ and $\ojlim$, then there is an isometric isomorphism $\psi: E_\oo\to\tilde{E}_\oo$ such that $\psi\circ\j\oo n=\oj\oo n$.
\end{prop}

The isomorphism $\psi$ is simply given by $\jlim_n x_n \mapsto \ojlim_n x_n$.

\begin{proof}
    We claim that for any $x\d\in\C(E,j)$, the limit $\ojlim_nx_n\coloneqq \lim_n\oj\oo n x_n$ exists in $F$. This follows from
    \[
        \lim_{n\gg m} \norm{\oj\oo nx_n - \oj\oo m x_m} = \lim_{n\gg m} \norm{\oj\oo n(x_n - \j nm x_m)} + 0 \le \lim_{n\gg m} \norm{x_n-\j nm x_m}= 0.
    \]
    It is also clear, that for any two $x\d,y\d\in\C(E,j)$ with $\seminorm{x\d-y\d}=0$ we get $\ojlim_nx_n=\ojlim_ny_n$ (just check that $\ojlim_n z_n =0$ for all $z\in\C_0(E,j)$).

    Consider the nets $\j\oo n :E_n\to E_\oo$ and $\oj\oo n :E_n\to E_\oo$ which take \jt-convergent nets to Cauchy nets and apply the universal properties of $\tilde E_\oo$ and $E_\oo$, respectively.
    Therefore, there are contractions $\j\oo\oo  : \tilde E_\oo \to E_\oo$ and $\psi\coloneqq \oj\oo\oo:E_\oo \to\tilde E_\oo$, such that
    \[
        \j\oo\oo (\ojlim_n x_n) = \lim_n\, \j\oo n x_n = \jlim_n x_n
        \qandq \oj\oo\oo (\jlim_n x_n) = \ojlim_n x_n.
    \]
    Therefore $\psi$ is an isometric isomorphism between $E_\oo$ and $\tilde E_\oo$.
    It remains to be shown that $\oj\oo n$ is equal to $\j\oo n$ up to this isomorphism.
    This follows from
    \[
    \psi(\j\oo m x_m) = \psi \paren[\big]{\jlim_n \j nm x_m} =  \ojlim_n \j nm x_m
        = \lim_n \oj\oo n \j nm x_m = \oj\oo m x_m.\qedhere
    \]
\end{proof}

\begin{rem}
    Only the asymptotic properties of the connecting maps $\j nm$ and the spaces $E_n$ matter for the structure of the limit space (see also \cref{thm:equivalent_js}).
    One can even relax the assumption that all $E_n$ are Banach spaces and just assume a soft inductive system of normed spaces.
    In this case, the limit space will automatically be complete, i.e., a Banach space, and it agrees with the limit space of the system of Banach spaces that is obtained by completion of the normed spaces and continuous extension of the connecting maps.
\end{rem}

\begin{lem}
    Assume that the connecting contractions $\j nm$ are asymptotically isometric in the sense that
     \begin{equation}\label{eq:asymptotically_isometric}
         \lim_{n\gg m} \lambda_{nm } =1,\quad \lambda_{nm}=\inf_{\substack{x_m\in E_m\\\norm{x_m}=1}} \norm{\j nm x_m}.
    \end{equation}
    Then a net $x\d$ is \jt-convergent if and only if $\lim_n \j\oo nx_n$ exists (and one has $\jlim_n x_n =\lim_n \j\oo nx_n$).
\end{lem}

\begin{proof}
    The "only if" part holds for all soft inductive systems.
    For the converse, assume that the limit $\lim_n \j\oo nx_n$ exists.
    Now \eqref{eq:asymptotically_isometric} implies that $\lim_m \norm{x_m}=\lim_{n\gg m}\norm{\j nm x_m}$.
    We find that
    \begin{align*}
        \lim_{n\gg m}\norm{x_n - \j nmx_m}
        &\le \lim_{n\gg m} \norm{\j\oo n (x_n-\j nm x_m)} \\
        &\le \lim_{n\gg m} \norm{\j\oo n x_n- \j\oo mx_m} + \lim_{n\gg m}\norm{\j\oo m x_m - \j\oo n \j nm x_m} =0.
    \end{align*}
    The first term is zero since we assumed $\j\oo n x_n$ to be Cauchy, and the second one vanishes (already in the limit $n\to \oo$) because basic sequences satisfy \eqref{eq:j_convergence_2}.
\end{proof}

\begin{rem}
    If the directed set $N$ is countable and all Banach spaces $E_n$ are separable, then so is $E_\oo$.
    This can be seen by noting that a collection of a countable total subset $\set{x_i\up n\given i\in \NN}$ in each $E_n$ gives us a countable subset $\set{\j\blob n a_i\up i\given i\in\NN,\ n\in N}\subset\C(E,j)$ with seminorm-dense span.
\end{rem}

We are interested in the degree to which the notion of \jt-convergence and the limit space depend on the explicit choice of connecting maps $\j nm$.

\begin{prop}\label{thm:equivalent_js}
    Let $\set{E_n}_{n\in N}$ be a family of Banach spaces indexed by a directed set $N$ and let $\j nm, \oj nm:E_n\to E_m$ be two families of connecting maps so that $(E,j)$ and $(E,\ojt\,)$ both are soft inductive systems.
    The following are equivalent
    \begin{enumerate}[(1)]
        \item\label{eq:equivalence_of_connecting_maps}
            \jt-convergence is equivalent to \ojt-convergence, i.e.\ $\C(E,j)=\C(E,\ojt\,)$.
        \item\label{eq:equivalence_of_basics}
            For all $l$ and $x_l\in E_l$ one has
            \begin{equation}
                \lim_{n\gg m}\norm{(\j nm\oj ml-\j nl)x_l} = \lim_{n\gg m}\norm{(\oj nm\j ml - \oj nl)x_l} =0.
            \end{equation}
    \end{enumerate}
    In this case, the limit spaces of both inductive systems are isometrically isomorphic via the identification
    \begin{equation}\label{eq:equivalence_of_limit}
        \jlim_n x_n \longleftrightarrow \ojlim_n x_n.
    \end{equation}
    If one constructs the limit space as in \eqref{eq:limit_space}, then they are even equal (not just isomorphic).
\end{prop}

\begin{proof}
   \ref{eq:equivalence_of_connecting_maps} is equivalent to $\j nm$ and $\oj nm$ defining the same notion of convergence.
   This is, in turn, equivalent to \jt-convergence of all \ojt-basic nets and \ojt-convergence of all \jt-basic nets, which is precisely the condition \ref{eq:equivalence_of_basics}.

   The second part is clear: Both the \jt-convergent and null nets are the same for both systems so that the quotients $\C(E,j)/\C_0(E,j)$ and $\C(E,\ojt\,)/\C_0(E,\ojt\,)$ agree.
\end{proof}

Equipped with this, we briefly discuss the notion of a split inductive system.
This is an extra structure that singles out a \jt-convergent net for every point of the limit space.
This structure is present in several of the examples that we discuss in \cref{sec:ex} and does not trivialize in the case of strict systems.

\begin{defin} \label{def:split}
    A {\bf split} inductive system $(E,j,s)$ is soft inductive system $(E,j)$ together with a linear contraction $s\d : E_\oo \to\C(E,j)$ so that $\jlim_n s_n(y) = y$ for all  $y\in E_\oo$, i.e., $s\d$ is a right inverse of $\jlim : \C(E,j)\to E_\oo$.
\end{defin}

We have the following application of \cref{thm:equivalent_js}:

\begin{cor}\label{thm:split}
    Let $(E,j,s)$ be a split inductive system and define new connecting maps $\oj nm = \j\oo n \circ s_n$.
    Then $(E,\ojt\,)$ is a soft inductive system, and the notions of \jt-convergence and \ojt-convergence are equivalent (hence they define the same limit space).
    Furthermore, the following stronger version of \eqref{eq:soft_inductive_sys} holds for the $\ojt$-maps
    \begin{equation}\label{eq:split_soft}
        \lim_m \sup_n \,\norm{(\oj nl - \oj nm \oj ml)x_l} =0 \quad\forall l\in N, \, x_l\in E_l.
    \end{equation}
\end{cor}

A natural way to obtain a split inductive system is the following\footnote{This is the abstract version of the soft inductive system that we will use for the classical limit in \cref{sec:classical_limit}. Here the limit space is a space of functions on the classical phase space, and the maps $i$ and $p$ are suitable quantization and dequantization maps.}:
Suppose that we are given a net of Banach spaces $E_n$ a space $E_\oo$ and nets of contractions $i_n:E_n\to E_\oo$ and $p_n:E_\oo\to E_n$ so that $i_n\circ p_n \to \id_{E_\oo}$ strongly.
Then we obtain a soft inductive system by setting $\j nm = p_n\circ i_m$.
In fact, $(E,j,p)$ is a split inductive system, and by \cref{thm:split}, all split inductive systems are essentially of this form.

\section{Nets of operations on inductive systems}\label{sec:operations}

The universal property of the limit space of a strict inductive system can be regarded as a convergence result for certain operations.
But if we regard it as such, the assumption $T_n\j nm = \j nm T_m$ for all $n>m$ is unnecessarily restrictive.
We will now define convergence for nets of operators between two inductive systems $(E,j)$ and $(\tilde E,\ojt\,)$ but we will almost exclusively work with the cases where either both are the same system or one is constant (as in \cref{exa:constant}).
Whenever we consider two inductive systems, we assume they are defined w.r.t.\ the same directed set.
All definitions, statements, and proofs given in this section also apply to the broader class of soft inductive systems introduced in the previous section.

\begin{defin}
    Let $(E,j)$ and $(\tilde E,\ojt\,)$ be inductive systems and let $T\d$ be a uniformly bounded net of linear operators $T_n : E_n \to \tilde E_n$.
    We say that $T\d$ is {\bf \jt\ojt-convergent} if it maps \jt-convergent nets to \ojt-convergent nets, i.e., if for every $x\d\in\C(E,j)$ one has $T\d x\d\in \C(\tilde E,\ojt\,)$.
\end{defin}

In the case that $(\tilde E,\ojt\,)=(E,j)$, this \jt\jt-convergence means holds if and only if $T\d$ preserves \jt-convergence.
It follows that there always is a well-defined limit for such operations which is an operator between the limit spaces.

\begin{prop}\label{thm:jjconv}
    Let $(E,j)$ and $(\tilde E,\ojt\,)$ be inductive systems and let $T\d$ be an \jt\ojt-convergent net of operators.
    Then there is linear operator $T_\oo : E_\oo \to \tilde E_\oo$, such that
    \begin{equation}\label{eq:jj_limit}
        T_\oo \big(\jlim_n x_n \big) = \ojlim_n T_n x_n, \quad x\d\in \C(E,j).
    \end{equation}
    Its norm is bounded by $\norm{T_\oo}_{\L(E_\oo,\tilde E_\oo)} \le \limsup_n \norm{T_n}_{\L(E_n,\tilde E_n)}$.
\end{prop}

\begin{proof}
    Well-definedness follows from uniform boundedness of the net $T\d$: If $x\d$ and $y\d$ are \jt-convergent with the same limit, then $\norm{\jlim_n T_nx_n - \jlim_n T_n y_n} \le \lim_n \norm{T_n} \norm{x_n-y_n} = 0$.
    The norm bound is also immediate: $\norm{T_\oo \jlim_n x_n} = \lim_n \norm{T_n x_n} \le \limsup_n \norm{T_n} \norm{x_n} = \limsup_n \norm{T_n}\norm{\jlim_n x_n}$.
\end{proof}

\cref{thm:soft_universal_property}, which states that $(E_\oo,\j\oo\blob)$ satisfies the universal property of the limit space, follows as the special case where $(\tilde E,\ojt\,)=(F,\id_F)$ is a constant inductive system.
Another application is the notion of \jts-convergence discussed in \cref{def:js_conv} which is \jt\ojt-convergence if $(\tilde E,\ojt\,) =  (\CC,\id)$.
One can view \jt\ojt-convergence as a generalization of strong convergence:

\begin{exa}
    Let $(E,\id_E)$ and $(F,\id_F)$ be constant inductive systems. Then a uniformly bounded sequence of operators $T_n : E\to F$ is $\id_E\id_F$-convergent if and only if it is strongly convergent and the limit $T_\oo$ is the strong limit.
\end{exa}

Another special case of \jt\ojt-convergence is the convergence of operations $T_{n}:E_n \to E_{f(n)}$, which change the index within a fixed inductive system.
This allows for greater flexibility in describing operations on the limit space (see, for example, \cref{sec:thompson}).
If $f:N\to N$ is a monotone cofinal mapping, i.e.,  $n\le m\implies f(n)\le f(m)$ and $\lim_n f(n)=\oo$, then we obtain a soft inductive system by setting $\tilde E_n = E_{f(n)}$ and $\oj nm = \j{f(n)}{f(m)}$.
The limit space of this inductive system $(\tilde E,\ojt\,)$ is canonically isomorphic to $E_\oo$.
Thus a \jt\ojt-convergent net of operations defines an operation $T_\oo$ on $E_\oo$.

When it comes to proving \jt\ojt-convergence, we have the following criterion:

\begin{lem}
    Let $(E,j)$ and $(\tilde E,\ojt\,)$ be inductive systems and let $T\d$ be a uniformly bounded net of linear operators $T_n: E_n\to \tilde E_n$.
    \begin{enumerate}[(1)]
        \item\label{it:basics_suffice}
            If $T\d$ maps \jt-basic sequences to \ojt-convergent sequences, it is \jt\ojt-convergent.
        \item\label{it:jj_formula}
            $T\d$ is \jt\ojt-convergent if and only if
            \begin{equation}\label{eq:jj_conv}
                \lim_{n\gg m} \norm{(\oj nm T_m - T_n\j nm)x_m} = 0 \quad \forall x\d\in \C(E,j).
            \end{equation}
        \item\label{it:jj_composition}
            If $T\d$ is \jt\ojt-convergent and $S\d$ is a $\ojt\hat\jmath$-convergent net of uniformly bounded operators from $(\tilde E,\ojt\,)$ to an inductive system $(\hat E, \hat\jmath)$, then the composition $S\d T\d$ is $j\hat\jmath$-convergent and $(ST)_\oo = S_\oo T_\oo$.
    \end{enumerate}
\end{lem}

\begin{proof}
    \ref{it:jj_composition} is clear.
    \ref{it:basics_suffice}:
    We assume that $T\d$ maps basic sequences to \jt-convergent ones and set $M =\sup_n \allowbreak \norm{T_n}_{\L(E_n,\tilde E_n)}$. Let $x\d$ be \jt-convergent.
    For $\eps>0$, pick a basic net $y\d$ such that $\seminorm{x\d-y\d}<\eps$, then
    \[
        \lim_{m}\,\seminorm{T\d x\d - \j\blob m T_m x_m} \le \, M \seminorm{x\d - y\d} + \lim_m \seminorm{y\d-\j\blob m T_mx_m} + M \lim_m \norm{x_m-y_m} < 2M \eps.
    \]

    \ref{it:jj_formula}
    Assume $T\d$ to be \jt\ojt-convergent, then
    \[
        \lim_{n\gg m}\norm{(\oj nm T_m - T_n\j nm)x_m} \le
        \lim_{n\gg m} \paren[\Big]{\norm{T_nx_n - \oj nm T_m x_m} + \norm{T_n}\norm{x_n -\j nm x_m}} = 0.
    \]
    For the converse, we have
    \[
        \lim_{n\gg m}\,\norm{(T_nx_n - \oj nm T_m x_m} \le \lim_{n\gg m} \,\norm{T_n}\norm{x_n -\j nmx_m} + 0=0.
    \]
\end{proof}

The universal property of the soft inductive limit space follows from \cref{thm:jjconv} as the special case with a constant inductive system.

Let us assume that we have a, say, isometric action of a group $G$ on each Banach space $E_n$.
If the action preserves \jt-convergence (i.e.\ if $x\d\mapsto g\cdot x\d$ is \jt\jt-convergent), then the limit operations form an action of $G$ on $E_\oo$.
Actually, the same holds if $G$ is just a semigroup.
If the group action is strongly continuous with respect to some topology on $G$, then we would want the same to hold for the limiting action.
In the one-parameter case, we will reduce the problem to studying the convergence properties of the infinitesimal generators.
This requires generalizing \jt\jt-convergence to nets of unbounded operators:

\begin{defin}\label{def:unbounded}
    Let $(E,j)$ and $(\tilde E,\ojt\,)$ be inductive systems and let $A\d$ be a net of unbounded operators $A_n: \dom{A_n} \to \tilde E_n$.
    We define the {\bf net domain} of $A\d$ as
    \begin{equation}\label{eq:net_domain}
        \dom{A\d} \coloneqq \set{x\d \in \C(E,j) \given x_n \in\dom{A_n},\ A\d x\d \in \C(\tilde E,\ojt\,)\,}.
    \end{equation}
    We say that {\bf $A_\oo$ is  well-defined} if $\jlim_n A\d x\d$ is the same for all $x\d\in \dom{A\d}$ with the same \jt-limit.
    In this case we define the limit operator on $\dom{A_\oo}= \set{\jlim_n x_n \given x\d\in\dom{A\d}}$ by
    \begin{equation}\label{eq:limit_operator}
        A_\oo : \dom{A_\oo} \to \tilde E_\oo,\ A_\oo(\jlim_n x_n) \coloneqq \jlim_n A_nx_n.
    \end{equation}
\end{defin}

Observe that the well-definedness of $T_\oo$ is always guaranteed for a net of contractions $T\d$ and that $T\d$ is \jt\ojt-convergent if and only if $\dom{T\d}= \C(E,j)$.

\begin{lem}\label{thm:unb_op_con}
    Let $(E,j)$ be an inductive system and let $A\d$ be a net of operators $A_n: \dom{A_n}\to E_n$. Then:
    \begin{enumerate}[(1)]
        \item\label{it:well-def}
            $A_\oo$ is well-defined if and only if for all  $x\d\in \dom{A\d}$ with $\lim_n \norm{x_n}= 0$ we have $\lim_n \norm{A_nx_n}= 0$,
        \item\label{it:closed_dom} if all $A_n$ are closed operators, then $\Gamma_\oo =\set{\jlim_n x_n\oplus \jlim_n A_nx_n \given x\d\in\dom{A\d}}$ is a closed subspace of $E_\oo\oplus E_\oo$.
        \item\label{it:w_graph_con}
            Assume that there is a $w^*$-dense subset of the dual space $E_\oo^*$ which arises as limits of \jts-convergent sequences $\varphi\d$ with the properties that $\varphi_n\in\dom{A_n^*}$ and that $A\d^*\varphi\d$ is \jts-convergent. Then $A_\oo$ is well-defined.
    \end{enumerate}
    Now assume that $A_\oo$ is well-defined. Then
    \begin{enumerate}[resume*]
        \item\label{it:dens-def} $A_\oo$ is densely defined if and only if $\dom{A\d}$ is seminorm dense in $\C(E,j)$.
        \item\label{it:closed} if all $A_n$ are closable, then $A_\oo$ is a closed operator,
        \item\label{it:core} if for each $n$, $\D_n$ is a core for $A_n$, then every $x_\oo \in \dom{A_\oo}$ is the \jt-limit of an $x\d\in \dom{A\d}$ such that $x_n\in \D_n$ for all $n$.
    \end{enumerate}
\end{lem}

\begin{proof}
Denote the graph of $A_n$ by $\Gamma_n\subset E_n\oplus E_n$, and by $\norm{\placeholder}_{A_n}$ the restriction of the norm on $E_n\oplus E_n$ to $\Gamma_n$,  i.e.,\ $\norm{x_n}_{A_n} =  \norm{x_n}+\norm{A_n x_n}$.

    \ref{it:well-def}: This is immediate because the difference of two nets in $\dom{A\d}$ with the same limit is an element of $\dom{A\d}$ such that $\jlim_n x_n= 0$.

    \ref{it:closed_dom}: We will show that $\Gamma_\oo$ is complete by proving that every absolutely summable series converges in $\Gamma_\oo$ \cite[Thm.~III.3]{reed_simon_1} which will follow from the same property for all $\Gamma_n$.
    Let $x_\oo\up k\oplus y_\oo\up k \in \Gamma_\oo$ be an absolutely summable sequence, i.e., $\sum_k \norm{x_\oo\up k\oplus y_\oo\up k}_{E_\oo\oplus E_\oo}<\oo$.
    We can pick $x\d\up k \in\dom{A\d}$ so that $\jlim_n x_n\up k = x_\oo\up k$, $\jlim_n A_nx_n\up k =y_\oo\up k$ and $\norm{x_n\up k\oplus A_nx_n\up k}_{E_n\oplus E_n} = \norm{x_\oo\up k\oplus y_\oo\up k}_{E_\oo\oplus E_\oo}$.
    Then the series $\sum_k (x_n\up k\oplus A_n x_n\up k)$ is absolutely and hence converges to some $x_n\up\oo\oplus A_nx_n\up\oo\in\Gamma_n$.
    We claim that $x\d\up\oo\in\dom{A\d}$ and that $\sum_k (x_\oo\up k\oplus y_\oo\up k) = x_\oo\up\oo\oplus\jlim_n (A_nx_n\up\oo)$.
    This follows from the fact that these series approximate their limits uniformly in the index $n$.

    \ref{it:w_graph_con}: Let $x\d\in\dom{A\d}$ converges to zero, then
    \begin{align*}
        \norm{\jlim_n A_n x_n}
        &= \sup \lim_n \abs{\dual{A_nx_n}{\varphi_n}} \\
        &= \sup \lim_n \abs{\dual{x_n}{A_n^*(\varphi_n)}}
        = \sup\abs{\dual{x_\oo}{\jslim_n A_n(\varphi_n)}} =0
    \end{align*}
    where the supremum is over all nets $\varphi\d$ with  $\norm{\jslim_n\varphi_n}=1$ which satisfy the specified assumptions.

    \ref{it:dens-def}: This follows item \ref{it:subspace_density} of \cref{thm:density_thm}.

    \ref{it:closed}: Follows from item \ref{it:closed_dom}.

    \ref{it:core}: Let $\eps_n \searrow 0$ as $n\to 0$ and let $x\d\in \dom{A\d}$.
    Pick for each $n$ a $y_n\in\D_n$ such that $\norm{x_n - y_n}_{A_n} < \eps_n$.
    Then $\lim_{n\gg m}\norm{y_n - \j nm y_m} \le \lim_{n\gg m} (\eps_n+\eps_m+\norm{x_n -\j nm x_m}  ) =0$, i.e., $y \in \C(E,j)$ and $\norm{x_\oo - y_\oo} < \lim_n \eps_n =0$.
    One gets $A\d y\d\in \C(E,j)$ by a similar argument.
\end{proof}

\section{Dynamics on inductive systems}\label{sec:dynamics}

We say that $\set{T(t)}_{t\ge0}$ is a dynamical semigroup on a Banach space if it is a strongly continuous one-parameter semigroup such that each $T(t)$ is a contraction.
The generator of a dynamical semigroup is the closed dissipative operator
\begin{equation}\label{eq:generator}
    A x = \lim_{t\searrow0} \, \frac{T(t)x-x}t, \quad
    \dom{A} = \set*{x\in E \given t\mapsto T(t)x \text{ is in } C^1(\RR_+,E)\,}.
\end{equation}
The resolvents of the generator are the operators $R(\lambda) = (\lambda-A)^{-1}$, $\lambda\in\CC$ with $\Re\lambda>0$, they are an important tool in the theory of dynamical semigroups and enjoy many nice properties, e.g., they are analytic in $\lambda$ and one has $\dom{A}=\Ran R(\lambda)$.
For an introduction to dynamical semigroups, we refer the reader to \cite{engelnagel}.

Consider now a dynamical semigroup $T_n(t)$ on each Banach space $E_n$ of an inductive system $(E,j)$ and let $A_n$ be the net of generators.
We will need the notion of the net domain for nets of unbounded operators introduced in \cref{def:unbounded}.
The main result of this section is the following theorem which relates \jt-convergence preservation of the semigroups to a type of convergence of the infinitesimal generators and \jt-convergence preservation of the resolvents.
It can be regarded as a generalization of the Trotter-Kato approximation theorems, to which our result reduces in the case of a constant inductive system.

\begin{thm}\label{thm:evolution}
    Let $(E,\jt)$ be an inductive system and let $T\d(t)$ be a net of dynamical semigroups with $A\d$ the corresponding net of generators.
    Let $\lambda\in\CC$ be such that $\Re\lambda > 0$. The following are equivalent:
    \begin{enumerate}[(1)]
        \item\label{it:semigroups}
            $T\d(t)$ is \jt\jt-convergent and the limit operators $T_\oo(t)$ are strongly continuous in $t$.
        \item\label{it:resolvents}
            The net $R\d(\lambda)$ of resolvents $R_n(\lambda)=(\lambda-A_n)^{-1}$ is \jt\jt-convergent and the limit operation $R_\oo(\lambda)$ has dense range.
        \item\label{it:net_core}
            There is a seminorm dense subspace $\D\subset \dom{A\d}$ such that $(\lambda-A\d)\D$ is also seminorm dense.
        \item\label{it:generator}
            $A_\oo$ is well-defined, hence closed, and generates a dynamical semigroup.
    \end{enumerate}
    If these hold, the limit operations $T_\oo(t)$ form a strongly continuous one-parameter semigroup, and their generator is $A_\oo$.
    Furthermore, the net domain is given by $\dom{A\d}=R\d(\lambda)\C(E,j)$ and the limits $\R_\oo(\lambda)$ of the resolvents are the resolvents of $A_\oo$.
    In particular, these claims hold for some $\lambda$ if and only if they hold for all $\lambda$.
\end{thm}

If one only assumes that all $T\d(t)$ are \jt\jt-convergent, then it still follows that $T_\oo(t)$ is a one-parameter semigroup.
This is because $T\d(t)T\d(s) = T\d(t+s)$ holds as an equation of operators acting $\C(E,j)$.

Let us discuss some consequences of this theorem.
In the case where one already has a candidate for the semigroup on the limit space and is interested in showing that the dynamics converge, we have the following criterion:

\begin{cor}\label{thm:evolution_cor}
    Let $T\d$ be a net of dynamical semigroups and let $A\d$ be the net of generators. Let $S(t)$ be a dynamical semigroup on $E_\oo$ with generator $B$ and let $\D_\oo$ be a core for $B$.
    Suppose that there is a $\D\subset \dom{A\d}$ with $\jlim \D =\D_\oo$, such that $\jlim_n A_n x_n = B(\jlim_n x_n)$ for all $x\d\in\D$.
    Then $T\d(t)$ is \jt\jt-convergent, $A_\oo = B$ and $T_\oo(t) = S_\oo(t)$.
\end{cor}

\begin{proof}
    We check item \ref{it:net_core} of \cref{thm:evolution}. $\D$ is seminorm dense because $\D_\oo$ is dense in $E_\oo$ (see \cref{thm:density_thm}).
    Similarly $[(\lambda-A\d)\D]_\oo = (\lambda - B)\D_\oo$ is dense because $\D_\oo$ is a core for $B$.
\end{proof}

Before we come to the proof, we briefly discuss the analogous theorem for uniformly continuous semigroups.
I.e., for semigroups so that $T(t)$ depends continuously on $t$ in the operator norm topology of $\L(E)$.
For such semigroups, the generator is always bounded, and the semigroup is equal to the exponential series $T(t) = e^{tA}$.
For example, the convergence of this series and \jt-convergence preservation of $A\d$ implies \jt-convergence preservation of the semigroup under the right assumptions.

\begin{cor}
    Let $T\d$ be a net of uniformly continuous dynamical semigroups, and let $A\d$ be the net of generators. The following are equivalent
    \begin{enumerate}[(1)]
        \item $T\d(t)$ is \jt\jt-convergent for all $t\geq0$ and $T_\oo(t)$ is a uniformly continuous semigroup,
        \item $A\d$ is \jt\jt-convergent.
    \end{enumerate}
    In this case, it holds that $T_\oo(t) = e^{tA_\oo}$.
\end{cor}

\begin{proof}[Proof of \cref{thm:evolution}]
    For the proof, we introduce two additional statements
    \begin{enumerate}
        \item[(1$'$)]\label{it:semigroups'} For all $t\geq0$, $T\d(t)$ is \jt\jt-convergent and $\set{x\d\in\C(E,j) \given x_n\in\dom{A_n},\, \norm{A\d x\d}_\nets<\oo }$ is seminorm dense in  $\C(E,j)$.
        \item[(2$'$)]\label{it:resolvents_all_lambda} The resolvents $R\d(\mu)$ are \jt-convergent for all $\mu$ with $\Re\mu>0$ and the range of $R_\oo(\mu)$ is dense for all $\mu$.
    \end{enumerate}
    We will show the following implications
    \[
    \begin{tikzcd}
        (1) \arrow[r, Rightarrow]{} & (2') \arrow[r, Rightarrow]{} & (4) \\
        (1')\arrow[u, Rightarrow]{}\arrow[ur, Leftarrow]{} & (2) \arrow[u, Rightarrow]{} & (3)\arrow[u, Leftarrow]{} \arrow[l, Rightarrow]{}
    \end{tikzcd}
    \]

    \ref{it:semigroups} $\Rightarrow$ \hyperref[it:resolvents]{(2$'$)}:
    We use the integral formula for the resolvent $R_n(\lambda) = \int_0^\oo e^{-\lambda t}T_n(t)\, dt$ \cite[Ch.~II]{engelnagel} and get
    \[
        \seminorm{R\d(\lambda)x\d - \j\blob m R_m(\lambda)x_m} \le \int_0^\oo e^{-\Re\lambda t} \seminorm{T\d(t)x\d - \j\blob m T_m(t)x_m}\,s dt
    \]
    for \jt-convergent $x\d$ because of dominated convergence (for exchanging the limit in the definition of the seminorm with the integral).
    Applying dominated convergence again to the limit in $m$ and using that $T\d$ is \jt\jt-convergent shows that $R\d(\lambda)$ also is. This argument works for all $\lambda$.

    Similar to the above one can check that the integral formula and the resolvent formula remain valid for the limits, i.e., $R_\oo(\lambda)=\int_0^\oo e^{-\lambda t}T_\oo(t)\, dt$ and $R_\oo(\lambda)-R_\oo(\mu) =(\mu-\lambda)R_\oo(\lambda)R_\oo(\mu)$.
    The resolvent formula shows that the range of $R_\oo(\lambda)$ is independent of $\lambda$. We can now approximate any $x_\oo\in E_\oo$ by elements of the form $\lambda R_\oo(\lambda)x_\oo$ with $\lambda>0$:
    \[
        \norm{x_\oo-R_\oo(\lambda)} \le \int_0^\oo \lambda e^{-\lambda t} \norm{x_\oo -T_\oo(t)x_\oo}\, dt.
    \]

    The density of $\Ran R_\oo(\lambda)$ is equivalent to the density of the range of $R\d(\lambda)$ as an operator on $\C(E,j)$.
    From the resolvent equation $R\d(\lambda)-R\d(\mu) = (\mu-\lambda)R\d(\lambda)R\d(\mu)$ it follows that this range is independent of $\lambda$.
    We can now approximate any $x\d\in \C(E,j)$ by $\lambda R\d(\lambda)x\d$ with large $\lambda>0$:
    \begin{align*}
        \lim_{\lambda\to\oo}\seminorm{x\d-R\d(\lambda)}
        &\le \lim_{\lambda\to\oo}\lim_n \int_0^\oo \lambda e^{-\lambda t}\norm{x_n -T_n(t)x_n}\, dt \\
        &= \lim_{\lambda\to\oo}\int_0^\oo e^{-t} \norm{x_\oo-T_\oo(t/\lambda)x_\oo}\, dt
        =0
    \end{align*}
    where we used dominated convergence (twice) and strong continuity of $T_\oo(t)$.
    In particular, it also follows that $R_\oo$ satisfies the resolvent equation.

    \hyperref[it:resolvents]{(2$'$)} $\Rightarrow$ \hyperref[it:semigroups']{(1$'$)}:
    We use the following formula \cite[Ch.~X§1.2]{kato2013perturbation} valid for dynamical semigroups
    \[
        \norm*{T_n(t)x_n - \paren*{(t/k) R_n(t/k)}^k x_n } \le \frac{t^2}{2k} \norm{A_n^2x_n}, \quad x_n\in \dom{A_n^2},\, k\in\NN.
    \]
    By assumption $\paren*{(t/k) R_n(t/k)}^kx\d$ is \jt-convergent for all \jt-convergent $x\d$ and the right-hand side is uniformly bounded if $x\d\in R\d \dom{A\d}$.
    Therefore $T\d(t)$ preserves \jt-convergence of $x\d$ in $\D =R\d(\lambda) \dom{A\d}$. It remains to be shown that this subspace is seminorm dense.
    This will then automatically show density of the subspace $\set{x\d\in \C(E,j)\given x_n\in \dom{A_n},\, \norm{A\d x\d}_\nets<\oo}$.
    It is straightforward to see that $\dom{A\d}=R\d(\lambda)\C(E,j)$ and hence that $\D = R\d(\lambda)^2\C(E,j)$.
    Therefore the seminorm density of $\dom{A\d}$ is guaranteed by the assumption of the density of the range of $R_\oo(\lambda)$.
    Now let $x\d\in\C(E,j)$ and $\eps>0$. Pick a $y\d\in \dom{A\d}$ such that $\seminorm{x\d-y\d}<\frac\eps2$.
    Then we can pick a $\lambda>0$ such that $\seminorm{y\d-\lambda R\d(\lambda)y\d}<\frac\eps2$ with $z\d=\lambda R\d(\lambda)x\d\in \D$, we have $\norm{x\d-z\d} \le \norm{x\d-y\d}+\norm{y\d-z\d} < \eps$.
    This shows that $\D$ is dense.

    \hyperref[it:semigroups']{(1$'$)} $\Rightarrow$ \ref{it:semigroups}: By assumption, it suffices to show strong continuity on the dense subspace $\D_\oo$ of \jt-limits of nets in $\D=\set{x\d\in\C(E,j)\given x_n\in\dom{A_n},\,\norm{A\d x\d}_\nets<\oo}$.
    Let $x\d\in\D$.
    We use the formula $T_n(t)x_n = x_n + \int_0^t T_n(s)A_nx_n \,ds$ to get the estimate
    \[
        \norm{T_n(t)x_n-x_n} \le \int_0^t \norm{T_n(s) A_n x_n} \, ds \le t \norm{A\d x\d}_\nets\quad\forall n.
    \]
    Since the right-hand side goes to zero as $t\to0$, we obtain $\seminorm{T\d(t)x\d-x\d}=\norm{T_\oo(t)x_\oo-x_\oo}\to 0$ and hence strong continuity of $T_\oo(t)$.

    \hyperref[it:resolvents]{(2$'$)} $\Rightarrow$ \ref{it:generator}:
    It is straightforward that the limits $R_\oo(\lambda)$ satisfy the resolvent equation $R_\oo(\lambda)-R_\oo(\mu) = (\mu-\lambda)R_\oo(\lambda)R_\oo(\mu)$.
    This implies that the range of  $R_\oo(\lambda)$ is independent of $\lambda$ (and dense by assumption).
    Such families of operators are well-studied under and are usually called ``pseudoresolvents'' (see \cite[Ch.~III, 4.6]{engelnagel}).
    In fact, it follows that there is a closed operator $B:E_\oo\supset \dom{B}\to E_\oo$ with  $\dom B = \Ran R_\oo(\lambda)$ and $R_\oo(\lambda) =(\lambda-B)^{-1}$.
    We can use this to prove well-definedness of $A_\oo$ (in the sense of \cref{def:unbounded}): For $x\d\in\dom{A\d}$ with $\jlim_nx_n=0$, we have
    \begin{align*}
        \jlim_n A_nx_n
        &=\jlim_n (\lambda-A_n)x_n \\
        &= (\lambda-B) R_\oo(\lambda) \jlim_n (\lambda-A_n)x_n \\
        &= (\lambda-B) \jlim_n R_n(\lambda)(\lambda-A_n) x_n \\
        &= (\lambda-B)\jlim_nx_n =0.
    \end{align*}
    It is readily checked that we have $\dom{A\d}=R\d(\lambda)\C(E,j)$ and this implies that $\dom{A_\oo}=\Ran R_\oo(\lambda) =\dom B$.
    This can now be used to show that $A=B$, by considering
    \begin{align*}
        B R_\oo(\lambda)\jlim_n x_n
        &= (1+\lambda R_\oo(\lambda))\jlim_n x_n \\
        &=\jlim_n (x_n+\lambda R_n(\lambda)x_n) \\
        &= \jlim_n A_n R_n(\lambda)x_n \\
        &= A_\oo R_\oo(\lambda)\jlim_n x_n
    \end{align*}
    We can now use the Lumer-Phillips theorem \cite[Ch.~III, Thm.~3.15]{engelnagel} to prove that $A_\oo = B_\oo$ generates a dynamical semigroup.
    Since we already now that $\Ran(\lambda-A_\oo)^{-1}=\Ran R_\oo(\lambda)$ is dense, we only have to show that $A_\oo$ is dissipative, i.e., that $\norm{(\lambda-A_\oo)x_\oo}\ge \lambda \norm{x_\oo}$.
    This follows from dissipativity of all $A_n$ as for any $x\d\in\dom{A\d}$ we find $\norm{(\lambda-A_n)x_n} \ge \lambda \norm{x_n}$ and taking the limit  $n\to\oo$ shows the desired inequality.

    \ref{it:generator} $\Rightarrow$ \ref{it:net_core}: Put $\D = \dom{A\d}$ which is seminorm dense because $\D_\oo=D(A_\oo)$ is dense by assumption.
    Then $\D_\oo = \Ran R_\oo(\lambda)$ where $R_\oo(\lambda)$ is the resolvent of $A_\oo$.

    \ref{it:net_core} $\Rightarrow$ \ref{it:resolvents}:
    Let $x\d\in \C(E,j)$ be a net in the dense subspace $(\lambda-A\d)\D$ and let $y\d\in\D$ be so that $x\d=(\lambda-A\d)y\d$.
    Then $R\d(\lambda)x\d = y\d\in\C(E,j)$. Since the subspace we chose $x\d$ from is semi\-norm dense, this shows that $R\d(\lambda)$ is \jt\jt-convergent.
    Furthermore $R\d(\lambda)\C(E,j)$ is dense because it contains the dense subspace $R\d(\lambda-A\d)\D =\D$.

    \hyperref[it:resolvents]{(2)} $\Rightarrow$ \hyperref[it:resolvents]{(2$'$)}:
    Let $x\d$ be \jt-convergent and let $\mu$ be so that $\abs{\lambda-\mu}<\Re\lambda$. Then the series expansion
    \[
        R_n(\mu) x_n = R_n(\lambda)(1 + (\mu-\lambda)R_n(\lambda))^{-1}x_n = R_n(\lambda) \sum_{k=0}^\oo [(\mu-\lambda)R_n(\lambda)]^kx_n
    \]
    converges uniformly in $n$. Therefore $R\d(\mu)x\d$ is approximated by \jt-convergent nets in the $\norm{}_\nets$-norm and hence is itself \jt-convergent.
    By iterating the argument, we find that \jt-convergence holds for all $\mu$ in the right half-plane of $\CC$.
\end{proof}

In applications, there might be symmetries that are not present on the spaces $E_n$ but emerge in the limit.
An example is the translation symmetry of $L^2(\RR)$ viewed as the inductive limit of the spaces $L^2(I)$ where $I\subset \RR$ are finite-length intervals (ordered by inclusion with the obvious embeddings between the $L^2$-spaces).
Nevertheless, it can make sense to approximate the generator of such a symmetry by operators that do not generate anything (yet).

Recall that an operator $A$ is called \emph{dissipative}, if for all $\lambda>0$,
\begin{equation}\label{eq:dissipaive}
    \norm{(\lambda-A)y} \ge \lambda \norm y\quad \forall y\in\dom{A}.
\end{equation}
All generators of dynamical semigroups are dissipative.
In fact, a dissipative operator $A$ generates a dynamical semigroup if and only if $\Ran(\lambda - A)^{-1}$ is dense (this is the Lumer-Phillips theorem \cite[III, Thm.~3.15]{engelnagel}).

\begin{cor}
    Let $A\d$ be a net of dissipative closed operators.
    Let $\D\subset \dom{A\d}$ be seminorm dense in $\C(E,j)$ and assume that $(\lambda-A\d)\D$ is also seminorm dense.
    Then $A_\oo$ is well-defined, dissipative, and generates a dynamical semigroup.
    Furthermore, the resolvents $R\d(\lambda)$ are \jt\jt-convergent for all $\lambda$ with $\Re\lambda>0$ and $R_\oo(\lambda)=(\lambda-A_\oo)^{-1}$.
\end{cor}

\begin{proof}
    We can proceed as in the proof of \cref{thm:evolution}. The arguments used for the implications \ref{it:net_core} $\Rightarrow$ \ref{it:resolvents} $\Rightarrow$ \ref{it:generator} still work under our assumptions.
\end{proof}

It is not hard to see that \cref{thm:evolution} is stable under perturbations of the generators by \jt\jt-convergent nets of operators.
In fact, we can even allow for certain unbounded perturbations.
To state the conditions, we recall the notion of relative boundedness.
A linear operator $B:\dom{B} \to F$ on a Banach space $F$ is \emph{$A$-bounded} with respect to an operator $A:\dom{A}\to F$, if $\dom A\subset \dom B$ and if there are constants $a,b>0$ such that
\begin{equation}\label{eq:relatively_bounded}
    \norm{By} \le a \norm{Ay} + b\norm{x} \quad \forall y \in \dom{A}.
\end{equation}
The infimum $a_0$ over constants $a>0$ such that there is an $b>0$ for which \eqref{eq:relatively_bounded} holds, is called the $A$-bound of $B$.
The classical perturbation result that we will make use of is the following:
Let $T(t)$ be a dynamical semigroup with generator $A$ and let $B$ be a dissipative $A$-bounded operator with $A$-bound $a_0<1$, then $A+B$ generates a dynamical semigroup.

\begin{prop}
    Let $T\d(t)$ be a net of dynamical semigroups which is convergent in the sense of \cref{thm:evolution} and let $A\d$ be the net of generators.
    Let $B\d$ be a net of dissipative operators such that
    \begin{itemize}
        \item $\dom{B\d}\supset\dom{A\d}$. In particular, $\dom{B_n}\supset\dom{A_n}$ for all $n$,
        \item there are $a<1$ and $b>0$, such that
            \begin{equation}\label{eq:relative_boundedness_net}
                \norm{B_n x_n} \le a\norm{A_nx_n} + b\norm{x_n} \quad \forall n,\, x_n\in E_n.
            \end{equation}
    \end{itemize}
    Let $S_n(t)$ be the dynamical semigroup generated by $A_n+B_n$. Then the net $S\d(t)$ of semigroups converges in the sense of \cref{thm:evolution} and the generator of $S_\oo(t)$ is $A_\oo + B_\oo$.
\end{prop}

\begin{proof}
    We will check condition \ref{it:generator} of \cref{thm:evolution}.
    We define $C_n = A_n + B_n$ on $\dom{C_n}=\dom{B_n}$.
    The assumption $\dom{B\d}\supset\dom{A\d}$ implies $\dom{C\d}=\dom{A\d}$, $\dom{A_\oo}=\dom{C_\oo}$ and $\dom{B_\oo}\supset \dom{A_\oo}$.
    We check well-definedness of $C_\oo$: Let $x\d\in\dom{C\d}$ such that $x_\oo=0$, then $\lim_n \norm{C_n x_n} \le \lim_n(\norm{A_nx_n}+\norm{B_nx_n}) \le (1+a)\norm{A_\oo x_\oo}+b\norm{x_\oo} =0$.
    Finally, the relative-boundedness inequality \eqref{eq:relative_boundedness_net} carries over to the limit space by simply taking limits with $x\d$ being a \jt-convergent net in $\dom{A\d}$.
    In particular, we know that $C_\oo = A_\oo + B_\oo$.
    Now the standard perturbation theorem for dynamical semigroups \cite[III, Thm.~2.7]{engelnagel} implies that $C_\oo = A_\oo + B_\oo$ generates a dynamical semigroup.
\end{proof}

We can also use the technique of analytic vectors on $E_\oo$ to ensure the existence of the limit dynamics.
This follows from the easy Lemma:

\begin{lem}\label{thm:analytic_vecs}
    Let $A\d$ be a net of closed dissipative operators and assume that there is a seminorm-dense space $\D\subset\dom{A\d}$ of nets which are analytic in seminorm in the sense that there is a $t>0$ for each $x\d\in\D$ so that
    \begin{equation}\label{eq:analytic_semi}
        \sum_{k\in\NN}\, \frac{t^k\, \seminorm{A\d^kx\d}}{k!} <\oo.
    \end{equation}
    If $A_\oo$ is well-defined, then $A_\oo$ is a closed dissipative operator on $E_\oo$ and $\D_\oo =\jlim(\D)\subset\dom{A_\oo}$ is a dense set of analytical vectors.
\end{lem}

In many situations this implies that $A_\oo$ generates a semigroup of contractions (see, e.g., \cite[Prop.~3.1.18--3.1.22]{bratteli1}, \cite{nelson1959analytic} or \cite{reed_simon_2}).

\section{Inductive systems of C*-algebras and completely positive dynamics}\label{sec:cp}

Our techniques can be combined with additional structure.
Roughly speaking, any property that is asymptotically respected by the connecting maps $\j nm$ will pass onto the limit.
This section illustrates this by analyzing (soft) inductive systems of C*-algebras.
For the connecting maps, we allow for completely positive contractions instead of requiring *-homomorphisms.
For the limit space to be a C*-algebra, we assume asymptotic multiplicativity (similar to \cite{blackadar1997generalized}).
This will, for example, allow for a commutative limit space of non-commutative algebras (see \cref{exa:mean_field}), which would otherwise not be possible.
We will see that many expected results automatically hold for such a setup.
For example, the \jts-limit of a \jts-convergent net of states $\varphi\d$ is a state on $\A_\oo$, and the limit operation of a \jt\ojt-convergent net of completely positive contractions is a completely positive contraction.
In particular, the evolution theorem holds in the category of C*-algebras and completely positive contractions.

\begin{defin}
    A {\bf soft inductive system of C*-algebras} is a soft inductive system $(\A,j)$ of Banach spaces, such that all $\A_n$ are C*-algebras, the connecting maps $\j nm$ are completely positive contractions and satisfy the asymptotic homomorphism property:
    \begin{equation}\label{eq:asymmor}
        \lim_{n\gg m}\,\norm[\big]{\j nm((\j ml a_l)(\j ml b_l)) - (\j nl a_l)(\j nl b_l)} =0.
    \end{equation}
\end{defin}

If we assume that the C*-algebras are unital, then we should also assume that this structure is preserved by the connecting maps, i.e., we assume that $\j nm \1_{\A_m}=\1_{\A_n}$.
In this case, the net of units is automatically \jt-convergent.
It follows that the limit space of a soft inductive system of (unital) C*-algebras is again a (unital) C*-algebra:

\begin{prop}
    Let $(\A,j)$ be a soft inductive system of (unital) C*-algebras.
    Then the adjoint $a\d^* \coloneqq (a_n^*)$ of a \jt-convergent net $a\d$ is again \jt-convergent and products of \jt-convergent nets $a\d$ and $b\d$ are also \jt-convergent.
    The limit space becomes a (unital) C*-algebra with the operations
     \begin{equation}\label{eq:limit_csa}
         \paren*{\jlim_n a_n}^* \coloneqq \jlim_n a_n^*\qandq \paren*{\jlim_n a_n}\paren*{\jlim_n a_n} \coloneqq \jlim_n a_n b_n.
    \end{equation}
    In the unital case, the net of units $\1\d$ is always \jt-convergent and the unit of $\A_\oo$ is $\1_{\A_\oo}= \jlim_n \1_{\A_n}$.
\end{prop}

\begin{proof}
    Suppose $a\d$ is \jt-convergent, then $a\d^*$ is also \jt-convergent because
    \[
        \norm{a_n^* - \j nm a_m^*} =\norm{(a_n - \j nm a_m)^*} = \norm{a_n -\j nm a_m}.
    \]
    It suffices to check \jt-convergence of products on basic nets $a_n = \j nl a_l$ and $a_n =\j nl b_l$. We can assume both basic nets to start at the same index because of \cref{thm:density_thm}.
    Set $c\d = a\d b\d =(a_n b_n)$, then
    \begin{align*}
        \lim_{n\gg m}\norm{c_n - \j nm c_m}
        & = \lim_{n\gg m}\norm{(\j nla_l)(\j nlb_l) - \j mm((\j mla_l)(\j mlb_l))} =0.    \qedhere
    \end{align*}
\end{proof}

In fact, \cref{eq:asymmor} is equivalent to \jt-convergence of products of \jt-convergent nets.
For the remainder of this section, $(\A,j)$ denotes a unital soft inductive system of C*-algebras.
For a (soft) inductive system of C*-algebras $\nets(\A,j)=\prod \A_n$ is also a C*-algebra with $\C(\A,j)$ being a C*-subalgebra (now equipped with $\norm{}_\nets$) and by the above $\jlim : \C(\A,j)\to \A_\oo$ is a *-homomorphism.
In fact, we know that $\ker(\jlim)$ is the closed two-sided ideal of null nets so that $\A_\oo$ is isomorphic to the C*-quotient $\C(\A,j)/\C_0(\A,j)$:

The C*-algebras $\A_n$ (including the limit space) form a \emph{continuous field of C*-algebras} \cite[Ch.~10]{dixmier1982} over the topological space $N\cup\set\oo$ (equipped with the order topology) where the continuous sections are precisely defined to be \jt-convergent nets including their limit.
This is discussed in more detail in \cref{sec:comparison}.

\begin{cor}
    \begin{enumerate}[(1)]
        \item\label{it:limits_unitary_normal}
            Let $a\d$ be \jt-convergent. If all $a_n$ are unitary/normal/self-adjoint/ projections, then so is the \jt-limit $\jlim_n a_n$.
        \item\label{it:limits_spectrum}
            Let $a\d$ be \jt-convergent, then $\Sp(a_\oo) \subset \bigcap_k\bigcup_{k>n}\Sp(a_k)$.
        \item\label{it:functional_calc}
            Let $a\d$ be a \jt-convergent net of (normal) elements.
            Let $\Omega\subset \CC$ be open and let $f:\Omega\to \CC$ be an analytic (continuous) function.
            If $\Sp a_n\subset \Omega$ for all $n$, then the analytic (continuous) functional calculi $f(a\d)$ are also \jt-convergent and $\jlim_n f(a_n) = f(\jlim_n a_n)$.
        \item\label{it:positivity}
            An element of the limit space $a_\oo\in\A_\oo$ is positive if and only if there is a \jt-convergent net $a\d\in\C(\A,j)$ so that $\jlim_n a_n=a_\oo$.
    \end{enumerate}
\end{cor}

\begin{proof}
    \ref{it:limits_unitary_normal}:
        All of these follow from the two properties that the adjoint and the product preserve \jt-convergence and correspond to the adjoint and product on $\A_\oo$.

    \ref{it:limits_spectrum}:
        Consider the C* inclusion $\C(\A,j)\subset\nets(\A,j)$.
        To compute $\Sp(a\d)$ we may regard $a\d$ as an element of $\prod \A_n=\nets(\A,j)$ and we have $z\in\Sp(a\d)$ if and only if $(z\1\d-a\d)$ is non-invertible if and only if $(z\1-A_n)$ non-invertible for some $n$, so that $\Sp(a\d)\subset\bigcup_n \Sp(a_n)$.
        Since $\jlim:\C(\A,j)\to\A_\oo$ is a *-homomorphism we have $\Sp(\jlim_n a_n)\subset \Sp(a\d)$.

    \ref{it:functional_calc}:
        One can start by showing the claim for polynomials and then extend to continuous/analytic functions by approximation.

    \ref{it:positivity}:
        $a_\oo$ is positive if and only if $a_\oo =(b_\oo)^*b_\oo$ for some $b_\oo$.
        Now pick $b\d\in\C(E,j)$ with $b_\oo=\jlim_nb_n$. Then $b\d^*b\d$ is also \jt-convergent and its limit is equal to $a_\oo$, which proves the claim.
\end{proof}

A dynamical semigroup on a C*-algebra is a strongly continuous\footnote{This is also called ``point-norm continuous'' elsewhere.} one-parameter group of completely positive contractions.
We include the following characterization of generators of dynamical semigroups on C*-algebras, which is a consequence of the Arendt-Chernoff-Kato theorem \cite{arendt1982generalized}:

\begin{lem}\label{thm:ccp}
    Let $\TT(t)$ be a strongly continuous one-parameter semigroup on a C*-algebra $\A$, and let $\LL$ be its infinitesimal generator.
    Then $\TT(t)$ is completely positive if and only if the generator is conditionally completely positive in the sense that for all $0\le x\in \dom{\LL\ox\id_n}\subset \MM_n(\A)$ and all states $\omega$ on $\mathfrak S(\MM_n(\A))$:
    \begin{equation}\label{eq:conditionally_cp}
          \omega(x) =0 \implies \omega((\LL\ox\id_n)(x))\geq0.
    \end{equation}
    Another equivalent condition is the complete positivity of the resolvent $\R(\lambda)=(\lambda-\LL)^{-1}$ for all $\lambda>0$.
\end{lem}

Here $\MM_n(\A)$ is the C*-algebra of $n\times n$ matrices with entries in $\A$, which is isomorphic with $\MM_n(\CC)\ox\A$. The domain of ${\LL\ox\id_n}$ consists of all matrices whose entries are in $\dom{\LL}$, i.e., $\dom{\LL\ox\id_n}=\MM_n(\dom\LL)$.
\begin{proof}
    It is sufficient to prove that $\TT(t)$ is positive if and only if \cref{eq:conditionally_cp} holds for $n=1$, which is the content of the Arendt-Chernoff-Kato theorem.
\end{proof}

The Lemma and the Lumer-Phillips Theorem \cite[Ch.~III, Thm.~3.15]{engelnagel} imply that generators of dynamical semigroups are precisely the conditionally completely positive closed dissipative operators $\LL$ so that $(\lambda-\LL)^{-1}$ has full range.
We finish this section by repeating the convergence theorem for nets of completely positive semigroups on soft inductive systems of C*-algebras.

\begin{thm}
    Let $(\A,j)$ be a (soft) inductive system of C*-algebras and let $\TT\d$ be a net of dynamical semigroups with $\LL\d$ denoting the corresponding net of generators.
    Let $\lambda\in\CC$ be such that $\Re\lambda > 0$. The following are equivalent:
    \begin{enumerate}[(1)]
        \item\label{it:semigroups*}
            $\TT\d(t)$ is \jt\jt-convergent and the limit operators $\TT_\oo(t)$ are strongly continuous in $t$, hence form a dynamical semigroup.
        \item\label{it:resolvents*}
            The net $\R\d(\lambda)$ of resolvents $\R_n(\lambda)=(\lambda-\LL_n)^{-1}$ is \jt\jt-convergent and the limit operation $\R_\oo(\lambda)$ has dense range.
        \item\label{it:net_core*}
            There is a seminorm dense subspace $\D\subset \dom{\LL\d}$ such that $(\lambda-\LL\d)\D$ is also seminorm dense.
        \item\label{it:generator_cs*}
            $\LL_\oo$ is well-defined, hence closed, and generates a dynamical semigroup on $\A_\oo$.
    \end{enumerate}
    If these hold, then it follows that $\dom{\LL\d}=\R\d(\lambda)\C(E,j)$, that the limits of the resolvents are the resolvents of $\LL_\oo$ and that the semigroup generated $\LL_\oo$ are the limits of $\TT\d(t)$.
    In particular, these claims hold for some $\lambda$ if and only if they hold for all $\lambda$ with positive real part.
\end{thm}

\begin{proof}
    One only has to check that (conditional) complete positivity is preserved in the limit.
    That these are equivalent is a consequence of \cref{thm:ccp}.
    That complete positivity passes to the limit follows from combining the fact that \jt-limits of positive nets are positive  and that the matrix amplifications $\MM_\nu(\A_n)$ form a soft inductive system with the connecting maps $j_{nm}\ox\id_\nu$ whose limit space is naturally isomorphic with $\MM_\nu(\A_\oo)$.
    The former implies that the $T_\oo(t)$ are positive semigroups, and the latter implies that the same applies to all matrix amplifications so that $T_\oo(t)$ is completely positive.
\end{proof}

\newcommand\B{\mc B}
\def\jklim{(j\!\ox\!k)\mkern-.5mu\hbox{-}\mkern-3mu\lim}
\def\klim{\ensuremath{k}\mkern-1mu\hbox{-}\mkern-3mu\lim}
\renewcommand\k[2]{k_{{#1}{#2}}}

Finally, we discuss tensor products of soft inductive limits of C*-algebras.
If $\A$ and $\B$ are unital C*-algebras we denote by $\A\odot\B$, $\A\ox_{min}\B$ and $\A\ox_{max}\B$ the algebraic, minimal and maximal (C*-) tensor product respectively.
Let $(\A,j)$ and $(\mc B,k)$ be soft inductive limits of C*-algebras over the same directed set $N$ and set
$$
(\A\ox_{*}\B)_n=\A_n\ox_*\B_n \qandq (j\ox k)_{nm}=\j nm\ox k_{nm},
$$
where $*$ is either \emph{min} or \emph{max}. Note that $(j\ox k)_{nm}$ are unital completely positive maps in both cases.

\begin{prop}\label{prop:tensor}
    Let $\ox_*$ denote either the minimal or maximal C*-tensor product.
    \begin{enumerate}[(1)]
        \item\label{it:TP_soft}
            $(\A\ox_*\B,j\ox k)$ is a soft inductive system of C*-algebras.
            If $a_{i,\blob}\in\C(\A,j)$ and $b_{i,\blob}$, $i=1,\ldots k$, then $\sum_{i=1}^k a_{i,\blob} \ox b_{i,\blob} \in\C(\A\ox_*\B,j\ox k)$.
            This induces an embedding of the algebraic tensor product $\A_\oo\odot\B_\oo$ into the limit space $(\A\ox_*\B)_\oo$. This embedding has dense range.
        \item\label{it:TP_phi}
            The maps $\phi\d$ with $\phi_n:=\j\oo n \ox k_{\oo n}:\A_n\ox_*\B_n\to \A_\oo\ox_*\B_\oo$ take $(j\ox k)$-convergent nets to Cauchy nets.
            The limit operation $\phi_\oo:(\A\ox_*\B)_\oo\to\A_\oo\ox_*\B_\oo$ is a surjective *-homomorphism.
        \item\label{it:TP_max}
            $(\A\ox_{max}\B)_\oo$ is the maximal C*-tensor product of $\A_\oo$ and $\B_\oo$.
        \item\label{it:TP_min}
            $(\A\ox_{min}\B)_\oo$ is a C*-tensor product of $\A_\oo$ and $\B_\oo$.
            $\phi_\oo$ is the canonical homomorphism onto the minimal C*-tensor product.
    \end{enumerate}
\end{prop}

In general, $(\A\ox_{min}\B)_\oo$ is not equal to the minimal C*-tensor product. Counterexamples exist already for strict inductive systems $(\A,j)$ of C*-algebras where the $\j nm$ are *-homomorphisms \cite[\,\!II.9.6.5]{blackadar2006operator}.

\begin{proof}
    \ref{it:TP_soft}: It suffices to check $(j\ox k)$-convergence and asymptotic multiplicaticity on basic sequences of the form $(j\ox k)_{\blob l} x_l$ with $x_l\in\A_l\odot\B_l$. For these, the properties follow from their counterparts on $(\A,j)$ and $(\B,k)$ and the triangle inequality.
    Furthermore, such basic sequences are seminorm dense, implying that the embedding of the algebraic tensor product is dense.

    \ref{it:TP_phi}:
    It suffices to check basic nets with $(j\ox k)_{\blob l}x_l$ with $x_l= \sum a_i \ox b_i\in\A_n\odot\B_n$ where the claim follows again from the triangle inequality.
    That $\phi_\oo$ is a *-homomorphism is clear from the asymptotic multiplicativity of $\j\oo n$ and $k_{\oo n}$.
    That it is surjective follows from the fact that the algebraic tensor product $\A_\oo\odot\B_\oo$ is dense in $\A_\oo\ox_*\B_\oo$ and is part of the range of $\phi$.

    \ref{it:TP_max} and \ref{it:TP_min}:
    We start by showing that $(\A\ox_*\B)_\oo$ is a C*-tensor product, i.e., that $\norm{\sum_i a_{i,\oo}\ox b_{i,\oo}}_\beta=\norm{\jklim_n a_{i,n}\ox b_{i,n}}$ is a C*-norm on the algebraic tensor product \cite[Sec.~IV.4]{takesaki1}.
    The triangle inequality, submultiplicativity, and the C*-property $\norm{x^*x}_\beta=\norm{x}_\beta^2$ are inherited from the norm on the limit space.
    Furthermore, $\norm{\jklim_n a_n\ox b_n}=\lim_n \norm{a_n\ox b_b}\le \lim_n\norm{a_n}\norm{b_n}\le \norm{a_\oo}\norm{b_\oo}$.
    Non-degeneracy of $\norm{}_\beta$ holds by \ref{it:TP_phi} which proves $\norm{}_\beta\le \norm{}_*$.
    If $*=$ \emph{max}, then this implies $\norm{}_\beta=\norm{}_{max}$ \cite[Sec.~IV.4]{takesaki1}.
\end{proof}

\section{Examples and Applications} \label{sec:ex}

\subsection{Quantum dynamics in the classical limit}\label{sec:classical_limit}

An approach to the classical limit via soft inductive limits of C*-algebras was published in \cite{classicallimit}.
We discuss here a corresponding version in the Schr\"odinger picture, emphasizing convergence of quantum dynamics to classical dynamics in the classical limit.
We plan to publish these and more results on the classical limit with full proofs in the future and will focus on conveying the bigger picture here.

The underlying directed set for the classical limit is the interval $(0,1]$ directed towards zero.
For this reason, we write ``$0$'' instead of ``$\infty$'' to denote limit objects.
An element $\hb\in (0,1]$ is thought of as an \emph{action scale} of the quantum system with $d$ canonical degrees of freedom with Hilbert space $L^2(\RR^d, dx)$.
The role of the connecting maps $\j\hb\hbp$ that we will introduce is to change the action scale from \hb\ to \hbp.
We will refer to a family of objects (e.g.\ states) indexed by $\hb\in(0,1]$ as an (\hb-)sequence because the directed set is totally (and strictly) ordered so that "net" seems to be an unnecessarily complicated term.

The Banach space in which the states of the quantum system at scale \hb\ live is the trace class $\TC\h \coloneqq \TC(\H)$.
To define the connecting maps $\j\hb\hbp :\TC\hp\to\TC\h$, we start with the well-known coherent-state quantization and its dequantization counterpart, which we denote by $\j\hb0$ and $\j0\hb$, respectively.
By $\ket z\h$, $z\in\RR^{2d}$, we denote the coherent-state vectors obtained from displacing the ground state $\ket 0\h$ of the harmonic oscillator\footnote{The Harmonic oscillator has the \hb-scaling $H\h= \frac12(x^2 - \hb^2 \Delta_x)$ and the displacement operators are $W_z^\hb\psi (x) = \exp\{(i/\hb)(p\cdot x+i \hb \nabla_q)\}\psi(x)$, $z=(q,p)$. If $\ket0\h$ is the ground state of $H\h$, then the coherent states are $\ket z\h=W_z^\hb\ket0\h$.}.
With these, we define
\begin{align}
    &&\j\hb0(\rho_0) &\coloneqq \int_{\RR^{2d}} \rho_0(z) \, \kettbra z\h \ dz,&\rho_0&\in L^1(\RR^{2d}) \label{eq:jh0}\\
    &&\j0\hb(\rho\h)(z) &\coloneqq  \tfrac1{(2\pi\hb)^{d}} \ {\bra z \rho\h \ket z\h} ,&\rho\h&\in\TC\h.\label{eq:j0h}
\end{align}
These maps take probability distributions (or rather their densities) on phase space to quantum states and vice versa\footnote{This is a consequence of the overcompleteness property $\int \kettbra z\h\, dz = (2\pi\hb)^d\1$.}.
They are also known as upper and lower symbols or Wick dequantization and anti-Wick quantization, respectively (see \cite{folland2016harmonic}).
The connecting maps are simply defined by
\begin{equation}\label{eq:classical_split}
    \j\hb\hbp \coloneqq \j\hb0\circ \j0\hbp : \TC\hp\to\TC\h.
\end{equation}
By construction, the connecting maps are completely positive and trace-preserving (so-called \emph{quantum channels}). In particular, they are linear contractions.

\begin{lem}
    $(\TC,j)$ is a soft inductive limit of Banach spaces, i.e.,
    \begin{equation}\label{eq:CL_soft_ind}
        \lim_{\hb\ll\hbp}\norm{(\j\hb{\hb''}-\j\hb\hbp\j\hbp{\hb''} )\rho_{\hb''}}_1 =0\quad \forall {\hb''}>0,\ \rho_{\hb''}\in\TC_{\hb''}.
    \end{equation}
    The limit space $\TC_0$ is isometrically isomorphic with $L^1(\RR^{2d})$ with the isomorphism being defined by
    \begin{equation}\label{eq:CL_limit_space_iso}
        \TC_0\cong L^1(\RR^{2d}) \quad \text{via}\quad \jlim\h \rho\h \mapsto \lim\h \j0\hb\rho\h.
    \end{equation}
    Indeed, if an \hb-sequence $\rho\d$ is \jt-convergent, then its dequantizations $\j0\hb\rho\h$ converge in the topology of $L^1(\RR^{2d})$.
    The abstract maps $\j0\hb$ (defined as in \cref{eq:embedding_into_limit}) and the one defined in \cref{eq:j0h} coincide (modulo the above isomorphism), i.e.\ one has $\jlim\h \j\hb\hbp\rho\hp =\j0\hbp\rho\hp$.
\end{lem}

\begin{proof}
    As we have $(\j\hb{\hb''}-\j\hb\hbp\j\hbp{\hb''} ) = \j\hb0(\id - \j0\hb\j\hb0)\j0{\hb''}$, it follows that the norm in \cref{eq:CL_soft_ind} is bounded by $\norm{(\id - \j0\hbp\j\hbp0)\sigma_0}_1$ with $\sigma_0 = \j0{\hb''}\rho_{\hb''}$.
    It is well-known that $\j0\hb\j\hb0$ is the heat transform at time \hb, i.e.\ it convolves with a Gaussian having variance \hb, and the claim follows from strong-continuity of the heat semigroup on $L^1(\RR^{2n})$.

    One sees that $\j0\hb\rho\h$ is a Cauchy sequence in $\hb$ for any $\rho\d\in\C(\TC,j)$ by considering
    \begin{align*}
        \norm{\j0\hb\rho\h - \j0\hbp\rho\hp}_1
        &\le \norm{\j0\hb(\rho\h-\j\hb\hbp\rho\hp)}_1+\norm{(\j0\hb\j\hb0-\id)\j0\hbp\rho\hp}_1 \xrightarrow{\hb\ll\hbp}0.\qedhere
    \end{align*}
\end{proof}

A special feature of this soft inductive system, not assumed in the abstract setup of soft inductive systems, are the quantization maps $\j\hb0$ and their properties.
For any $\rho_0\in L^1(\RR^{2d})$ one has that $\j\blob0\rho_0$ is \jt-convergent and the limit is $\jlim\h \j\hb0\rho_0 = \rho_0$.
This has the consequence that for any dense subspace $\D_0\in L^1(\RR^{2d})$, the space $\D = \set{\j\blob0\rho_0\given \rho_0\in \D_0}$ is seminorm dense in $\C(E,j)$.
Furthermore, this makes for the following simplification: \jts-convergence of a sequence $A\d$ of observables $A\h\in\LL(\H)$  is equivalent to existence of the limit $\wslim\h\j0\hb A\h$ in $L^\oo(\RR^{2d})=(L^1(\RR^{2d}))'$.
Another interesting property of this specific soft inductive system is that even though the spaces $\TC\h$ are the same for all $\hb$, the limit space is different.

We now turn to the discussion of dynamics in the classical limit. This is done separately for Hamiltonian dynamics and irreversible dynamics generated by a Lindblad operator.
Both cases will be proved by applying \cref{thm:evolution_cor}, which reduces the problem to an infinitesimal one, provided that one already has good control over the expected limit dynamics.
We start with the Hamiltonian case, where the main result is:

\begin{thm}\label{thm:CL_evolution}
    Let $V:\RR^d\to \RR$ be a $C^2$-potential with bounded second-order derivatives and consider the Schr\"odinger operator $H\h = -\frac{\hb^2}2 \Delta + V(x)$ and the classical Hamiltonian function $H_0(q,p)=\frac12 p^2+V(q)$.
    Then the quantum dynamics $\rho\h \mapsto  e^{-\frac i\hb tH\h}\rho\h e^{\frac i\hb tH\h}$ is convergent in the sense of \cref{thm:evolution}.
    The limiting operation is the classical time evolution generated by $H_0$, i.e., the classical limit of time-evolved quantum states follows the classical flow generated by $H_0$.
    In particular,
    \begin{equation}
        \frac d{dt} \bigg|_{t=0} \ \jlim\h \paren[\Big]{e^{-\frac i\hb H\h}\rho\h e^{\frac i\hb t H\h}} = \braces[\Big]{H_0,\jlim\h \rho\h}\quad \forall\rho\d\in \C(\TC,j)
    \end{equation}
    provided that $\jlim\h\rho\h$ is suitably differentiable, where $\braces{f,g}$ denotes the Poisson bracket of functions $f$ and $g$.
\end{thm}

The assumptions on the potential guarantee that $H\h$ is essentially self-adjoint on the domain of Schwartz-functions for every \hb\ \cite{bauer2022self} and for the dynamics of $H_0$ to exist for all times and all initial values (by the Picard-Lindel\"of theorem).

\begin{proof}[Sketch of proof]
    It can readily be checked that $C_c^2(\RR^{2d})$ is a core for the classical dynamics generated by $H_0$ on $L^1(\RR^{2d})$.
    Applying \cref{thm:evolution_cor} to $\D = \set{\j\blob0\rho_0\given \rho_0\in C_c^2(\RR^{2d})}$ reduces the problem to showing that commutators converge to Poisson brackets on $\D$, i.e.\ showing that
    \begin{equation}\label{eq:bracket_thm}
        \jlim\h \paren[\Big]{-\frac i\hb \bracks[\big]{H\h,\rho\h}} = \braces{H_0,\rho_0}\quad \forall \rho\d\in\D.
    \end{equation}
    One can reduce the proof to the case where $H = \j\hb0 H_0$ where the claim can be proved explicitly using the explicit form of the coherent state quantization.
\end{proof}

The proof, in fact, works for a much larger class of Hamiltonians containing all coherent-state quantized Hamiltonians $H\h = (2\pi\hb)^{-d} \int H_0(z) \kettbra z\h dz$ of $C^2$-functions with bounded second-order derivatives and all Weyl quantizations of $C^{2d+3}$-functions with uniformly bounded derivatives of second and higher orders (both of these conditions guarantee essential self-adjointness \cite{bauer2022self,fulsche2023simple}).
A version of this theorem for bounded Hamiltonians was already included in \cite{classicallimit}.
The extension to unbounded Hamiltonians is taken from \cite{masterthesis} and again builds on \cref{thm:evolution}.

We now turn to discuss irreversible dynamics.
On the quantum side, irreversible dynamics correspond to strongly continuous semigroups of trace-preserving completely positive maps \cite{siemon}.
Consider the Lindbladian of a Gaussian (or quasi-free) dynamical semigroup \cite{hybrids, BarchWer}
\begin{align}\label{eq:quasifree_lindblad}
    \LL\h(\rho\h)
    &= -\frac i{2\hb}\sum_{jk} \paren[\Big]{ A_{jk}\bracks[\big]{R_jR_k,\rho\h} + i M_{jk} \paren[\big]{R_j [R_k,\rho\h] + [\rho\h,R_j] R_k}}
\end{align}
where $R = (x_1,\ldots,x_d,-i\hb\partial_1,\ldots,-i\hb \partial_d)$ is the vector of canonical operators and where $A$ is a symmetric real matrix and $M$ is a positive semi-definite complex matrix.
We remark that \cref{eq:quasifree_lindblad} can be transferred to the form $\LL\h(\rho\h) = -(i/\hb)[H,\rho] + (1/2\hb) \sum_j (L_j[\rho\h,L_j^*] + [L_j,\rho\h]L_j^*)$ with jump operators $L_j = \sum_k (\sqrt M)_{jk}R_k$ and Hamiltonian $H = \frac12 R\cdot A R$.
Denote the symplectic matrix by $\sigma$.

\begin{thm}
    For any complex positive semi-definite matrix $M$ and any real symmetric matrix $A$, the Gaussian dynamical semigroup generated by \cref{eq:quasifree_lindblad} is convergent in the sense of \cref{thm:evolution}.
    The limit semigroup on $L^1(\RR^{2d})$ is generated by the first-order differential operator $\LL_0 = z^\top \cdot (A-\Im M)\sigma\nabla$, $z\in\RR^{2d}$.
\end{thm}

The notation $\Im M$ means $\frac i2(M^*-M)$.
Note that any real $2d\times 2d$ matrix can be written as $K = (A-\Im M)\sigma$ for a real symmetric matrix $A$ and a complex positive semi-definite matrix $M$.
The limit semigroup that arises is the push-forward along the flow $e^{t(A-\Im M)\sigma}$ on phase space.

\begin{proof}[Sketch of proof]
    It can easily be shown that $C_c^\oo(\RR^{2d})$ is a core for the classical dynamics generated by $\LL_0$ on $L^1(\RR^{2d})$.
    The idea is to apply \cref{thm:evolution_cor} with $\D = \set{\j\blob 0\rho_0 \given \rho_0\in C_c^2(\RR^{2d})}$.
    From the intuition that the classical limit of $-(i/\hb)[R_j,\,\cdot\,]$ is the Poisson bracket $\braces{z_j,\,\cdot\,}$ with $z=(q,p)$, one already expects that the classical limit of a quasi-free Lindbladian is
    \begin{align}\label{eq:lb_qf}
        \LL_0(\rho_0)
        &= \frac12\sum_{ij} \paren[\Big]{ A_{ij}\braces{r_ir_j,\rho_0} + i \paren[\big]{ M_{ij} -\Bar{M_{ij} }} r_i \braces{r_j,\rho_0} }
        = r\cdot \paren{A-\Im M}\sigma \nabla\rho_0.
    \end{align}
    The only step that remains is proving that $\LL\d(\j\blob0\rho)$ is indeed \jt-convergent to $\LL_0(\rho_0)$ for all $\rho_0\in C^\infty_c(\RR^{2d})$.
    This relies on explicit properties of the coherent-state quantization and will be published separately.
\end{proof}

To illustrate this theorem, we discuss the \emph{damped oscillator}.
Classically, its dynamics is described by the vector field $X =  (-\alpha q-p)\partial_q + (q-\alpha p)\partial_p$ on phase space,
where $\alpha>0$ determines the strength of the damping.
Matrices implementing this are for example
\begin{equation}\label{eq:dosc_mat}
    A = \mat{1&0\\0&1} \qandq M = \alpha \mat{1&i\\-i&1}.
\end{equation}
The jump operators corresponding to this are the creation and annihilation operators $L_q =\sqrt{\alpha/2} (x-i\hb\partial_x)$ and $L_p =L_q^*$ (because $\sqrt M=\sqrt{2/\alpha}M$).
This is the standard model for describing a laser coupled to a thermal bath.

\begin{figure}[htpb]\centering
    \includegraphics[height = 6cm]{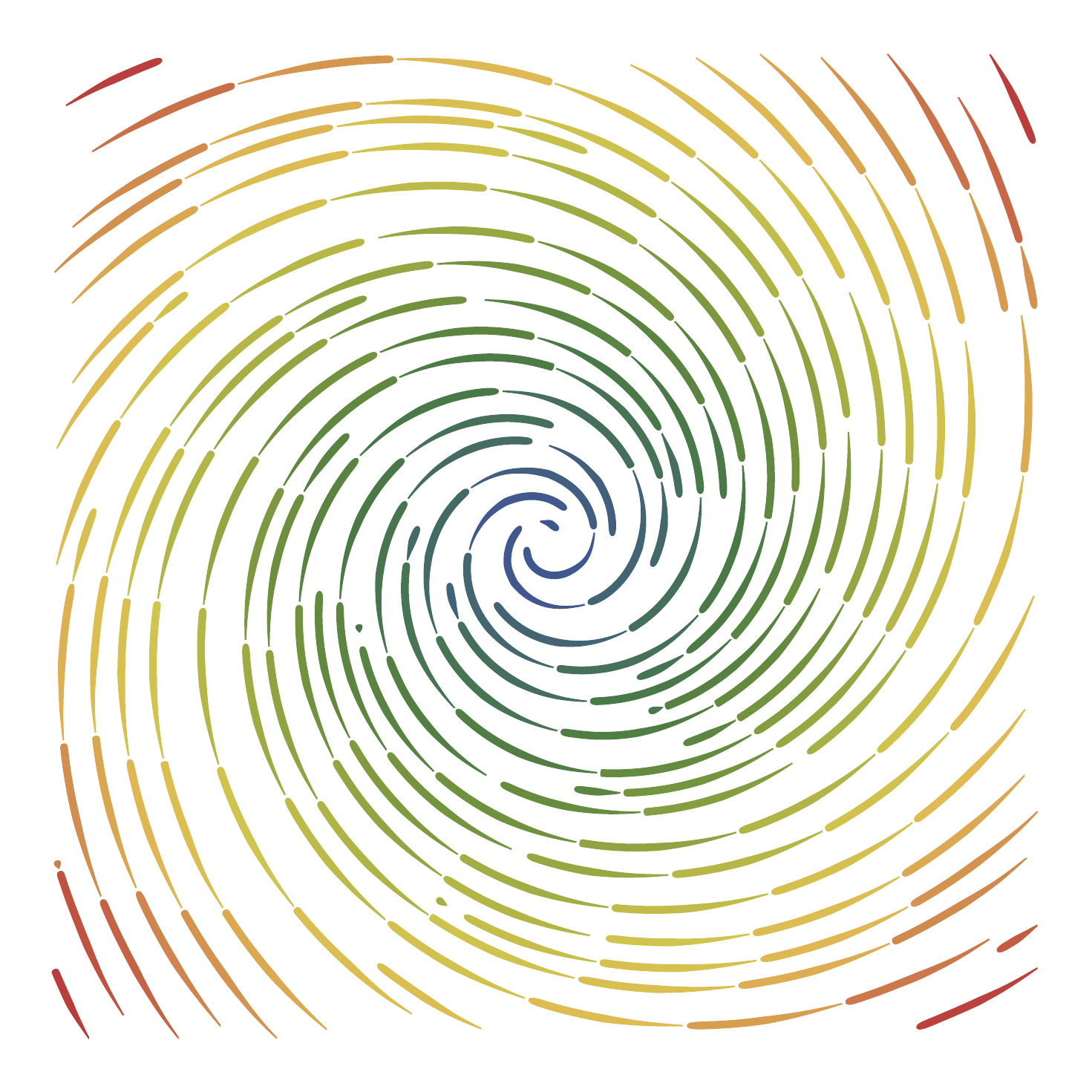}
    \begin{caption}{Streamlines of the vector field $X =  (-\alpha q-p)\partial_q + (q-\alpha p)\partial_p$ on phase space describing the (classical) damped harmonic oscillator with damping constant $\alpha>0$.}\end{caption}
\end{figure}

This example reveals an interesting difference between classical and quantum open systems.
The classical limit of irreversible Gaussian semigroups with the \hb-scaling as in \eqref{eq:quasifree_lindblad} is always reversible (but not necessarily of Hamiltonian type).
This is because the noise on the quantum side (necessary for complete positivity \cite{hybrids}) is not needed on the classical side.
Indeed, this noise is precisely described by $\Re M$.
However, one can obtain proper irreversible semigroups if one keeps the noise by rescaling $\Re M$ an \hb-dependent way \cite{hybrids}, leading to dissipative classical dynamics.

\subsection{Mean field limit}\label{exa:mean_field}

\input{mf}

\subsection{Spin systems, dynamics and the thermodynamic limit}\label{sec:quasi-local}

\newcommand\La{\Lambda}

In this subsection, we discuss quantum spin systems in the thermodynamic limit using quasi-local algebras, which are an example of inductive systems of C*-algebras \cite{bratteli2}.
We will show that the usual assumptions for the existence of the dynamics in the thermodynamic limit already guarantee that the dynamics is convergent in the sense of \cref{thm:evolution} irrespective of \emph{large class of boundary conditions} imposed on finite lattices (e.g., periodic or anti-periodic boundary conditions).

The directed set $(N,\le)$ is that of finite subsets $\Lambda$ of the lattice $\ZZ^d$, ordered by inclusion.
The algebra $\A_\Lambda$ is defined as $\A_\Lambda = \bigotimes_{i\in\Lambda}\A_{\set{*}}$ where the "one-site algebra" $\A_{\set{*}}$ is a unital C*-algebra and the tensor product is the minimal one.
In the standard case we have $\A_{\set{*}}=\MM_n(\CC)$.
The connecting maps $\j\Lambda{\Lambda'}$ are the natural inclusions and will usually be suppressed unless explicitly required, i.e.\ we will simply write $\A_\Lambda\subset \A_{\Lambda'}$ if $\Lambda \subset \Lambda'$.
This yields a strict inductive system $(\A, j)$ with *-homomorphisms as connecting maps.
By construction $\Lambda\cap\Lambda'=\emptyset$ implies that $[\A_\Lambda,\A_{\Lambda'}]=\set0$ so that $\A_\oo$ becomes a quasi-local algebra \cite[Def.~2.6.3]{bratteli1}.
It is not hard to see that a net of elements $(a_\Lambda)$, $a_\Lambda\in\A_\Lambda$, is \jt-convergent if and only if $\lim_\Lambda a_\Lambda$ exists in the norm of $\A_\oo$.
This defines an isomorphism of the quasi-local algebra and the limit space via $\jlim_\Lambda a_\Lambda \equiv \lim_\Lambda a_\Lambda$, and the basic sequences are just the constant sequences $a$ so that $a\in\A_\La$ for some $\Lambda$.
From another perspective $\A_\oo$ can be understood as the infinite tensor product $\A_\oo = \bigotimes_{x\in\ZZ^d}\A_{\{*\}}$.

Formally speaking, all dynamics of quantum spin systems are defined in terms of \emph{interactions} (see the remark after \cite[Prop.~6.2.3]{bratteli2}):
A (formal) interaction is a map $\Phi$ that associates to a finite subset $\Lambda\subset \ZZ^d$ a hermitian element of $\A_\Lambda$.
From an interaction, we obtain a Hamiltonian and a bounded *-derivation
\begin{equation}\label{eq:Hloc}
    H_\Lambda = \sum_{\Lambda'\subseteq \Lambda}\Phi(\Lambda')\in \A_\Lambda \qandq \delta_\Lambda = i [H_\Lambda,\placeholder] \in \LL(\A_\Lambda),
\end{equation}
for every finite region $\Lambda$. Note that $H_\Lambda$ and $\delta_\Lambda$ are bounded for all $\Lambda$.
The net domain of the net of generators $\delta\d$ as defined in \cref{eq:net_domain} is $\dom{\delta\d} = \set{(a_\Lambda) \given \lim_\Lambda a_\Lambda,\,\lim_\Lambda \delta_\Lambda(a_\Lambda)\ \text{exist}}$.

If $\Phi$ satisfies the assumption
\begin{equation}\label{eq:abs_sum}
    p_\Phi(x) = \sum_{\La\ni x} \norm{\Phi(\La)}<\oo\quad \forall x\in \ZZ^d,
\end{equation}
we can define a *-derivation $\delta$ for the infinite system on the dense subalgebra $\D=\bigcup_{\Lambda\subset\ZZ^{d}}\A_{\Lambda}\subset\A_{\infty}$ by
\begin{equation}\label{eq:der}
    \delta(a) = i\! \sum_{\Lambda'\cap\Lambda\neq\emptyset}[\Phi(\Lambda'),a], \quad a  \in\A_{\Lambda}.
\end{equation}
This is indeed well-defined, i.e., does not depend on the choice of $\Lambda$ with $a\in\A_\Lambda$, because:
\begin{align}
\label{eq:fe}
\norm{\delta(a)} & \leq 2\sum_{x\in\La}\sum_{\La'\ni x}\norm{\Phi(\La')}\norm{a} \leq 2|\La|\sup_{x\in\La}p_{\Phi}(x)\norm{a},
\end{align}

\begin{defin}
    A {\bf net of boundary conditions} is a net $\beta\d$ of *-derivations $\beta_\La:\A_\La\to\A_\La$ so that for all $a\in \D$, $\beta_\La(a)\to0$.
\end{defin}

Applying our criteria from \cref{thm:unb_op_con} and \cref{thm:analytic_vecs} we obtain the following:

\begin{thm}
\label{thm:spindyn}
    Assume that $\Phi$ satisfies \cref{eq:abs_sum} and let $\beta\d$ be a net of boundary conditions.
    Then
    \begin{enumerate}[(1)]
        \item\label{it:delta_con} $(\delta_\La+\beta_{\La})(a) \to \delta(a)$ for all $a\in\D$.
        \item\label{it:delta_well-def} $\delta_\oo$ is well-defined (see \cref{def:unbounded}) and a closed operator, such that $\delta \subseteq\delta_{\infty}$.
    \end{enumerate}
    Assume exponential decay of long-range interactions in the sense that
    \begin{align}\label{eq:exp_decay}
         \sum_{n=0}^{\infty}\,\bigg(e^{n\lambda}\,\sup_{x\in\ZZ^{d}}\sum_{{\substack{\La\ni x \\ |\Lambda|=n+1}}}\norm{\Phi(\La)} \bigg) <\oo
    \end{align}
    for some $\lambda>0$, then
    \begin{enumerate}[resume*]
        \item\label{it:delta_analytic} the net of dynamics $e^{t(\delta\d+\beta\d)}$ satisfies the equivalent properties of \cref{thm:evolution}.
            The subspace $\D$ is a core for $\delta_\oo$, i.e., $\delta_\oo=\Bar\delta$.
            In particular, the limit dynamics $e^{t\delta_{\infty}}$ is independent of the boundary conditions.
    \end{enumerate}
\end{thm}

The assumption of exponential decaying long-range interactions is just a sufficient condition for $\Bar\delta$ to be a generator.
Whenever $\Bar\delta$ generates a strongly continuous group, the independence of boundary conditions and the \jt\jt-convergence of the dynamics holds.
This can, for example, be checked by showing that the range of $(\delta\pm\lambda)$ has dense range.

\begin{proof}
Consider first the case without boundary conditions $\beta\d$:
    \ref{it:delta_con}: Note that by \cref{eq:abs_sum} and \cref{eq:fe}, $\delta$ is well defined on $\D$. For $\La'$ and $a\in\A_{\La'}$, it follows that
    \begin{align*}
    \norm{(\delta_{\La}-\delta)(a)} & = \norm{\sum_{{\substack{\La''\cap\La'\neq\emptyset \\ \La''\subsetneq\La}}} [\Phi(\La'),a]} = \norm{\sum_{\La''\cap\La'\neq\emptyset} (1-\chi_{\La}(\La'')) [\Phi(\La''),a]}
    \end{align*}
    for any $\La\supseteq\La'$, where $\chi_{\La}$ is the indicator function of the subset $\La$, meaning it takes the value ``$1$'' if and only if it is evaluated on a subset $\La''\subseteq\La$, and ``$0$'' otherwise. Since this function converges pointwise to the function $1$ in the limit $\La\rightarrow\ZZ^{d}$, we obtain the results by dominated convergence due to \cref{eq:fe} (as the sum only runs over finite subsets of $\ZZ^{d}$).

    \ref{it:delta_well-def}:
    We apply \cref{it:w_graph_con} of \cref{thm:unb_op_con} to the space $P\d$ of \jts-convergent nets of the form $\varphi\d = (\varphi_\oo|_{\A_\La})_\La$ for some $\varphi_\oo\in\A_\oo^*$.
    It suffices to check \jts-convergence of $\delta\d^*(\varphi\d)$ for $\varphi\d\in P\d$ on basic nets which are essentially just elements of $\D$.
    On these it holds since $\delta_\La^*(\varphi_\La)(a) = \varphi_\oo(\delta_\La(a))\to \varphi_\oo(\delta(a))$ for all $a\in\D$.

    For the case including boundary conditions $\beta\d$: Both \cref{it:delta_con,it:delta_well-def} also apply to $\tilde\delta\d=\delta\d+\beta\d$ and $\tilde\delta_\oo=\delta_\oo$ for all $\beta\d$.

    \ref{it:delta_analytic}:
    From \cref{eq:exp_decay} it follows that all $a\in\D$ are analytic vectors for $\delta$ and hence for $\delta_\oo$ \cite[Thm.~6.2.4]{bratteli2}.
    We apply \cref{thm:analytic_vecs} to $\tilde\delta\d$ and the net domain of basic nets.
    It follows from the theory of one-parameter groups on Banach spaces that a closed operator $\delta_\oo$ with a dense set of analytic vectors such that $\pm\delta_\oo$ is dissipative, generates a strongly continuous one-parameter group \cite[Thm.~3.2.50]{bratteli1}.
    Since $\delta_\oo$ is a *-derivation, this is a group of *-automorphisms.
\end{proof}

To give some context to \cref{thm:spindyn}, in particular, the independence of the limit dynamics from boundary conditions, we consider a nearest-neighbor interaction with (anti) periodic or open boundary conditions. For simplicity, we restrict the limit to sequences of growing cubes $\La_n = [-n,n]^d$, which are cofinal for $N$, with $\La_{n}\to\ZZ^{d}$. If the long-range interactions decay exponentially, this theorem shows that the dynamics are convergent and that the limit is independent of the boundary conditions.
A possible way to see this would consist in setting all interactions on the boundary to zero, i.e., to pick $\beta_\La$ such that $\delta_\La+\beta_\La$ is zero on $\A_{\partial\Lambda}$.
Clearly, the theorem also implies that the dynamics converges to the one generated by $\bar{\delta}$ (this does not depend on how one defines the boundary $\partial\La$ as long as every $\La_0$ is contained in $\La \setminus \partial \La$ for some $\La\supset \La_0$).

\subsection{Quantum scaling limits} \label{sec:quantscale}

Another important application of inductive systems and the associated notion of convergence for dynamical semigroups and their generators concerns the construction of models in quantum field theory (QFT) via scaling limits in the framework of the Wilson-Kadanoff renormalization group (RG) \cite{wilson1975kondo, kadanoff2014renormalization}. The specific scaling limit procedure, coined \emph{operator-algebraic renormalization} (OAR), was formulated by one of the authors and has been explicitly realized in the context of bosonic and fermionic field theories \cite{stottmeister2021oar, morinelli2021scalar, osborne2023ising}, lattice gauge theory \cite{brothier2019gauge, brothier2019canonical}, conformal field theory \cite{osborne2022cftapprox, osborne2021cftsim} as well as more general anyonic models \cite{stottmeister2022anyon}. In OAR, the typical setting is that of inductive systems of C*-algebras, which are understood as realizations of the RG. Specifically, the algebras $\A_{n}$ represent a given physical system at different scales ``$n$'' and the connecting maps $\j nm$ are the renormalization group or scale transformations inducing an RG flow on the respective state spaces given an initial net $\omega_{n}\in\fS(\A_{n})$:
\begin{align}
\label{eq:rgflow}
\omega^{(m)}_{n}(A_{m}) & = \omega_{n}\circ\j nm(A_{m}), & A_{m}\in\A_{m}.
\end{align}
In this setting, $\omega^{(m)}_{n}$ is called the $n$th renormalized state at scale $m$, and we are interested in the possible cluster points of the nets $\omega^{(m)}_{n}$ (in $n$) at all scales which precisely correspond to the cluster points of the net $\omega\d$. We point out that the notation here differs from the convention that we use in most of the paper (to be consistent with the OAR notation): The upper index $m$ of $\omega_n\up m$ specifies the scale of the state while the lower index identifies it as part of a net at that scale.\\

As a specific example, we discuss the construction of a free-fermion field theory on a $d+1$-dimensional space-time cylinder, $\RR\times \T^{d}_{L}$, with spatial volume $(2L)^{d}$ from an inductive system of $d$-dimensional many-fermion models (with periodic boundary conditions) using wavelet theory \cite{daubechies1992wavelets}. We start by providing the kinematical setup followed by a discussion of dynamics, the latter, for simplicity, restricted to the case $d=1$.

The directed set $(N,\le)$ is that of the natural number $\NN_{0}$ with the usual ordering $\le$, and we associated with each natural number $n$ a finite dyadic partition $\La_{n} = \vep_{n}\set{-L_{n},...,L_{n}-1}^{d}$ of the $d$-dimensional torus $\T^{d}_{L} = [-L,L)^{d}$, where $\vep_{n} = 2^{-n}\vep_{0}$, $L_{n} = 2^{n}L_{0}$, for some positive integer $L_{0}\in\NN$ such that $\vep_{n}L_{n}=L$. The ordering $\le$ is compatible with the ordering $\subseteq$ of partitions by inclusion.

The algebra $\A_{n}$ is defined as $\A_{n} = \fA_{\CAR}(\fh_{n}))$, the C*-algebra of canonical anti-commutation relations (CAR) over the (one-particle) Hilbert space $\fh_{n} = \ltwo(\La_{n})$ with inner product denoted by $\ip{\placeholder}{\!\placeholder}_{n}$. We recall that this algebra is generated by a single anti-linear operator-valued map, $\fh_{n}\ni\xi\mapsto a(\xi)$, and that $a^{\dag}(\xi) = a(\xi)^{*}$.

The connecting maps $\j nm$ are defined as compositions, $\j nm = \j n{\!\ n-1}\circ...\circ \j {m+1}{\!\ m}$, of quasi-free unital *-homomorphisms, $\j {n+1}{\!\ n} a(\xi) = a(v_{n+1\!\ n}\xi)$, determined by isometries between successive Hilbert spaces, $v_{n+1\!\ n}:\fh_{n}\rightarrow\fh_{n+1}$, which are explicitly given in terms of a finite-length low-pass filter $h_{\alpha}\in\CC$, $\alpha\in\ZZ^{d}$, of an orthonormal, compactly supported scaling function $s\in C^{r}(\RR^{d})$:
\begin{align}
\label{eq:waveletiso}
v_{n+1\!\ n}\xi & = \sum_{x\in\La_{n}}\xi_{x}\sum_{\alpha\in\ZZ^{d}}h_{\alpha}\delta^{(n+1)}_{x+\vep_{n+1}\alpha}, & \xi &\in\fh_{n},
\end{align}
where $\delta^{(n+1)}_{y}$, $y\in\La_{n+1}$, are the standard basis vectors of $\fh_{n+1}$. The low-pass filter and the scaling function are related via the scaling equation:
\begin{align}
\label{eq:seq}
s(x) & = \sum_{\alpha\in\ZZ}h_{\alpha}\!\ 2^{\frac{1}{2}}s(2x-\alpha).
\end{align}
This yields strict inductive systems of Hilbert spaces $(\fh,v)$ and C*-algebras $(\A, j)$ with isometries respectively *-homomorphisms. The morphism property follows from the fact that the algebraic structure of $\A_{n}$ is determined by the inner product of $\fh_{n}$, i.e., we have:
\begin{align}
\label{eq:waveletj}
\{\j nm a(\xi), \j nm a^{\dag}(\eta)\} & = \{a(v_{nm}\xi),a^{\dag}(v_{nm}\eta)\}\1_{n} = \ip{v_{nm}\xi}{v_{nm}\eta}_{n}\1_{n}, & \xi,\eta & \in\fh_{m} \\ \nonumber
& = \ip{\xi}{\eta}_{m}\1_{n} = \j nm(\ip{\xi}{\eta}_{m}\1_{m}) \\ \nonumber
& = \j nm\{a(\xi),a^{\dag}(\eta)\},
\end{align}
where $\{x,y\}=xy+yx$ is the anti-commutator. It is a direct consequence of the theory of wavelets that the inductive system $(\fh,v)$ provides a multiresolution analysis of the space $L^{2}(\T^{d}_{L})$ in the sense of  Mallat and Meyer \cite{mallat1989multiresolution, meyer1989wavelets} based on the scaling function $s$. Precisely, this means that the limit space is $\fh_\oo = L^{2}(\T^{d}_{L})$ and that the closed subspaces $v_{\oo{n}}\fh_{n}$ satisfy,
\begin{align}
\label{eq:multres}
v_{\oo{0}}\fh_{0}\subset...\subset v_{\oo{n}}\fh_{n} & \subset v_{\oo{n+1}}\fh_{n+1}\subset...,
\end{align}
with
\begin{align}
\label{eq:multrescond}
\overline{\bigcup_{n\in\NN_{0}} v_{\oo{n}}\fh_{n}} & = L^{2}(\T^{d}_{L}), & \bigcap_{n\in\NN_{0}} v_{\oo{n}}\fh_{n} & = v_{\oo{0}}\fh_{0} = \CC^{2L_{0}},
\end{align}
and the additional properties:
\begin{itemize}
	\item[1.] $f(2^{n}\!\ \cdot\!\ )\in v_{\oo{n}}\fh_{n} \Leftrightarrow f \in v_{\oo{0}}\fh_{0}$ for $f\in L^{2}(\T^{d}_{L})$,
	\item[2.] $f\in v_{\oo{n}}\fh_{n} \Rightarrow f(\!\ \cdot\!\ -x)\in v_{\oo{n}}\fh_{n}$ for $x\in\La_{n}$,
	\item[3.] the functions $s^{(n)}_{L}(\!\ \cdot\!\ -x) = \sum_{\alpha\in\ZZ}\vep_{n}^{-\frac{1}{2}}s(\vep_{n}^{-1}(\!\ \cdot\!\ -x-\alpha 2L))$ for $x\in\La_{n}$ provide an orthonormal basis of $v_{\oo{n}}\fh_{n}$.
\end{itemize}
We note that $v_{\oo{n}}$ has an explicit representation in terms of the scaling function $s$:
\begin{align}
\label{eq:waveletisoasymp}
(v_{\oo{n}}\xi)(x) & = \sum_{y\in\La_{n}}\xi_{y}\!\ s^{(n)}_{L}(x-y) = (\xi\ast_{\La_{n}}s^{(n)}_{L})(x), & \xi & \in\fh_{n}.
\end{align}
It follows from the functorial properties of the CAR algebra \cite{bratteli2} that
\begin{align}
\label{eq:carlim}
\A_{\infty} & = \overline{\bigcup_{n\in\NN_{0}}\j\oo{n}\A_{n}} = \fA_{\CAR}(L^{2}(\T^{d}_{L})).
 \end{align}
Additionally, the inductive systems $(\fh, v)$ and $(\A,j)$ are split in the sense of \cref{def:split} because the coisometries $v_{\oo{n}}^{*}=p_{{n}\oo}$ define a linear contraction $p_{\blob\oo}:\fh_{\oo}\to\C(\A,j)$, with $v\textup{-lim}_{n}p_{{n}\oo}\xi = \xi$, $\xi\in\fh_\oo$,  and, in turn, we obtain quasi-free completely positive contractions $s_{{n}\oo}:\A_{\infty}\to\A_{n}$ \cite{evans1998qsym} that provide the right inverse of $\jlim$. In particular, we observe that the conditional expectations $\j\oo{n}\circ s_{{n}\oo}=\cE_{n}:\A_{\oo}\to\j\oo{n}\A_{n}\subset\A_{\infty}$, which are induced by the range projections $p_{n} = v_{\oo{n}}\circ p_{{n}\oo}$, converge strongly to the identity.
\begin{lem}
\label{lem:condconv}
Given the above, we have the following:
\begin{align*}
\lim_{n}\norm{\cE_{n}(A)-A} & = 0, & A & \in\A_{\oo}.
\end{align*}
\end{lem}
\begin{proof}
It is sufficient to prove the statement for $A\in\cup_{n\in\NN_{0}}\j\oo{n}\A_{n}$. Thus, given $m$ and $A\in\A_{m}$ we have for all $n\ge m$:
\begin{align*}
\norm{\cE_{n}(\j nm A)-\j\oo{m}A} & = \norm{\j\oo{n}\j nm A - \j\oo m A} = 0,
\end{align*}
because $s_{{n}\oo}\circ\j\oo{m} = \j nm$ by construction.
\end{proof}
Given the inductive systems $(\fh,v)$ and $(\A, j)$, we are in a position to consider scaling limits of many-fermion systems based on the algebras $\A_{n}$. This requires additional data in the form of \jts-convergent nets of states $\omega\d\in\fS(\A\d)$ associated with Hamiltonians $H\d$ that exhibit criticality, i.e., $H_{n}$ should have vanishing spectral gap for some choice of couplings  in the limit $n\to\oo$.
\begin{defin}
\label{def:sl}
A scaling limit associated with the inductive systems $(\A,j)$ is a \jts-con\-vergent net of states $\omega\d\in\fS(\A\d)$. We also call the pair $(\A_{\oo},\omega_{\oo})$ the scaling limit of $(\A\d,\omega\d)$.
\end{defin}
It is at this point that we restrict to the case $d=1$ for concreteness and simplicity, although there are no structural obstacles to discuss models associated with free fermionic quantum fields for arbitrary $d$. For $x\in\La_{n}$ we introduce the abbreviation $a_{x}=a(\delta^{(n)}_{x})$. The net of Hamiltonians acting on the anti-symmetric Fock space $\fF_{-}(\fh_{n})$ is
\begin{align}
\label{eq:H}
H_{n} & = \vep_{n}^{-1}\sum_{x\in\Lambda_{n}}\hspace{-0.2cm}\big(\!(J_{n}\!-\!h_{n})(a_{x}\!+\!a^{\dag}_{x})(a_{x}\!-\!a^{\dag}_{x})\!+\!J_{n}(\!(a_{x+\vep_{n}}\!\!+\!a^{\dag}_{x+\vep_{n}})\!-\!(a_{x}\!+\!a^{\dag}_{x})\!)(a_{x}\!-\!a^{\dag}_{x})\!\big),
\end{align}
where $h_{n},J_{n}>0$ are scale-dependent dimensionless coupling constants. The notation reflects the fact that $J_{n}$ and $h_{n}$ describe the spin-spin coupling and the transverse magnetic field of the associated spin systems, the transverse-field Ising spin chain in the Ramond sector \cite{difrancesco1997cft}, via the Jordan-Wigner transformation \cite{evans1998qsym}. Calculating the spectrum of $H_{n}$ shows that it exhibits criticality for $J_{n}=h_{n}$, and a formal analysis suggests that the scaling limit $n\to\oo$ should provide a field theory of two free Majorana fermions
\begin{align}
\label{eq:majfermion}
\psi_{\pm|x} & = e^{\pm i\frac{\pi}{4}}a_{x}+e^{\mp i\frac{\pi}{4}}a^{\dag}_{x}, & x & \in\T^{1}_{L}.
\end{align}
Intuitively, this is expected due to the following momentum-space representation of $H_{n}$:
\begin{align}
\label{eq:momH}
H_{n} & = \tfrac{J_{n}}{4L}\!\!\sum_{k\in\Gamma_{n}}\!\!\begin{pmatrix} \hat{\psi}_{+|k} \\ \hat{\psi}_{-|k} \end{pmatrix}^{\!\!*}\!\!\underbrace{\begin{pmatrix} -\sin(\vep_{n}k) & \hspace{-1cm}i((\cos(\vep_{n}k)-1) + \lambda_{n}) \\ -i((\cos(\vep_{n}k)-1) + \lambda_{n}) & \hspace{-1cm}\sin(\vep_{n}k) \end{pmatrix}}_{=J_{n}^{-1}h_{n}(k)}\!\! \begin{pmatrix} \hat{\psi}_{+|k} \\ \hat{\psi}_{-|k} \end{pmatrix},
\end{align}
where $\hat{\psi}_{\pm|k} = \psi_{\pm}(e_{k})$ for $e_{k}=e^{ik\!\ \cdot\!\ }\in\fh_{n}$ with $k\in\Gamma_{n}=\tfrac{\pi}{L}\{-L_{n},...,L_{n}-1\}$, and $\lambda_{n}=1-\tfrac{h_{n}}{J_{n}}$ is the dimensionless lattice mass.\\

We obtain a \jts-convergent net of states by considering the ground states $\omega_{n}$ of $H_{n}$ which are quasi-free and, therefore, determined by their two-point function:
\begin{align}
\label{eq:gs2p}
\omega_{n}(\psi_{\pm}(\xi)\psi_{\pm}(\eta)^{*}) & = 2\ip{\xi}{(P_{n})_{\pm\pm}\eta}_{n}, & \omega_{n}(\psi_{\pm}(\xi)\psi_{\mp}(\eta)^{*}) & = 2\ip{\xi}{(P_{n})_{\pm\mp}\eta}_{n},
\end{align}
for $\xi, \eta\in\fh_{n}$. Here, $P_{n}$ is a projection on $\fh_{n}\otimes\CC^{2}$ determined by the momentum-space kernel:
\begin{align}
\label{eq:gs2pkernel}
P_{n}(k) & = \tfrac{1}{2\mu_{n}(k)}\big(\mu_{n}(k)\1_{2}+h_{n}(k)\big), & \mu_{n}(k)^{2} & = (J_{n}-h_{n})^{2}+4J_{n}h_{n}\sin(\tfrac{1}{2}\vep_{n}k)^{2}.
\end{align}
Here, $\mu_{n}(k)^{2}=J_{n}^{2}(\lambda_{n}^{2}+4(1-\lambda_{n})\sin(\tfrac{1}{2}\vep_{n}k)^{2})$ can be recognized as the dispersion relation of a harmonic fermion lattice field of mass $\lambda_{n}$ as long as $J_{n}\ge h_{n}$, which corresponds to the disordered phase of the spin chain associated with $H_{n}$.

Now, it is an immediate consequence of the results on the multiresolution analysis associated with the inductive system $(\fh,v)$ that the net of ground states $\omega\d$ is \jts-convergent provided we impose the RG conditions $\lim_{n}\vep_{n}^{-1}\lambda_{n}=m_{0}\ge0$, $\lim_{n}J_{n}=J>0$ \cite{osborne2022cftapprox}. In particular, this results in the quasi-free state $\omega_{\oo}$ on $\A_{\oo}$ determined by a projection $P_{m_{0}}$ on $L^{2}(\T^{1}_{L})\otimes\CC^{2}$ with momentum-space kernel:
\begin{align}
\label{eq:gs2pkernellim}
P_{m_{0}}(k) & = \tfrac{1}{2\mu_{0}(k)}\big(\mu_{0}(k)\1_{2}-k\sigma_{3}-m_{0}\sigma_{2}\big), & \mu_{0}(k)^{2} & = m_{0}^{2}+k^{2},
\end{align}
where $\sigma_{2},\sigma_{3}$ are the standard Pauli matrices.
\begin{thm}
\label{thm:rglim}
Given the above, and suppose that the coupling constants $J_{n}, h_{n}$ satisfy $\lim_{n}\vep_{n}^{-1}\lambda_{n}=m_{0}\ge0$, $\lim_{n}J_{n}=J>0$. Then, the net of ground states $\omega_{n}$ of $H_{n}$ as in \cref{eq:H} is \jts-convergent, i.e., the RG flow,
\begin{align*}
\lim_{n}\omega^{(m)}_{n} & = \omega^{(m)}_{\infty},
\end{align*}
is convergent for any $m$, and the projective limit of RG-limit states equals the \jts-limit of $\omega\d$: $\varprojlim_{m}\omega^{(m)}_{\infty} = \omega_{\infty}$. Moreover, the scaling limit $(\A_{\oo},\omega_{\oo})$ induces the vacuum representation of two free Majorana fermions of mass $m_{0}$ on the space-time cylinder $\RR\times\T^{1}_{L}$.
\end{thm}
Specifically, in the massless case $m_{0}=0$, the projection $P_{m_{0}}$ becomes diagonal, i.e., the chiral fields $\psi_{\pm}$ decouple, $\omega_{\infty}(\psi_{\pm}(f)\psi_{\mp}(g))=0$ for $f,g\in L^{2}(\T^{1}_{L})$, and their respective ground states $\omega^{(\pm)}_{\infty}$ are determined by the Hardy projections $P_{\pm}(k)=\tfrac{1}{2}(1\mp\sign(k))$, with the convention $\sign(0)=1$.\\

Let us conclude by discussing the convergence of the dynamics given by $H_{n}$. By construction $H_{n}$ is a self-adjoint element of $\A_{n}$ and induces an inner *-derivation $\delta_{n} = i[H_{n},\!\ \cdot\!\ ]$ as well as an automorphism group $\TT_{n}(t) = \Ad_{U_{n}(t)}$ implemented by the unitaries $U_{n}(t)=e^{itH_{n}}$. Some general facts about the convergence of dynamics on C*-algebras and their implementations with respect to *-representations are discussed in \cref{sec:appendix_implementors}. As shown in \cref{prop:strictccomp}, the net of RG-limit states $\omega^{(\blob)}_{\infty}$ provides a compatible net of *-representations $\pi\d$ in the sense of \cref{def:comp} via the GNS construction. Consequently, we obtain an implementation of $\TT\d(t)$ with respect to $\pi\d$ by unitaries $V_{n}(t) = \pi_{n}(U_{n}(t))$ (cp.~\cref{eq:impdyn}):
\begin{align}
\label{eq:impdynunitary}
\pi_{n}(\TT_{n}(t)A_{n}) & = V_{n}(t)\pi_{n}(A_{n})V_{n}(t)^{*}, & A & \in\A_{n}.
\end{align}
The quasi-free structure implies that the $\TT_{n}(t)$ is given in terms of the one-particle Hamiltonian $h_{n}\in\LL(\fh_{n}\otimes\CC^{2})$ with momentum-space kernel $h_{n}(k)$ defined in \cref{eq:H}:
\begin{align}
\label{eq:1pdyn}
\TT_{n}(t)\psi(\xi_{+},\xi_{-}) & = \psi(e^{-2ith_{n}}(\xi_{+},\xi_{-})), & \psi(\xi_{+},\xi_{-}) & = \psi_{+}(\xi_{+})+\psi_{-}(\xi_{-}).
\end{align}
The \jt\jt-convergence of $\TT\d(t)$ follows from the strong $vv$-convergence of $e^{ith\d}$ using the results of \cite{osborne2022cftapprox} and the following estimate on basic sequences,
\begin{align}
\norm{(\j nm\TT_{m}(t)\j ml-\TT_{n}(t)\j nl)\psi(\xi_{+},\xi_{-})}  \leq \norm{(v_{nm}e^{-2ith_{m}}v_{ml}-e^{-2ith_{n}}v_{nl})(\xi_{+},\xi_{-})},
\end{align}
for $(\xi_{+},\xi_{-})\in\fh_{l}\otimes\CC^{2}$ since $\psi$ generates $\A_{l}$. In particular, it follows using \cite{osborne2022cftapprox} that $e^{ith\d}$ is $vv$-convergent to the dynamics $e^{ith_{\oo}}$ on $\fh_{\oo}$ with one-particle Hamiltonian $h_{\oo}$ determined by the momentum-space kernel:
\begin{align}
\label{eq:1pdynlim}
h_{\oo}(k) & = -k\sigma_{3}-m_{0}\sigma_{2}.
\end{align}
Finally, we would like to conclude the $JJ$-convergence of the implementing unitaries $V_{n}(t)$, where $JJ$-convergence convergence refers to the connecting maps $J_{nm}$ induced by the $\j nm$ via the GNS construction (cf.~\cref{prop:strictccomp}). This can be achieved for slightly modified $V_{n}(t)$ by \cref{prop:strictcimp} as the vacuum state $\omega_{\oo}$ given by \cref{eq:gs2pkernellim} is invariant under the quasi-free dynamics generated by $h_{\oo}$. The necessary modification is due to the fact that $H_{n}$ has a divergent vacuum expectation value with respect to $\omega^{(n)}_{\oo}$, i.e., $\lim_{n}\omega^{(n)}_{\oo}(H_{n})=\oo$, which needs to subtracted:
\begin{align}
\label{eq:impdynunitarymod}
\tilde{V}_{n}(t) & = e^{it\!\; :\!\:\pi_{n}(H_{n})\!\: :} = e^{it(\pi_{n}(H_{n})-\omega^{(n)}_{\oo}(H_{n}))}.
\end{align}
The modified implementors satisfy $\norm{(\tilde{V}_{n}(t)-\1)\Omega^{(n)}_{\infty}}=0$, where $\Omega^{(n)}_{\infty}$ is the GNS vector $\omega^{(n)}_{\infty}$, and, therefore, converge by \cref{prop:strictcimp}.

\subsection{Recovering symmetries: Thompson's group actions à la Jones}
\label{sec:thompson}
The question of convergence of dynamics in the setting of inductive systems naturally extends to that of symmetries or symmetry groups. Specifically, in the context of quantum scaling limits of \cref{sec:quantscale}, we can ask how spacetime symmetries of a continuum QFT are approximated via lattice discretizations.
In the setting of Wilson-Kadanoff RG, a distinguished role is played by the fixed points of the RG, which are generically expected to be associated with conformal field theories (CFTs), which poses a particularly large, i.e., infinite-dimensional, symmetry group in $1+1$ dimensions \cite{belavin1984cft, polyakov1984cft}.
Within the class of models considered in the previous subsection, the conformal symmetry group corresponds to the orientation-preserving diffeomorphisms $\textup{Diff}_{+}(\T^{1}_{L})$ of the spatial circle $\T^{1}_{L}$. It has been proposed by Jones to choose an approximation in terms of piecewise-linear homeomorphisms of the unit interval with dyadic-rational breakpoints and slopes of powers of $2$ \cite{jones2017unitary,jones2018nogo}, i.e., Thompson's groups $F\subset T$ (the larger including dyadic rotations), as these allow for a natural action on (strict) inductive limits over dyadic partitions $\Lambda_{n}\subset\T^{1}_{L}$ (using the notation of \cref{sec:quantscale}), see \cite{brothier2019gauge,osborne2019thompson,kliesch2020continuum,brothier2023haagerup} for further results.

In the framework discussed here, the action of an element $f\in F$ of Thompson's group $F$, rescaled to a map $f_{L}:\T^{1}_{L}\to\T^{1}_{L}$\footnote{We define $f_{L}:=2L f(\tfrac{1}{2L}(\placeholder+L)\!)-L$.}, on an inductive system over dyadic partitions $\Lambda_{n}$ can be understood as an instance of a convergent net of operations as in \cref{sec:operations}: Consider $f$ as a map between two incomplete dyadic partitions of the unit interval (see \cref{fig:treeaction} for an illustration) and, thus, $f_{L}$ as a map between two incomplete dyadic partitions $\Lambda,\Lambda'\subset \T^{1}_{L}$.
If the connecting maps $\j nm$ of the inductive system under consideration can be generated from a single map $j^{(0)}_{21}$ that implements the local refinement $\Lambda_{n}\subset\{x\}\to\{x,x+\vep_{n+1}\}\subset\Lambda_{n+1}$, we can consistently define the action of $f_{L}$ on the inductive system by shifting the local indices of elements $a\in\A_{n}$ such that $\Lambda,\Lambda'\subset\Lambda_{n}$.
\begin{figure}[ht!]
    \centering
	\begin{tikzpicture}
	
	\draw (2,2) node{$f$};
	
	\draw (3,2) node{$\simeq$};
	
	\draw (5,-0.1) to (5,0.1) (6,-0.1) to (6,0.1) (7,-0.1) to (7,0.1) (8,-0.1) to (8,0.1);
	\draw (5,0) node[below]{$\tfrac{1}{4}$} (6,0) node[below]{$\tfrac{1}{2}$} (7,0) node[below]{$\tfrac{3}{4}$} (8,0) node[below]{$1$};
	\draw[thick] (4,0) to (8,0);
	\draw (3.9,0) to (4.1,0) (3.9,1) to (4.1,1) (3.9,2) to (4.1,2) (3.9,3) to (4.1,3) (3.9,4) to (4.1,4);
	\draw (4,0) node[below]{$0$} (4,1) node[left]{$\frac{1}{4}$} (4,2) node[left]{$\frac{1}{2}$} (4,3) node[left]{$\frac{3}{4}$} (4,4) node[left]{$1$};
	\draw[thick] (4,0) to (4,4);

	\draw (4,0) to (5,2) to (6,3) to (8,4);
		
	\end{tikzpicture}
	\caption{An example of a Thompson's group element $f\in F$ as a map between two incomplete dyadic partitions.}
\label{fig:treeaction}
\end{figure}
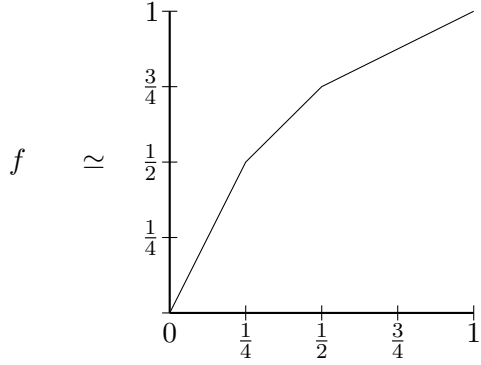

An explicit example is given by the connecting maps $\j nm$ between CAR algebras $\A_{n}=\fA_{\CAR}(\fh_{n})$ based on the Haar wavelet, i.e., \cref{eq:waveletiso} together with the specific choice $h_{\alpha}=\tfrac{1}{\sqrt{2}}(\delta_{\alpha,0}+\delta_{\alpha,1})$ corresponding to the scaling function $s = \chi_{[0,1)}$ (the indicator function of the unit interval).
By construction, the connecting maps are generated by the single map $j^{(0)}_{21}:\fA_{\CAR}(\CC_{x})\to\fA_{\CAR}(\CC_{x}\oplus\CC_{x+\vep_{n+1}})$, $a_{x}\mapsto\tfrac{1}{\sqrt{2}}(a_{x}+a_{x+\vep_{n+1}})$, for $x\in\Lambda_{n}$. Now, an element $f\in F$ acts on a basic sequence of the form $\j\blob k(a(\xi_{k}))$, $\xi_{k}\in\fh_{k}$, by:
\begin{align}
\label{eq:jonesaction}
f_{L}\cdot\j\blob k(a(\xi_{k})\!) & = \j\blob n(f_{L}\cdot\j nk(a(\xi_{k})\!)\!) = \j\blob n(f_{L}\cdot a(v_{nk}\xi_{k})\!) = \j\blob n\Big(\sum_{x\in\Lambda_{n}}(v_{nk}\xi_{k})_{x}\!\ a_{f_{L}(x)}\Big),
\end{align}
for any $n$ such $\Lambda,\Lambda'\subset\Lambda_{n}$. In particular, we have the following automorphic action of $F$ (and $T$) on $\A_{\oo} = \fA_{\CAR}(L^{2}(\T^{1}_{L}))$ determined on the dense *-subalgebra $\bigcup_{k\in\NN_{0}}\j\oo{k}\A_{n}$ by
\begin{align}
\label{eq:jonesactionlim}
f_{L}\cdot\j\oo k(a(\xi_{k})\!) & = a(f_{L}\cdot v_{\oo k}\xi_{k}) = \sum_{x\in\Lambda_{n}}(v_{nk}\xi_{k})_{x}\!\ a(\chi^{(n)}_{[0,1)|L}(\placeholder -f_{L}(x)\!)\!),
\end{align}
where $f_{L}\cdot v_{\oo k}\xi_{k} = \sum_{x\in\Lambda_{n}}(v_{nk}\xi_{k})_{x}\!\ \chi^{(n)}_{[0,1)|L}(\placeholder -f_{L}(x)\!)$, since the bound $\norm{f_{L}\cdot\j\oo k(a(\xi_{k})\!)}=\norm{\xi_{k}}_{\fh_{k}}$. \cref{eq:jonesactionlim} shows that this Jones action is the second-quantized (or quasi-free) form of a Pythagorean representation of $F$ (respectively $T$) studied in \cite{brothier2019pythagorean}, and we have:
\begin{align}
\label{eq:jonesactionlimequiv}
(f_{L}\cdot v_{\oo k}\xi_{k})(x) & = \sum_{y\in\Lambda_{n}}(v_{nk}\xi_{k})_{y}\!\ \chi^{(n)}_{[0,1)|L}(x-f_{L}(y)\!) = |(f_{L}^{-1})'(x)|^{\frac{1}{2}}(v_{\oo k}\xi_{k})(f_{L}^{-1}(x)\!)
\end{align}
which may directly be compared with the automorphic quasi-free action of $\textup{Diff}_{+}(\T^{1}_{L})$ on $\A_{\infty}$ induced by its unitary action on $L^{2}(\T^{1}_{L})$:
\begin{align}
\label{eq:diffaction}
(u_{\phi}\xi)(x) & = |(\phi^{-1})'(x)|^{\frac{1}{2}} \xi(\phi^{-1}(x)\!),
\end{align}
for $\phi\in\textup{Diff}_{+}(\T^{1}_{L})$ and $\xi\in L^{2}(\T^{1}_{L})$. On basic sequences of the form as above, such comparison boils down to a one-particle space estimate:
\begin{align}
\label{eq:1pestimate}
\norm{\phi\cdot\j\oo k(a(\xi_{k}) - f_{L}\cdot\j\oo k(a(\xi_{k})} & = \norm{u_{\phi}\cdot v_{\oo k}\xi_{k}-f_{L}\cdot v_{\oo k}\xi_{k}}_{L^{2}},
\end{align}
 Thus, the action of general conformal transformation can be strongly approximated on $\A_{\oo}$ if the associated action on the one-particle limit space $\fh_{\oo} = L^{2}(\T^{1}_{L})$ can be strongly approximated.

Finally, we note that such an approximation is typically not possible in terms of implementors (see \cref{sec:appendix_implementors}) as the action of non-trivial elements of Thompson's groups is not implementable in the presence of a non-vanishing central charge \cite{delVecchio2019solitons}. In the latter situations, a different strategy invoking the approximation of conformal symmetries by the Koo-Saleur formula \cite{koo1994virasoro} proves successful \cite{osborne2022cftapprox}.

\section{Comparison with the literature}\label{sec:comparison}

We compare our work with three related works:
The concept of generalized inductive limits of C*-algebras introduced by Blackadar and Kirchberg in \cite{blackadar1997generalized}, the concept of continuous fields of Banach spaces or C*-algebras \cite{dixmier1982} over the topological space $N\cup\set\oo$, and finally, an abstract approach due to Kurtz \cite{kurtz1970} that assumes only a set of convergent nets and also covers an evolution theorem.
While \cite{blackadar1997generalized} is concerned with generalizations of the notion of inductive limit, \cite{dixmier1982, kurtz1970} generalize to a setting where one has a certain notion of convergence already given.
We will see that our setup of \jt-convergent nets and limit space always defines the structures studied in the latter two articles.

In \cite{blackadar1997generalized}, Blackadar and Kirchberg generalize inductive limits of C*-algebras by relaxing almost all properties (e.g., linearity, multiplicativity) of the connecting maps to asymptotic versions with the exception of the strict transitivity property $\j nl = \j nm\circ \j ml$ if $n>m>l$ which is required to hold always.
They often specialize to the case where the connecting maps are completely positive linear contractions that are asymptotically multiplicative (in the same sense as in \cref{eq:asymmor}).
While they discuss a notion similar to what we call \jt\jt-convergence, they do not seem to be interested in the convergence of dynamics.
Instead, they apply their setup to discuss different classes of C*-algebras that arise as (generalized) inductive limits of finite-dimensional algebras.
Our discussion of soft inductive limits shows that the strict transitivity of the connecting maps is not essential, and it would be interesting to discuss this in the context of finite-dimensional approximations of C*-algebras.
One should try to answer the question: Is the class of C*-algebras that arise as soft inductive limits of finite-dimensional C*-algebras really larger?
We also mention recent works \cite{courtney2023nuclearity,courtney2023completely} building on and further generalizing the work of Blackadar and Kirchberg.

The concept of continuous fields of C*-algebras is often used for studying limiting phenomena such as the classical and mean-field limit \cite{rieffel1990deformation, Drago2022}.
For this, one takes as the topological space $\Bar N = N\cup \set\oo$ equipped with the order topology, i.e., the topology generated by order intervals $(n,\oo]$, $n\in N$, where $N$ is is the directed set (typically $((0,1],\ge)$ or $(\NN,\le)$).
In a sense, the idea is to generalize the fact that a net $(x_\lambda)_{\lambda\in\Lambda}$ in a topological space $X$ converges to a point $x_\oo\in X$ if and only if the function $\Bar\Lambda\in\lambda\mapsto x_\lambda\in X$ is continuous ($\Bar\Lambda=\Lambda\cup\set\oo$ is again equipped with the order topology).
A continuous field of Banach spaces \cite[Ch.~10]{dixmier1982} over a topological space $T$ is a collection of Banach spaces $E_t$, $t\in T$, together with a specified subspace $\Gamma\subset\prod_{t\in T} E_t$ of "continuous vector fields" (which we call "convergent nets" if $T=\Bar N$).
This subspace is assumed to satisfy the axioms (a) $\norm{x_t}$ is continuous in $t\in T$, (b) the set of $x_s$ with $(x_t)\in\Gamma$ is dense in $E_s$ for all $s\in T$ and (c) if $(x_t)\in\prod_{t\in T} E_t$ is such that for every $\eps>0$ and every $t_0\in T$ there is a $(y_t)\in\Gamma$ such that $\norm{x_s-y_s}<\eps$ for all $s$ in a neighborhood of $t_0$, then $(x_t)\in\Gamma$.
Using that the order topology is discrete except at $\oo$ and introducing the seminorm $\seminorm{x\d}=\limsup_{n\in N}\norm{x_n}$ on $\prod_{n\in\Bar N} E_n$ lets us rewrite the axioms for $X=\Bar N$ as:
\begin{enumerate}[(a)]
    \item $\norm{x_\oo}=\lim_n \norm{x_n} =\seminorm{x\d}$ for all $x\d\in\Gamma$,
    \item The set $\set{x_n\given x\d\in\Gamma}$ is dense in $E_n$ for all $n\in\Bar N$,
    \item $\Gamma$ is a seminorm-closed subspace of $\prod_{n\in\Bar N} E_n$.
\end{enumerate}
Since the sup-norm dominates the seminorm on $\prod_{n\in\Bar N}E_n$, it follows from (a) that $\Gamma$ is also a norm closed subspace.
For a continuous field of C*-algebras, one requires that all $E_n$ are C*-algebras and that
\begin{enumerate}[resume*]
    \item
        $\Gamma$ is closed under multiplication and involution.
\end{enumerate}
It is now clear that the (soft) inductive system structure implies that of a continuous field over $\Bar N$:

\begin{lem}
    Let $(E,j)$ be a (soft) inductive system of Banach spaces (resp.\ C*-algebras) over a directed set $N$.
    Then one obtains a continuous field of Banach spaces (resp.\ C*-algebras) over $\Bar N$ by adding the limit space $E_\oo$ and by defining $\Gamma$ to be the collection of \jt-convergent nets with $x_\oo =\jlim_n x_n$.
\end{lem}

In \cite{kurtz1970}, the author works in an abstract setting similar to continuous fields over $\Bar N$, and sufficient conditions for convergence of dynamics are considered.
Given a net of Banach spaces, the author considers the product $\nets = \prod_{n\in N}E_n$ with the sup-norm (we use similar notations as in our work, not the one from \cite{kurtz1970}).
He now assumes a closed subspace $\C\subset\nets$ of "convergent nets" and a bounded linear operator, $\LIM: \C\to E_\oo$, to be given.
If $\LIM$ is surjective then $E_\oo \cong \C /\C_0$ with $\C_0=\ker(\LIM)$ just as in our construction.

If we assume that $(\set{E_n},\Gamma,\Bar N)$ is a continuous field of Banach spaces, then the above setting is implied:
It follows from linearity and axiom (a) that for each $x\d\in\Gamma$ the element $x_\oo$ is uniquely determined by the elements $(x_n)_{n\in N}$%
\footnote{To see this let $x\d,y\d\in\Gamma$ and assume that $x_n = y_n$ for all $n\in N$. Then $\norm{x_\oo-y_\oo}=\seminorm{x\d-y\d}=0$ and hence $x_\oo=y_\oo$.}.
We thus obtain a closed subspace $\C=\set{(x_n)_{n\in N}\given x\d\in \Gamma}\subset \nets = \prod_{n\in N}E_n$ of convergent nets and a well-defined linear operator $\LIM : \C \to E_\oo$ mapping $x\d$ to the unique $\LIM_nx_n = x_\oo$ so that $(x\d,x_\oo)\in \Gamma$.
This is, however, always implied if the continuous field arises from a (soft) inductive system, as discussed above.

In this setting, our notion of \jt\jt-convergence of a net of contractions $T\d$, $T_n\in \LL(E_n)$, may be generalized by requiring $T\d:\C\subset\C$ and $T\d\C_0\subset \C_0$ which is sufficient for the definition of a limit operator $T_\oo$.
Using this notion, one can discuss the convergence of dynamics.
As already said, \cite{kurtz1970} also covers an evolution theorem.
This theorem is similar to the direction \ref{it:net_core} $\Rightarrow$ \ref{it:generator} of our \cref{thm:evolution} and allows for more flexibility because of the general setup.
For instance, his theorem can be directly applied even if \jt\jt-convergence holds only on a closed subspace of $\C(E,j)$ (and then defines dynamics on a subspace of $E_\oo$).
In \cite{kurtz1972}, this evolution theorem and the abstract approach are used to prove a neat probabilistic generalization of the Lie-Trotter product formula using stochastic processes.
\\[0pt]

The comparison with the latter two approaches has shown that the connecting maps are not too essential for our analysis.
The connecting maps become important as they provide a way to check properties of the often intractable class of convergent nets (e.g., $\C(\A,j)$ is closed under products).
Ultimately this works because of the seminorm-density of basic sequences.
A downside of the latter two approaches for describing limit phenomena such as those considered in \cref{sec:ex} is that one needs to know  the limit space and notion of convergence are not derived notions but have to be defined from the outset.

\subsection*{Acknowledgments}

Lauritz van Luijk thanks Niklas Galke for helpful discussions and comments.

\appendix

\section{Convergence of implementors for dynamics in representations}\label{sec:appendix_implementors}

In applications, one often has explicit representations of the C*-algebras at hand.
The Hilbert spaces on which the algebras are represented often form a separate inductive system, i.e.\ we are dealing with an inductive system of C*-algebras and an inductive system of Hilbert space connected via a net of representations.
In such a setting, it is often of interest to understand not only the convergence of dynamics on the algebras but also of its implementors.

For a soft inductive system $(\H,J)$ of Hilbert spaces we ask that inner products $\ip{\psi_n}{\phi_n}$ converge for all $J$-convergent $\psi\d,\phi\d\in\C(\H,j)$.
Evaluating this condition on basic nets, one sees that it is equivalent to the convergence of $J_{nm}^*J_{nm}$ in the weak operator topology of $\LL(\H_l,\H_m)$ as $n\to\oo$ for sufficiently large $m\in N$, and holds automatically if the connecting maps are isometries.
In this case the limit space $\H_\oo$ is a Hilbert space with the inner product $\ip{\psi_\oo}{\phi_\oo}=\lim_n \ip{\psi_n}{\phi_n}$.
For any $J$-convergent net $\psi\d$, the "bras" $\bra{\psi\d}$ are $J^*$-convergent and converge to $\bra{\psi_\oo}$, where we use Dirac notation $\bra{\psi}=\ip\psi{\placeholder}$.

\begin{defin}
\label{def:comp}
    Let $(\H,J)$ be a soft inductive system of Hilbert spaces and let $(\A,j)$ be a soft inductive system of C*-algebras (over the same directed set).
    A net $\pi\d$ of *-representations $\pi_n:\A_n\to\LL(\H_n)$ is {\bf compatible}, if for any $a\d\in\C(\A,j)$ and $\psi\d\in\C(j,J)$ the net vectors $\pi\d(a\d)\psi\d = (\pi_n(a_n)\psi_n)_n$ is $J$-convergent\footnote{In the sense of the previous sections, $\pi\d$ maps \jt-convergent nets to $JJ$-convergent nets.}.
    In this case one obtains a *-representation $\pi_\oo$ of $\A_\oo$ on $\H_\oo$ by assigning to $a_\oo$ the limit operation of $\pi\d(a\d)$.
\end{defin}

A natural source of compatible nets of representations of inductive systems of Hilbert space are projective nets of states on strict inductive systems of C*-algebras $(\A, j)$ \footnote{This is not, in general, true for soft C*-inductive systems.}:
\begin{prop}\label{prop:strictccomp}
    Given a strict inductive system of C*-algebras $(\A,j)$ and a state $\omega_\oo$ on $\A_\oo$, we consider the net of states $\omega_n = \omega_{\infty}\circ \j\oo n$. Then, a compatible net of representations $\pi\d$ exists on a strict inductive system of Hilbert spaces $(\H,J)$ induced by $\omega\d$.
    Moreover, each $\H_{n}$ is cyclic for $\pi_{n}(\A_n)$ and there is a $J$-convergent net of cyclic unit vectors $\Omega\d\in\C(\H,J)$ that implements $\omega\d$:
    \begin{align}\label{eq:strictccompgns}
        \omega_{n}(a_{n}) & = \langle\Omega_{n}, \pi_{n}(a_{n})\Omega_{n} \rangle & \forall n, a_{n}\in\A_n.
    \end{align}
\end{prop}

\begin{proof}
    We apply the GNS construction to the net of states $\omega_{\infty\blob}$ which yields triples $(\H\d, \pi\d, \Omega\d)$ with each $\Omega_{n}$ being cyclic for $\pi_{n}(\A_n)$. By construction, we have \eqref{eq:strictccompgns} and, by strictness, we observe:
    \begin{align*}
        \omega_{n} \circ \j nm & = \omega_{m} & \forall m\leq n.
    \end{align*}
    Therefore, we can define linear maps $J_{nm}:\H_{m}\rightarrow\H_{n}$ by\footnote{This definition enforces the projective consistency condition, $\omega_{n} \circ \j nm = \omega_{m}$, because it entails $J_{nm}\Omega_{m} = \Omega_{n}$ and $J_{nm}^{*}\Omega_{n}=\Omega_{m}$.}:
    \begin{align*}
        J_{nm}\pi_{m}(a_{m})\Omega_{m} & = \pi_{n}(\j nm (a_{m}))\Omega_{n} & \forall m\leq n,
    \end{align*}
    which entails $J_{nm}J_{ml} = J_{nl}$ and $J_{nm}\Omega_{m} = \Omega_{n}$. Each $J_{nm}$ is linear by the linearity of $\pi_{n}$ and $\j nm$, and it is a contraction because of Kadison's inequality for completely positive contractions:
    \begin{align*}
        \norm{J_{nm}\pi_{m}(a_{m})\Omega_{m}}^{2}
        & = \langle J_{nm}\pi_{m}(a_{m})\Omega_{m}, J_{nm}\pi_{m}(a_{m})\Omega_{m}\rangle \\
        & = \langle\pi_{n}(\j nm (a_{m}))\Omega_{n}, \pi_{n}(\j nm (a_{m}))\Omega_{n} \rangle = \langle\Omega_{n}, \pi_{n}(\j nm (a_{m}^{*})\j nm (a_{m}))\Omega_{n}\rangle \\ 
        & \leq \langle\Omega_{n}, \pi_{n}(\j nm (a_{m}^{*}a_{m}))\Omega_{n}\rangle = \omega_{n}(\j nm (a_{m}^{*}a_{m})) \\
        & = \omega_{m}(a_{m}^{*}a_{m}) = \norm{\pi_{m}(a_{m})\Omega_{m}}^{2}.
    \end{align*}
    The convergence of scalar products is implied by the asymptotic morphisms property \eqref{eq:asymmor} of the connecting maps $j$:
    \begin{align*}
        & |\ip{J_{nk}\pi_{k}(a_{k})\Omega_{k}}{J_{nl}\pi_{l}(b_{l})\Omega_{l}}-\ip{J_{mk}\pi_{k}(a_{k})\Omega_{k}}{J_{ml}\pi_{l}(b_{l})\Omega_{l}}| \\
        & = |\omega_{n}(j_{nk}(a_{k}^{*})j_{nl}(b_{l}))-\omega_{m}(j_{mk}(a_{k}^{*})j_{ml}(b_{l}))| \\
        & \leq |\omega_{\infty}(j_{\infty n}(j_{nk}(a_{k}^{*})j_{nl}(b_{l})\!)\!-\!j_{\infty k}(a_{k}^{*})j_{\infty l}(b_{l})\!)|\!+\!|\omega_{\infty}(j_{\infty m}(j_{mk}(a_{k}^{*})j_{ml}(b_{l})\!)\!-\!j_{\infty k}(a_{k}^{*})j_{\infty l}(b_{l})\!)| \\
        & \xrightarrow{n,m\rightarrow\infty}0
    \end{align*}
    for all $a_{k}\in\A_{k}$ and $b_{l}\in\A_{l}$, which implies the results because each $\Omega_{n}$ is cyclic.
\end{proof}

\begin{rem}\label{rem:strictccom}
    We will obtain isometries as connecting maps for the compatible system of Hilbert spaces if the connecting maps are *-homomorphisms.
    Moreover, the compatibility between the connecting maps $j$ and $J$ becomes independent of the net of GNS vectors $\Omega\d$:
    \begin{align} \label{eq:strictmorcond}
        \pi_{n}(\j nm (a_{m}))J_{nm} & = J_{nm}\pi_{m}(a_{m}).
    \end{align}
\end{rem}
\vspace{11pt}

Now, let us assume that we are given a net of endomorphism semigroups $\TT\d(t)$ on a strict inductive system $(\A,j)$ of C*-algebras, consisting of unital *-endomorphisms $\TT_n(t)$ on each $\A_n$, together with a compatible net of representations $\pi\d$ on a strict inductive system of Hilbert spaces $(\H, J)$.
We assume further that each $\TT_{n}(t)$ is implemented by a strongly continuous semigroup of bounded operators $V_{n}(t)\in\LL(\H_{n})$ with $\sup_{n,t}\norm{V_{n}(t)}:=M<\infty$ in the sense that (cp.~\eqref{eq:strictmorcond}):
\begin{align}
\label{eq:impdyn}
V_{n}(t)\pi_{n}(a_{n}) & = \pi_{n}(\TT_{n}(t)(a_{n}))V_{n}(t).
\end{align}
\begin{rem}\label{rem:impdyn}
For any two implementing semigroups $V_{n}(t)$, $W_{n}(t)$, it follows from \eqref{eq:impdyn} that $V_{n}(t)^{*}W_{n}(t)\in\pi_{n}(\A_{n})'$ and $V_{n}(t)W_{n}(t)^{*}\in\pi_{n}(\TT_{n}(t)(\A_{n}))'$. The operators $V_{n}(t)$ are called units of $\TT(t)$ in the context of $E_{0}$-semigroups \cite{arveson2003ncdynamics}.\\
In accordance with Prop.~\ref{prop:strictccomp}, implementing semigroups $V_{n}(t)$ can be obtained from a net, $\omega\d=\omega_\oo\circ \j\oo \blob$, of $\TT\d(t)$-invariant states by defining $V_{n}(t)$ via:
\begin{align*}
V_{n}(t)\pi_{n}(a_{n})\Omega_{n} & = \pi_{n}(\TT_{n}(t)(a_{n}))\Omega_{n},
\end{align*}
which enforces $V_{n}(t)^{*}V_{n}(t)=\1$ and the invariance of GNS vectors, i.e., $V_{n}(t)\Omega_{n} = \Omega_n = V_{n}(t)^{*}\Omega_{n}$, which follows from $\TT_{n}(t)(\1_{n})=\1_{n}$ and $\ip{V_{n}(t)^{*}\Omega_{n}}{\pi_{n}(a_{n})\Omega_{n}} = \ip{\Omega_{n}}{\pi_{n}(a_{n})\Omega_{n}}$. The strong continuity of $V_{n}(t)$ follows from:
\begin{align*}
\norm{(V_{n}(t)-\1)\pi_{n}(a_{n})\Omega_{n}} & \leq \norm{(\TT_{n}(t)-\1)a_{n}} & \forall a_{n}&\in\A_{n}.
\end{align*}
\end{rem}

\begin{lem} \label{lem:strictcimp}
Let $\TT\d(t)$ be a net of dynamical semigroups on $(\A,j)$ and $V\d(t)$ be an implementing net of semigroups on $(\H, J)$ with respect to a compatible net of representations $\pi\d$ induced by a state $\omega_{\oo}$. Then, we have the following estimates for any $a_{k}\in\A_{k}$:
\begin{itemize}
	\item[1.]
	\begin{align} \label{eq:noninvstateest}
        & \norm{J_{nm}V_{m}(t)\pi_{m}(\j mk(a_{k})\!)\Omega_{m}\!-\!V_{n}(t)\pi_{n}(\j nk(a_{k})\!)\Omega_{n}} \\ \nonumber
        & \leq \norm{\j nm(\TT_{m}(t)(\j mk(a_{k})\!)\!)\!-\!\TT_{n}(t)(\j nk(a_{k})\!)}\!+\!\norm{a_{k}}(\norm{(V_{n}(t)\!-\!\1)\Omega_{n}}\!+\!\norm{(V_{m}(t)\!-\!\1)\Omega_{n}}).
	\end{align}
	If each $\omega_{n}$ is $\TT_{n}(t)$-invariant and $V_{n}(t)$ is induced by $\omega_{n}$, this reduces to:
	\begin{align} \label{eq:invstateest}
        & \norm{J_{nm}V_{m}(t)\pi_{m}(\j mk(a_{k})\!)\Omega_{m}\!-\!V_{n}(t)\pi_{n}(\j nk(a_{k})\!)\Omega_{n}} \\ \nonumber
        & \leq \norm{\j nm(\TT_{m}(t)(\j mk(a_{k})\!)\!)\!-\!\TT_{n}(t)(\j nk(a_{k})\!)}.
	\end{align}
	\item[2.] If the connecting maps of $(\A,j)$ are *-homomorphisms and $\TT\d(t)$ is implemented according to \eqref{eq:impdyn}, we have:
	\begin{align} \label{eq:morstateest}
        & \norm{J_{nm}V_{m}(t)\pi_{m}(\j mk(a_{k})\!)\Omega_{m}\!-\!V_{n}(t)\pi_{n}(\j nk(a_{k})\!)\Omega_{n}} \\ \nonumber
        & \leq \norm{\j nm(\TT_{m}(t)(\j mk(a_{k})\!)\!)\!-\!\TT_{n}(t)(\j nk(a_{k})\!)}\!+\!\norm{a_{k}}\norm{J_{nm}V_{m}(t)\Omega_{m}\!-\!V_{n}(t)\Omega_{n}}.
	\end{align}
\end{itemize}

\end{lem}
\begin{proof}
The first inequality follows from:
\begin{align*}
& \norm{J_{nm}V_{m}(t)\pi_{m}(\j mk(a_{k})\!)\Omega_{m}\!-\!V_{n}(t)\pi_{n}(\j nk(a_{k})\!)\Omega_{n}} \\
& = \norm{J_{nm}\pi_{m}(\TT_{m}(t)(\j mk(a_{k})\!)\!)V_{m}(t)\Omega_{m}\!-\!\pi_{n}(\TT_{n}(t)(\j nk(a_{k})\!)\!)V_{n}(t)\Omega_{n}} \\
& \leq \norm{\pi_{n}(\j nm(\TT_{m}(t)(\j mk(a_{k})\!)\!)\!)\Omega_{m}\!-\!\pi_{n}(\TT_{n}(t)(\j nk(a_{k})\!)\!)\Omega_{n}} \\
& \hspace{0.5cm}+\norm{J_{nm}\pi_{m}(\TT_{m}(t)(\j mk(a_{k})\!)\!)(V_{m}(t)-\1)\Omega_{m}}+\norm{\pi_{n}(\TT_{n}(t)(\j nk(a_{k})\!)\!)(V_{n}(t)-\1)\Omega_{n}},
\end{align*}
using the uniform bounds on $J_{\blob\blob}$, $\pi\d$, $\j\blob\blob$ and $\TT\d(t)$. The second inequality follows immediately because the $V_{n}(t)\Omega_{n}=\Omega_{n}$. The third inequality follows from \eqref{eq:strictmorcond} which gives:
\begin{align*}
& \norm{J_{nm}V_{m}(t)\pi_{m}(\j mk(a_{k})\!)\Omega_{m}\!-\!V_{n}(t)\pi_{n}(\j nk(a_{k})\!)\Omega_{n}} \\
& = \norm{\pi_{m}(\j nm(\TT_{m}(t)(\j mk(a_{k})\!)\!)\!)J_{nm}V_{m}(t)\Omega_{m}\!-\!\pi_{n}(\TT_{n}(t)(\j nk(a_{k})\!)\!)V_{n}(t)\Omega_{n}} \\
& \leq \norm{\pi_{n}(\j nm(\TT_{m}(t)(\j mk(a_{k})\!)\!)\!-\!\TT_{n}(t)(\j nk(a_{k})\!)\!)J_{nm}V_{m}(t)\Omega_{m}} \\
&\hspace{0.5cm}+\norm{\pi_{n}(\TT_{n}(t)(\j nk(a_{k})\!)\!)(J_{nm}V_{m}(t)\Omega_{m}\!-\!V_{n}(t)\Omega_{n})}.
\end{align*}
\end{proof}

Thus, if a dynamical semigroup of endomorphisms $\TT\d(t)$ is convergent in the sense of \cref{thm:evolution}, we can deduce the convergence of the implementing semigroups $V\d(t)$ with respect to $(\H,J)$, if the latter is induced by a net of (asymptotically) $\TT\d(t)$-invariant states.

\begin{prop}[Convergence of implementors for (asymptotically) invariant states]
\label{prop:strictcimp}
Let $\TT\d(t)$ be a net of dynamical semigroups on $(\A,j)$ and $V\d(t)$ be an implementing net of semigroups on $(\H, J)$ with respect to a compatible net of representations $\pi\d$ induced by a state $\omega_{\oo}$.\\
Then, the implementing semigroup $V\d(t)$ is $JJ$-convergent if $\TT\d(t)$ is \jt\jt-convergent and the $J$-convergent sequence of unit vectors $\Omega\d$ implementing $\omega\d$ is asymptotically $V\d(t)$-invariant, i.e., $\lim_{n}\norm{(V_{n}(t)-\1)\Omega_{n}}=0$.
\end{prop}

In situations where no $\TT\d(t)$-invariance is assumed on the net of states $\omega\d$ inducing $(\H,J)$, it is sufficient to show that the implementing semigroup preserves $J$-convergence of the net of implementing unit vectors $\Omega\d$.

\begin{prop}[Convergence of implementors]
\label{prop:strictcimpmor}
Let $\TT\d(t)$ be a net of dynamical semigroups on $(\A,j)$ and $V\d(t)$ be an implementing net of semigroups on $(\H, J)$ with respect to a compatible net of representations $\pi\d$ induced by a state $\omega_{\oo}$.\\
Then, the implementing semigroup $V\d(t)$ is $JJ$-convergent if $\TT\d(t)$ is \jt\jt-convergent and $V\d(t)$ preserves the $J$-convergence of the net of unit vectors $\Omega\d$ implementing $\omega\d$.
\end{prop}

\begin{rem}
\label{rem:strictcimp}
Prop.~\ref{prop:strictcimp} directly extends to general dynamical semigroups $\TT\d(t)$ if each $\omega_{n}$ is $\TT_{n}(t)$-invariant and $V_{n}(t)$ is induced by $\omega_{n}$ as in Rem.~\ref{rem:impdyn}.
\end{rem}

\section{Interchanging Lie-Trotter limits and inductive limits}\label{sec:trotter}

\begin{prop}
    Let $T\d(t)$ and $S\d(t)$ be nets of dynamical semigroups with nets of generators $A\d$ and $B\d$ that converge in the sense of \cref{thm:evolution}.
    Assume that for any $n$, the Trotter product converges strongly to a dynamical semigroup $U_n(t)$, i.e.,
    \begin{equation}\label{eq:trotter_n}
        \lim_{k\to\oo} \norm{ ([T_n(t/k)S_n(t/k)]^k - U_n(t))x_n } =0\quad \forall x_n\in E_n
    \end{equation}
    uniformly on compact time intervals.
    Assume that $\D = \dom{A\d}\cap \dom{B\d}$ is seminorm-dense.
    Consider the following statements
    \begin{enumerate}[(a)]
        \item\label{it:convergent_trotter_lim}
            $U\d(t)$ is convergent in the sense of \cref{thm:evolution},
        \item\label{it:trotter_seminorm}
            For all \jt-convergent $x\d$ and all  $t\ge0$,
            \begin{equation}\label{eq:trotter_seminorm}
                \lim_{k\to\oo}\ \seminorm[\big]{ \paren[\big]{[T\d(t/k)S\d(t/k)]^k - U\d(t)}x\d } =0.
            \end{equation}
    \end{enumerate}
    Then (a) $\Leftarrow$ (b). If $(\lambda-A\d-B\d)\D$ is also dense, the converse also holds, i.e., (a) and (b) become equivalent, and it follows that
    \begin{equation}\label{eq:trotter_at_lim}
        \lim_{k\to\oo}\ \norm[\big]{\paren[\big]{[T_\oo(t/k)S_\oo(t/k)]^k - U_\oo(t)}x_\oo} =0\quad \forall x_\oo\in E_\oo
    \end{equation}
    with uniform convergence on compact time intervals.
\end{prop}

Note that the rather complicated looking condition \eqref{eq:trotter_seminorm} follows if for all $m$ and $x_m$ the Trotter product converges uniformly in $n$ in the sense that
\begin{equation*}
    \lim_{k\to\oo}\ \norm[\big]{([T_n(t/k)S_n(t/k)]^k - U_n(t))\j nm x_m}=0\quad\text{uniform in $n$.}
\end{equation*}

In the context of finite-dimensional approximations of Hilbert spaces (which are inductive systems), this idea was recently used to prove a theorem on the validity of finite-dimensional approximations (e.g., by numerics) of infinite-dimensional Trotter problems in \cite{burgarth2022state}.

\begin{proof}
    \ref{it:trotter_seminorm} $\Rightarrow$ \ref{it:convergent_trotter_lim}: Clearly $[T\d(t/k)S\d(t/k)]^k x\d$ is \jt-convergence preserving for any $k$.
    Let $x\d\in\C(E,j)$.
    The assumption implies that $U\d(t)x\d$ can be approximated by the \jt-convergent nets $[T\d(t/k)S\d(t/k)]^kx\d$ in seminorm. That $\C(E,j)$ is seminorm-closed implies that $U\d(t)x\d$ is also \jt-convergent.
    For the strong continuity of $U_\oo(t)$, we check item \hyperref[it:semigroups']{(1$'$)} of \cref{thm:evolution} to the space $\D$ (this item is introduced in the proof implies item \ref{it:semigroups}.
    Since the space $\set{x\d\in \C(E,j)\given x_n\in \dom{C_n},\, \norm{C\d x\d}_\nets<\oo}$ contains $\D$ item \hyperref[it:semigroups']{(1$'$)} indeed holds.

    We now assume that $(\lambda-A\d-B\d)\D$ is seminorm dense and prove
    \ref{it:convergent_trotter_lim} $\Rightarrow$ \ref{it:trotter_seminorm}:
    Since we know that $U\d(t)$ is \jt\jt-convergent, we can directly proof \eqref{eq:trotter_at_lim}.
    The local uniformity in $t$ guarantees that $C_n$ is an extension of $A_n+B_n$, where $C_n$ is the generator of $U_n(t)$ \cite[Thm.~3.7]{chernoff}.
    Therefore, we have for all $x\d\in\D$ that $C\d x\d =A\d x\d+B\d x\d\in\C(E,j)$ and hence $\D\subset \dom{C\d}$.
    It follows that $\D_\oo \subset \dom{A_\oo}\cap \dom{B_\oo}$ but equality need not hold, to the best of our knowledge.
    For an element $x_\oo\in\dom{A_\oo}\cap\dom{B_\oo}$ to be in $\D_\oo$ requires that there is one net $x\d$ converging to $x_\oo$ such that $A\d x\d$ and $B\d x\d$ are \jts-convergent simultaneously.

    We now prove \cref{eq:trotter_at_lim}:
    Since $U\d(t)$ satisfies the conditions of \cref{thm:evolution}, we know that the limit semigroup $U_\oo(t)$ is generated by an operator $C_\oo$ and our assumption implies that $\D_\oo$ is a core for $C_\oo$.
    This is because $(\lambda-C_\oo)\D_\oo = (\lambda-A_\oo-B_\oo)\D_\oo$ which is dense.
    We can now apply the standard Trotter-Chernoff Theorem \cite[Ch.~II, Thm.~5.8]{engelnagel}, which shows that \eqref{eq:trotter_at_lim} holds.
\end{proof}

\printbibliography

\end{document}

%% file: base.tex
\newtheorem{thm}{Theorem}
\newtheorem*{thm*}{Theorem}
\newtheorem{letterthm}{Theorem} 

\newtheorem{lem}[thm]{Lemma}
\newtheorem{cor}[thm]{Corollary}
\newtheorem{prop}[thm]{Proposition}
\newtheorem*{prob*}{Problem}
\newtheorem{defin}[thm]{Definition}

\theoremstyle{definition}
\newtheorem{rem}[thm]{Remark}
\newtheorem{exa}[thm]{Example}
\theoremstyle{remark}
\newtheorem*{exa*}{Example}
\newtheorem*{rem*}{Remark}

\newcommand*{\1}{\text{\usefont{U}{bbold}{m}{n}1}}
\renewcommand{\bar}[1]{\overline{#1}}
\newcommand\CC{\mathbb C}
\newcommand\D{\mathcal D}
\newcommand\eps\varepsilon
\renewcommand\Im{\operatorname{Im}}
\renewcommand\H{\mathcal H}

\newcommand\NN{\mathbb N}
\def\placeholder{\,\cdot\,}

\newcommand\RR{\mathbb R}

\newcommand\restrictedto\upharpoonright

\renewcommand\Re{\operatorname{Re}}
\newcommand\oo\infty
\newcommand\ox\otimes
\newcommand\ZZ{\mathbb Z}

\DeclareMathOperator{\Ran}{Ran}
\NewDocumentCommand\TC{o}{{\IfNoValueTF{#1}{\mathcal{T}(\mathcal{H})}{\mathcal{T}(#1)}}}
\newcommand\boundedoperatorsymbol{{\mathcal B}}
\NewDocumentCommand\BO{mo}{ \IfNoValueTF{#2}{\boundedoperatorsymbol(#1)}{\boundedoperatorsymbol(#1,#2)}}

\def\up#1{^{(#1)}}

\let\mc\mathcal

\newcommand\mat[1]{ \begin{pmatrix} #1 \end{pmatrix} }

\renewcommand\limsup\varlimsup
\renewcommand\liminf\varliminf
\DeclareMathOperator\id{id}

\DeclareMathOperator{\Sp}{Sp}

\newcommand*\quotient[2]{{^{\textstyle #1}\big/_{\textstyle #2}}}

\let\lim\relax
\NewDocumentCommand\lim{o}{\IfNoValueTF{#1}{\mathop{\textup{lim}}}{{#1}\!-\!\mathop{\textup{lim}}}}

\DeclareFontFamily{U}{matha}{\hyphenchar\font45}
\DeclareFontShape{U}{matha}{m}{n}{ <-6> matha5 <6-7> matha6 <7-8> matha7 <8-9> matha8 <9-10> matha9 <10-12> matha10 <12-> matha12 }{}
\DeclareSymbolFont{matha}{U}{matha}{m}{n}
\DeclareFontFamily{U}{mathx}{\hyphenchar\font45}
\DeclareFontShape{U}{mathx}{m}{n}{ <-6> mathx5 <6-7> mathx6 <7-8> mathx7 <8-9> mathx8 <9-10> mathx9 <10-12> mathx10 <12-> mathx12 }{}
\DeclareSymbolFont{mathx}{U}{mathx}{m}{n}
\DeclareMathDelimiter{\vvvert} {0}{matha}{"7E}{mathx}{"17}%
\DeclarePairedDelimiterX{\normiii}[1]{\vvvert}{\vvvert} {\ifblank{#1}{\:\cdot\:}{#1}}

\DeclarePairedDelimiter\paren\lparen\rparen
\DeclarePairedDelimiter\bracks\lbrack\rbrack
\DeclarePairedDelimiterX\norm[1]\lVert\rVert{\ifblank{#1}{\placeholder}{#1}}
\DeclarePairedDelimiter\abs\lvert\rvert
\DeclarePairedDelimiter\braces\lbrace\rbrace
\DeclarePairedDelimiterX\ip[2]{\langle}{\rangle}{#1 , #2}

\newcommand\qandq{\quad\text{and}\quad}

\DeclarePairedDelimiter\ket\vert\rangle
\DeclarePairedDelimiter\bra\langle\vert
\DeclarePairedDelimiterX\braket[2]\langle\rangle{ #1 \delimsize\vert #2}
\DeclarePairedDelimiterX\brAAket[3]\langle\rangle{ #1 \delimsize\vert #2 \delimsize\vert #3 }
\DeclarePairedDelimiterX\ketbra[2]\vert\vert{ #1 \delimsize\rangle\delimsize\langle #2 }
\DeclarePairedDelimiterX\kettbra[1]\vert\vert{ #1 \delimsize\rangle\delimsize\langle #1 }

\providecommand\given{}  
\newcommand{\SetSymbol}[1][]{ \nonscript\ #1\vert \allowbreak \nonscript\ \mathopen{} } 
\DeclarePairedDelimiterX{\set}[1]\{\}{ \renewcommand\given{\SetSymbol[\delimsize]} #1 }

\newcommand\case[1]{{\begin{cases} #1 \end{cases}}}

\newcommand\hide[1]{}

%% file: mf.tex
The mean field limit is a kind of thermodynamic limit in which the limit parameter is the size of the system. The size is just the number of systems, which are all  of the same kind. The main theme is permutation symmetry of the systems, and we are especially interested in ``intensive'' observables, which can be understood as functions of averages over all sites. We follow the approach of \cite{raggio1989quantum, Wer92}, for which early versions of soft inductive limits were originally developed, initially as a tool to take the limit of partition functions and equilibrium states. However, this also worked well for dynamics, and prototypes of  Theorem~\ref{thm:evolution} are found in \cite{duffield1992mean, DW92}.

When $\AA$ denotes the observable algebra of a single system, the $N$-fold minimal C*-tensor product $\AA^{\otimes N}$ describes the system of size $N\in\Nl$. (See \cite{inflate} for a discussion of maximal products in this context). The permutations of the $N$ sites act on $\AA^{\otimes N}$, and we denote by $\sym_N$ the average over the $N!$ permutations. The permutation invariant observables are then $\AA_N=\sym_N\bigl(\AA^{\otimes N}\bigr)$. Observables at different system sizes $N>M$ are connected by
\begin{equation}\label{j-mf}
  \j NM(A)=\sym_N\bigl(A\otimes\1^{\otimes(N-M)}\bigr).
\end{equation}
This is a soft C*-inductive limit in the sense of the previous section. To show this, one considers the basic sequences: In a basic sequences $\j NM A$ with $A\in\AA_M$, we have an average over all embeddings of $M$ sites into the large set of $N$ sites. In the product of two such observables for large $N$, the sites of these averages do not overlap in leading order. This not only establishes the product property \eqref{eq:asymmor}, but at the same time shows that the product is abelian. So
abstractly, we know that $\AA_\infty\cong C(\Sigma)$, for some compact space $\Sigma$, the Gelfand spectrum of $\AA_\infty$, and the points of $\sigma$ correspond to the multiplicative states on $\AA_\infty$. There is a direct way to generate these from any state $\sigma$ on $\AA$, as the weak limit of the homogeneous tensor product states $\sigma^{\otimes N}$. Indeed, $\sigma^{\otimes N}(\j NMA)=\sigma^{\otimes M}(A)$ if $N>M$, so the expectations $\sigma^{\otimes N}(A_N)$ from a Cauchy sequence
for \jt-convergent $A\d$, and we define the function $A_\infty:\Sigma\to\Cx$ by
\begin{equation}\label{mfsig}
  A_\infty(\sigma)=\lim_N\sigma^{\otimes N}(A_N).
\end{equation}
Moreover, from the combinatorics of overlaps, such states are multiplicative, and conversely, all multiplicative states are determined by their one-particle restriction $\sigma$. We conclude that $\Sigma$ is the state space of the one-particle algebra $\AA$ with its weak* topology,

Similarly, any permutation invariant state on $\AA^{\otimes\infty}$ (considered as an inductive limit algebra, see Sect.~\ref{sec:quasi-local}) defines a convergent family of expectations, and hence a state on $\AA_\infty$. Such permutation invariant states (in the classical case) were called {\it exchangeable} by de~Finetti. It is immediate from the above considerations that such states can be written as an integral over pure states, i.e., homogenous product states. This is known as the de~Finetti Theorem, which was first proved in the quantum case by St\o rmer \cite{Stoermer}. The simple characterization of the pure states is what makes mean field theories completely solvable. This becomes clear when one replaces the permutation average with a translation average, with a view to translation invariant lattice interactions. Then the product theorem fails, although there is sufficient asymptotic abelianness to make the translation invariant states, i.e., the states on the limit space analogous to $\AA_\infty$, a simplex \cite{bratteli1}. But since this does not arise as the state space of a C*-algebra, the extremal states, known in that case as ergodic states, cannot be characterized as multiplicative and have no simple parametrization.

Dynamically, it is natural to consider first Hamiltonian dynamics on $\AA_N$, where the Hamiltonian densities $H_N/N$ are a \jt-convergent sequence \cite{DW92b}. One needs a bit of extra regularity, satisfied, e.g., by basic sequences. For such a mean field interaction, Theorem~\ref{thm:evolution} gives classical Hamiltonian dynamics on the state space $\Sigma$. More precisely, there is a Poisson bracket for functions on $\Sigma$, arising as the limit of commutators
$\{A_\infty,B_\infty\}:=\jlim_N iN[A_N,B_N]$ for suitably regular \jt-convergent sequences $A\d$ and $B\d$.
Of course, the commutator $[A_N, B_N]$ itself vanishes in the limit, but the scaling by $N$ picks out the leading order of overlaps (single site overlaps). The limit dynamics is then generated by the Hamiltonian function $H_\infty=\jlim_N(H_N/N)$. It is clear from dimensional considerations that this Poisson bracket does not arise from a symplectic form but from a degenerate antisymmetric form. It has the property that, for any Hamiltonian, the non-linear flow generated on the one-particle state space $\Sigma$ respects the unitary equivalence of states, i.e., leaves the spectrum of the density operator invariant. For $2$-level systems (``qubits''), this means that the dynamics respects the foliation of the Bloch ball into concentric spheres.

This setting generalizes easily to the {\it inhomogeneous} case, in which the evolution depends on additional random variables that are associated with the sites and have a limiting distribution. For the equilibrium case, this extension covers the BCS model \cite{RW91}, and the dynamics was worked out in \cite{DW92c} and is written in terms of integro-differential equations. Another variant considers {\it Bosonic systems}, i.e., the states for finite system size do not merely commute with the permutations but are even supported by the permutation invariant subspace of the $N$-particle Hilbert space. The theory then applies with the sole modification that the one-particle state space $\Sigma$ is replaced by the set of pure states only \cite{Wer92}. The salient de~Finetti Theorem was noted before by Hudson and Moody \cite{HudsonMoody}.

More interesting behavior is seen when the finite system dynamics is allowed to be dissipative \cite{DW92}, i.e., given by a semigroup of completely positive maps. It is then interesting to consider not just the mean field dynamics as defined above, which we call the bulk evolution of the mean field system, but also the dynamics of {\it tagged particles}. This is easily incorporated \cite{DW92} by modifying the inductive system leaving some set of sites out of the symmetrization \eqref{j-mf}. The number $M$ of tagged sites can increase with $N$, but $M/N$ should go to zero. With constant $M$, the resulting limit algebra is then
\begin{equation}\label{mftaggedA}
  \AA_\infty=\CC(\Sigma)\otimes\AA^{\otimes M}\cong\CC\bigl(\Sigma;\AA^{\otimes M}\bigr)
\end{equation}
where the second form denotes the algebra of functions on the one-particle state space $\Sigma$, taking values in $\AA^{\otimes M}$. There is nothing new to show for this limit: It is just the tensor product, in the sense of Prop.~\ref{prop:tensor}, of the mean-field described so far, tensored with the identity on $\AA^{\otimes M}$. We can also let $M{\to}\infty$ here,  so the second factor would be the inductive limit for the quasi-local algebra as discussed in Sect.~\ref{sec:quasi-local} with the limit space symbolically denoted $\AA^{\otimes\infty}$.

Of course, we could discuss systems where the tagged particles play a different dynamical role from the bulk particles. However, we want to use the tagging distinction just to get a more detailed view of the mean-field dynamics. The evolutions for different $M$ are then a consistent family of semigroups: The dynamics for $M$ tagged sites reduces to the one for $M'<M$ tagged sites on the observables which have $\idty$ on the $M-M'$ sites. In particular, the bulk dynamics corresponds to $M=0$. The local dynamics is then an evolution on $\AA^{\otimes M}$, which depends on the classical state $x$ like an external parameter.

The local dynamics may generate correlations between the tagged particles. A prototype for this are squared Hamiltonian generators. This example follows the general principle \cite[Thm.~2.31]{DavieSquare} that the square of the generator of a one-parameter group of isometries generates a contraction semigroup, which is described by applying the isometry group at an evolution parameter determined by a Wiener diffusion process. One then readily checks that the evolution arising from such a selection process does not factorize over tagged sites. A typical feature of such evolutions is that the operator norm of the generator grows like $N^2$.

But there are also many evolutions for which the norm grows only like $N$, e.g., those satisfying a condition similar to the mean field condition for Hamiltonians: Denoting the generator on the finite system $\AA^{\otimes N}$ by $G_N$, we can ask that $G_N/N$ arises by permutation averaging from $G_R/R$ for some $R$, where, however, we take the action of permutations on observables rather than on Hilbert spaces. Moreover, we extend this average equally over tagged and untagged sites.
In that case, analyzed under the term ``bounded polynomial generator'' in \cite[Prop.~3.6]{DW92}, the bulk evolution is a flow $\sigma\mapsto \flow_t\sigma$ on $\Sigma$, and the local dynamics is given by completely positive maps $\Lambda_t^\sigma$, so that for $A\in\AA_\infty$, i.e., a continuous function $A:\Sigma\to\AA^{\otimes M}$,
\begin{equation}\label{mflocal}
  \bigl(\TT_\infty(t)A\bigr)(\sigma)=(\Lambda_t^\sigma)^{\otimes M}A(\flow_t\sigma).
\end{equation}
Here $\Lambda$ satisfies the cocyle equation $\Lambda_{s+t}^\sigma=\Lambda_s^\sigma\ \Lambda_t^{\flow_s\sigma}$ and the consistency condition $\sigma(\Lambda_t^\sigma a)=(\flow_t\sigma)(a)$. So $\Lambda_t^\sigma$ is a Lindblad evolution, whose generator depends on time via $\flow_t\sigma$. Since the local evolution is a tensor product, no correlations are generated between tagged sited in this class of evolutions.

For Hamiltonian systems, the local dynamics is generated by a Hamiltonian $dH_\infty(\sigma)\in\AA$, which is the gradient of $H_\infty$ at $\sigma$, a summary of all directional derivatives. Explicitly, for any $\rho\in\Sigma$,
\begin{equation}\label{mfgradient}
   \rho\bigl(dH_\infty(\sigma)\bigr)=\left.\frac d{dt}H_\infty(t\rho+(1-t)\sigma)\right\vert_{t=0}.
\end{equation}
Note that, by construction, an additive constant in $dH$ is fixed so that $\sigma(dH(\sigma))=0$.  Local dynamics is thus given by the time-dependent Hamiltonian $dH(\flow_t\sigma)$, i.e., driven by the bulk flow.

Even within the class of bounded polynomial generators, however, there is a remarkable variety in behavior: The bulk evolution may or may not be Hamiltonian in terms of the Poisson structure described above, and independently, the local evolution may or may not be generated by a state-dependent Hamiltonian. An interesting subclass, in which the local dynamics is automatically Hamiltonian, is given when the Lindblad jump operators are themselves \jt-convergent, say, in the Heisenberg picture,
\begin{equation}\label{mfsymK}
  G_N(X) = i[H_N,X]+ N\sum_\alpha V_{\alpha,N}^* [X, V_{\alpha,N}]+ [V_{\alpha,N}^* ,X] V_{\alpha,N}
\end{equation}
with $V_{\alpha,N}=\sym_N(V_\alpha,R)$, for some $R$. This is $N$ times a double average, in which the $R$ sites are permuted independently to $N$ sites. The dominant contribution thus comes from terms where these sets of sites do not overlap and can hence be realized on $2R$ sites, which are then averaged into $N$ in the operator sense. It turns out that the local dynamics is Hamiltonian with a dependence on the bulk state given by
\begin{equation}\label{mfsymKH}
  H(\sigma)= dH_\infty(\sigma)+ \sum_\alpha \Im\Bigl(\overline{V_{\alpha,\infty}(\sigma)} dV_{\alpha,\infty}(\sigma)\Bigr).
\end{equation}
The bulk flow $\sigma_t=\flow_t\sigma$ is then given by the differential equation $\dot\sigma_t(A)=\sigma_t\bigl(i[H(\sigma_t),A]\bigr)$, and thus still leaves the spectrum of $\sigma$ invariant. However, since $V_{\alpha,\infty}$ and its complex conjugate may have linearly independent derivatives, $H(\sigma)$ is no longer a gradient. Hence, a Hamiltonian function does not generate the evolution via the Poisson bracket. In fact, {\it any} spectrum preserving ordinary differential equation for $\sigma_t$ can be approximately realized in this manner \cite{DW92}.

As a counterpoint, consider a case in which the operators $V_{\alpha,N}$ and $V_{\alpha,N}^*$ are {\it not separately} averaged over permutations. In this case, the local state tends to approach the bulk state, in the sense that the relative entropy $S\bigl((\Lambda^\sigma_t)^*\rho,\flow_t\sigma\bigr)$, which is always non-increasing, even goes to zero. The simplest example is a term in the generator, which shuffles the sites, and thus makes the local state approach the bulk. We take  $\Gamma_2(X)=F[X,F]+[F,X]F=2(FXF-X)$, where $F$ is the permutation operator on two sites. This symmetrizes for larger $N$ to
\begin{equation}\label{mfflipG}
  \Gamma_N(X)=\frac1{N-1}\sum_{ij}(F_{ij}XF_{ij}-X),
\end{equation}
where $F_{ij}$ is the permutation of sites $i$ and $j$. Here we count every pair twice and allow the zero contributions from $i=j$. Let us apply this to a basic sequence with tagged site $1$, of which a prototype is $X_N=A\otimes\sym_{N-1}(B_k\otimes\idty^{\otimes(N_k-1)}\bigr)$. The limit of $\Gamma_NX_N$ is an $\AA$-valued function $(\Gamma X)_\infty$ of the one-site state $\sigma$ obtained as
\begin{equation}\label{mfflipGi}
  \rho\bigl((\Gamma X)_\infty(\sigma)\bigr)=\lim_N\rho\otimes\sigma^{\otimes(N-1)}\bigl(\Gamma_N(X)\bigr)
           =2(\rho(\idty)\sigma-\rho)\bigl(X_\infty(\sigma)\bigr).
\end{equation}
On bulk observables, for which the symmetrization is over all sites, $\Gamma_N$ vanishes by \eqref{mfflipG}, and as an additional term in a polynomial generator, it does not change the bulk flow but modifies the local dynamics $\Lambda_t^\sigma$ to a Trotter limit of the original one, interlaced with exponential contraction to the bulk state, i.e.,  $\tilde\Lambda_t^\sigma\rho=\sigma+ e^{-2t}(\rho-\sigma)$.